\pdfoutput=1
\documentclass[11pt]{article}

\usepackage{authblk}
\usepackage{breakcites}
\usepackage{amsmath,amsfonts,amssymb,bbm,amsthm}
\usepackage{mathtools}
\usepackage{graphicx}
\usepackage{framed}
\newcommand{\Infer}{\Longrightarrow}
\newcommand{\ud}{\mathrm{d}}
\newcommand{\ab}{\allowbreak}
\newcommand{\Reduce}{\Longleftarrow}

\DeclareMathOperator*{\E}{E}

\newcommand{\dV}{\mathcal{V}}
\newcommand{\dX}{\mathcal{X}}
\newcommand{\Val}{\mathrm{Val}}
\newcommand{\op}{\overline{\phi}}
\newcommand{\up}{\underline{\phi}}
\newcommand{\sgm}{\sigma}
\newcommand{\sgp}{\sigma}
\newcommand{\osg}{{\overline{\sigma}}}
\newcommand{\ovp}{\overline{\varphi}}
\newcommand{\uvp}{\underline{\varphi}}
\newcommand{\Rev}{\textsc{Rev}}
\newcommand{\Utl}{\textsc{Utl}}
\newcommand{\Eff}{\textsc{Eff}}
\newcommand{\Real}{\mathbb{R}}
\newcommand{\One}{\mathbf{1}}
\newcommand{\Zero}{\mathbf{0}}
\newcommand{\frd}{\mathfrak{d}}
\newcommand{\frs}{\mathfrak{s}}
\newcommand{\frT}{\mathfrak{T}}
\newcommand{\Prog}{\mathsf{Prog}}
\newcommand{\val}{\mathsf{value}}

\newcommand{\F}{\mathcal{F}}

\newcommand{\I}{\mathcal{I}}
\newcommand{\dH}{\mathcal{H}}

\newcommand{\bal}{\mathsf{bal}}

\newcommand{\hu}{\hat{u}}
\newcommand{\bcdot}{\boldsymbol{\cdot}}
\newcommand{\SM}{\mathsf{S}}
\newcommand{\PM}{\mathsf{P}}
\newcommand{\SMA}{M^\SM}

\newcommand{\OPT}{\mathsf{OPT}}

\newtheorem{theorem}{Theorem}[section]
\newtheorem{definition}[theorem]{Definition}

\newtheorem{lemma}[theorem]{Lemma}
\newtheorem{corollary}[theorem]{Corollary}
\newtheorem{mechanism}{Mechanism}
\newtheorem{observation}[theorem]{Observation}
\newtheorem{example}[theorem]{Example}

\title{Optimal dynamic mechanisms with ex-post IR via bank
       accounts\thanks{Contacts: \texttt{\{mirrokni, renatoppl\}@google.com},
       \texttt{\{kenshinping, songzuo.z\}@gmail.com}.}}

\author[$\dag$]{Vahab Mirrokni}
\author[$\dag$]{Renato Paes Leme}
\author[$\ddag$]{Pingzhong Tang}
\author[$\ddag$]{Song Zuo}

\affil[$\dag$]{Google Research}
\affil[$\ddag$]{Institute for Interdisciplinary Information Sciences,
                Tsinghua University}

\begin{document}

\maketitle

% Page heads
%\markboth{G. Zhou et al.}{A Multifrequency MAC Specially Designed for WSN Applications}

% Title portion
% NOTE! Affiliations placed here should be for the institution where the
%       BULK of the research was done. If the author has gone to a new
%       institution, before publication, the (above) affiliation should NOT be changed.
%       The authors 'current' address may be given in the "Author's addresses:" block (below).
%       So for example, Mr. Abdelzaher, the bulk of the research was done at UIUC, and he is
%       currently affiliated with NASA.

\begin{abstract}
Lately, the problem of designing multi-stage dynamic mechanisms has been shown
to be both theoretically challenging and practically important. In this paper,
we consider the problem of designing revenue optimal dynamic mechanism for a
setting where an auctioneer sells a set of items to a buyer in multiple stages.
At each stage, there could be multiple items for sale but each item can only
appear in one stage. The type of the buyer at each stage is thus a
multi-dimensional vector characterizing the buyer's valuations of the items at
that stage and is assumed to be stage-wise independent. Compared with the
single-shot multi-dimensional revenue maximization, designing optimal dynamic
mechanisms for this setting brings additional difficulties in that mechanisms
can now branch on different realizations of previous random events and secure
extra revenue.

In particular, we propose a novel class of mechanisms called {\em bank account}
mechanisms. Roughly, a bank account mechanism is no different from any
stage-wise individual mechanism except for an augmented structure called bank
account, a real number for each node that summarizes the history so far. This
structure allows the auctioneer to store deficits or surpluses of buyer utility
from each individual mechanism and even them out on the long run. While such
simplicity allows for easy implementation in practice, it turns out it is also
sufficient to guarantee optimal revenue. We first establish that the optimal
revenue from any dynamic mechanism in this setting can be achieved by a
bank account mechanism, and we provide a simple characterization of the set of
incentive compatible and ex-post individually rational bank account mechanisms.
Based on these characterizations, we then investigate the problem of finding the
(approximately) optimal bank account mechanisms. We prove that there exists a
simple, randomized bank account mechanism that approximates optimal revenue up to a
constant factor. Our result is general and can accommodate previous
approximation results in single-shot multi-dimensional mechanism design. Based
on the previous mechanism, we further show that there exists a deterministic
bank account mechanism that achieves constant-factor approximation as well. Finally, we
consider the problem of computing optimal mechanisms when the type space is
discrete and provide an FPTAS via linear and dynamic programming.
\end{abstract}

\newpage

%\terms{Design, Algorithms, Performance}

% \keywords{Dynamic mechanism design, bank account mechanism,
% revenue maximization}

% \begin{bottomstuff}
% This work was supported by the National Basic Research Program of China Grant
% 2011CBA00300, 2011CBA00301, the Natural Science Foundation of China Grant
% 61033001, 61361136003, 61303077, 61561146398, a Tsinghua Initiative Scientific
% Research Grant and a China Youth 1000-talent program.
%
% Author's addresses: Song Zuo and Pingzhong Tang, IIIS, Tsinghua University;
% email: \{zoy.blood, kenshinping\}@gmail.com; Renato Paes Leme and Vahab
% Mirrokni, Google Research; email: \{renatoppl, mirrokni\}@google.com.
% \end{bottomstuff}

\section{Introduction}
  A wide range of internet advertising applications including sponsored search
and advertising exchanges employ a sequence of repeated auctions to sell their
inventory of page-views or impressions. In a variety of settings, optimizing
revenue for the seller (e.g., the search engine, publishers, or advertising
exchange) boils down to designing auctions that allocate the most effective ads
for users, and at the same time extract enough revenue and result in a stable
market without incurring too much cost to enter the market.
%from advertisers competing for these ad slots.

In order to achieve this goal, the seller should provide good incentives for
advertisers to reveal their true values. To cope with this goal, in many
settings, advertising exchanges and search engines enforce incentive
compatibility or truthfulness at the level of each individual auction, and as a
result, end up running repeated 2nd-price auctions or variants of repeated
Myerson optimal auctions~\cite{ostrovsky2011reserve}.

In these cases, we say that the seller runs a {\em static auction} optimizing
the social welfare and/or revenue of each individual auction. Advertisers, in
turn, declare their bids for a sequence of such static auctions and not for one
auction. On the other hand, more traditional online advertising campaigns employ
long-term contracts between the seller and buyers which can result in higher
revenue per impression for sellers.

One reason for lower per-impression revenue from the auction-based advertising
is that the resulting individual auctions may become too
thin~\cite{celis2014buy}, and therefore, they may lack enough competition
to incur higher prices. More traditional contract-based advertising takes care
of this problem by bundling their impressions together, and achieving higher
average revenue per impression (while keeping the surplus of advertisers
positive). Traditional contract-based advertising, on the other hand, does not
have several advantages of the auction-based advertising, e.g., they incur a
high cost for entering the market for smaller advertisers, and may not allow for
real-time fine-tune targeting. A desirable goal is to design a solution that
combines the advantages of auction-based and contract-based advertising. One way
to achieve this goal is {\em dynamic mechanisms} which allow for combining
several auctions, and optimize seller's revenue across all auctions.

Motivated by the above discussion, the area of dynamic mechanism design has
attracted a large body of research work in the last decade
\cite{cavallo2008efficiency,pai2008optimal,pavan2008dynamic,cavallo2009efficient,gershkov2009dynamic,bergemann2010dynamic,gershkov2010efficient,athey2013efficient,kakade2013optimal,papadimitriou2014complexity,pavan2014dynamic}
(see~\cite{bergemann2011dynamic} for a survey). In these settings, unlike
repeated auctions, the pricing and allocation of items in each auction may
depend on the history of the bids, prices, and allocations in the previous
auctions. While such dynamic auctions can produce more revenue for the sellers,
a couple of issues has slowed down their adoption in the market. Main issues
with a majority of results in dynamic mechanism design are two-fold:

{\bf \noindent Complexity:} The first issue is the high complexity and
impracticality of implementing such auctions: If the seller is allowed to use
the full power of dynamic mechanism design and the price and allocation of each
auction depend on the outcome of all previous auctions, the complexity of
computing the optimal mechanism or even describing a policy could be high. In
fact, it has been recently shown that in some settings computing the optimal
dynamic mechanism is computationally hard~\cite{papadimitriou2014complexity}.

{\bf \noindent Future commitment and lack of individual rationality.} The
second issue is that they may require big initial commitments from buyers and
sellers prior to the auction~\cite{kakade2013optimal}. Although there have
been recent movements in adopting such new type of contracts~\cite{mirrokni2015deals},
this issue makes it much harder for buyers and sellers to employ these
mechanisms. In order to get around this issue, it would be desirable to design
dynamic mechanisms that satisfy a strong notion of individual rationality,
i.e., ex-post individual rationality.

In this paper, we aim to design simple dynamic mechanisms that achieve ex-post
individual rationality and at the same time, approximate the optimal revenue of
the full-force dynamic mechanism. Toward this goal, we study a family of dynamic
mechanisms by applying the idea of {\em bank accounts}, in which single-shot
auctions are augmented with a structure we call {\em bank account}, a real
number for each node that summarizes the history so far. This structure allows
the seller to store deficits or surpluses of buyer utility from each individual
auction and even them out on the long run. We prove a number of interesting
properties for these simple mechanisms.

In this first part of the paper (Section \ref{sec:bam}), we prove a number of
characterization results. These results ensure that, to maximize revenue in
dynamic mechanisms, one only needs to restrict attention to a subset of
bank account mechanisms.

Theorem \ref{thm:rep} states that for any dynamic mechanism, $M$, there is a
constructive bank account mechanism (coined {\em core BAM}) that: (1) generates
the same expected utility as $M$, (2) generates weakly higher expected revenue
for the auctioneer than $M$, (3) and is deterministic if $M$ is deterministic.
Theorem \ref{thm:corebam} provides sufficient and necessary conditions for such
construction (Definition \ref{def:corebam}) to be an IC-IR bank account
mechanism.

Based on the characterization results, in the second part of the paper
(Sections \ref{sec:3app}-\ref{sec:fptas}), we investigate the problem of how to
find the optimal bank account mechanisms.

Theorem \ref{thm:3app} shows a simple bank account mechanism $B$ that guarantees
a constant approximation to the optimal dynamic mechanism. Then Theorem
\ref{thm:dbam} further derandomizes the construction to get a deterministic bank
account mechanism $B^\ud$ that achieves a constant approximation as well. We
remark that both of $B$ and $B^\ud$ are constructed based on a given
(approximately) optimal (deterministic) static mechanism.\footnote{The gap
between the optimal dynamic mechanism and static mechanism can be arbitrarily
large (see Example \ref{example:hard}).} Particularly, for the one item per
stage case, the optimal static mechanism is known and deterministic
\cite{myerson1981optimal}, and $B$ and $B^\ud$ are $3-$ and $5-$approximation of
the optimal revenue, respectively. Finally, for the problem of computing
optimal mechanisms when the type space is discrete, we provide an
FPTAS\footnote{Fully Polynomial Time Approximation Scheme.} in Theorem
\ref{thm:dp} via linear and dynamic programming.

In a parallel work \cite{mirrokni2016dynamic}, we consider the bank account
mechanisms under {\em interim IR} constraints. We also prove a characterization
result concerning sufficiency of bank account mechanism in that setting and
establish trade-offs between the maximum limit of bank account balance and
revenue.

Ashlagi et al.~\cite{ashlagi2016sequential} independently study the same
problem as ours except that, in their setting, they assume there is only one
item for sale at each stage. They independently propose a dynamic program to
compute the optimal mechanism, in the same spirit as ours in Section
\ref{sec:fptas} but different in implementation details. They also propose a
mechanism that $2$-approximates the optimal revenue for the one-item-per-stage
setting, in contrast to our $3$-approximation in that specific setting. The
approaches used to prove these results, however, are quite different. They
finally extend their dynamic program algorithm to compute the optimal mechanism
for the multi-bidder case and study the case with discounted infinite horizon.
We note that our framework can naturally generalize to the multi-bidder case
with one bank account for each of the bidders and we have performed stationary
analysis for the infinite horizon case in our {\em interim IR}
paper~\cite{mirrokni2016dynamic}.

% \paragraph{Related Work}
% The area of {\em dynamic mechanism design} has attracted a large body of
% research work in the last decade
% \cite{cavallo2008efficiency,pai2008optimal,pavan2008dynamic,cavallo2009efficient,gershkov2009dynamic,bergemann2010dynamic,gershkov2010efficient,athey2013efficient,kakade2013optimal,papadimitriou2014complexity,pavan2014dynamic}.
% We refer to \cite{bergemann2011dynamic} for a comprehensive survey on the
% topic and cite here a couple of representative papers:
% \cite{bergemann2010dynamic} study the problem of efficient mechanism design in
% a dynamic environment where agents receive private information over time, and
% generalize the idea of pivot mechanisms \cite{green1977characterization} to
% this dynamic setting. \cite{pavan2014dynamic} provide characterization of
% dynamic local and global incentive compatibility constraints, and use the
% characterization to design optimal dynamic mechanisms in Markov environments.
% However, the optimal mechanism is individually rational only when the initial
% signal being observed, while in this paper we consider ex-post individually
% rational. \cite{papadimitriou2014complexity} work on the discrete (except for
% the last stage) and correlated valuation distribution setting, and prove
% hardness results to compute deterministic optimal mechanism subject to ex-post
% individually rational constraint for the two-stage case. We remark that none of
% theses three papers consider the case that there can be multiple items for sale
% in each stage.

% Prelimiaries and Model
\section{Setting}\label{sec:setting}

  Consider the problem of selling a sequence of items to a buyer who has
  additive valuations for the items. At each stage, there can be any number of
  items for sale, but each item only appears in one stage. The order in which
  the items arrive are fixed and commonly known, and the items cannot be sold
  in stages after they have arrived. At each stage $t \in [T]$, a type
  (valuation) $v_t \in \dV_t = \Real^{k_t}_+$ of the buyer ($k_t$ is the number
  of items for sale at stage $t$) is privately drawn from a public distribution
  $\F_t$, i.e., $v_t \sim \F_t$. We assume that the prior distributions
  $\F_t$'s are independent stage-wise, while within each stage, the
  distribution can be correlated over different items. Moreover, we assume that
  all valuations have finite expectations ($\E[\One \cdot v_t] <
  +\infty$)\footnote{We use $\One$ to denote a vector in which all entries are
  $1$. Similarly for $\Zero$.}. This assumption is to rule out the cases where
  the buyer faces two different strategies that both lead to infinite
  utilities.

  According to the definition above, $v_t$ is, in general, a multidimensional
  vector. None of our proofs, except the proof of Theorem \ref{thm:dp}, makes
  use of the assumption that $v_t$ is single dimensional. For the proof of
  Theorem \ref{thm:dp}, it still extends to the case where $v_t$ is
  multidimensional, by using the technique from \cite{cai2012algorithmic}.

  We also remark that all of our results extend to the setting with discounted
  utilities over time.
  % {\edit All of our results apply to the case where $v_t$ is multidimensional
  % and the case where, at each stage, there are multiple items for sale.
  % Furthermore, items sold within each stage need not to be independent. None of
  % our proofs, except that of Theorem [], makes use of the fact that $v_t$ is
  % single dimensional. For the proof of Theorem [], it still extends to the
  % multidimensional and multiple item cases, however, the conclusion is that [].}

  The stage outcome at stage $t$ is specified by a pair $(x_t, p_t)$, where
  $x_t \in \dX_t = [0,1]^{k_t}$ denotes the allocation vector at stage $t$ and
  $p_t \in \Real$ denotes the stage payment. The buyer utility from this stage
  is
  \begin{align*}
    u_t(\bcdot; v_t) = x_t(\bcdot) \cdot v_t - p_t(\bcdot).
  \end{align*}

  The seller's objective is to design a {\em mechanism} that maximizes the
  total revenue $\sum_{\tau = 1}^T p_t$. For convenience, we use $v_{(\tau,
  \tau')} = v_\tau, \ldots, v_{\tau'}$ to denote the vector of buyer's types
  from stage $\tau$ to $\tau'$. Similarly, $x_{(\tau, \tau')} = x_\tau, \ldots,
  x_{\tau'}$, $p_{(\tau, \tau')} = p_\tau, \ldots, p_{\tau'}$. We use
  $\dV_{(\tau, \tau')} = \dV_\tau \times \cdots \times \dV_{\tau'}$ to denote
  the buyer's type spaces from stage $\tau$ to $\tau'$.

  \subsection{Dynamic Mechanisms}

    In this paper, we consider {\em direct mechanisms}, where the agent reports
    its private type to the mechanism in each stage.

    % [I am not sure about the correctness. In particular, the strategy here can
    % depend on realizations of previous random events. This feature is not
    % captured in the proof of revelation principle.

    % {\edit
    % In the dynamic setting, for any mechanism with action space $\mathcal{A}$,
    % a strategy is a mapping from any history of private types (including the
    % private type of current stage) to $\mathcal{A}$. So the mechanism is
    % essentially an (extensive form) Bayesian game, and the revelation principle
    % works in this Bayesian case.}]

    \begin{definition}[Direct Mechanism]\label{def:mech}
      A direct mechanism is a pair of allocation rule and payment rule. An
      allocation rule is a mapping $x_t : \dV_{(1, t)} \rightarrow \dX_t$ that
      specifies at any given stage the probability of allocation of an item to
      the agent conditioned on the history, i.e., the sequence of types
      reported so far. A payment rule $p_t : \dV_{(1, t)} \rightarrow \Real$
      specifies the payment the agent should make to the auctioneer at stage
      $t$, also conditioned on the sequence of types reported so far.
    \end{definition}
    We are interested in mechanism satisfying two properties: incentive
    compatibility (IC) and individual rationality (IR) constraints. In dynamic
    mechanism design, there are various flavors of those notions. The flavors
    we will be interested in this work are as follows. For IC, we will consider
    incentive compatible in perfect Bayesian equilibrium, which means that at
    any given history, truth-telling is the best response for current and
    future stages in expectation over the randomness of future types. For IR,
    we will be interested in Ex-post individually rationality, which means that
    for any type realization, the buyer's overall utility is non-negative.
    Notice that this is stronger than requiring individual rationality in
    expectation, where the buyer's utility is non-negative in expectation over
    his type.

    To make those notions formal, let a bidding strategy $b_{(t + 1, T)} =
    \langle b_t \rangle_{(t + 1, T)}$ be a family of (possibly randomized)
    maps $b_t : \dV_{(1, t)} \rightarrow \dV_t$ specifying which type to report
    at each stage given the history of types observed so far. For a fixed
    bidding strategy $b_{(t + 1, T)}$, let $U^{b_{(t + 1, T)}}_t(v_{(1, t)})$
    be the total utility that the buyer can obtain in expectation over his type
    from stages $t + 1$ to $T$ of the mechanism given that types $v_{(1, t)}$
    were reported in the stages $1$ to $t$ and the bidding strategy is employed
    from that point on. Formally:
    %
    % Any mechanism must satisfy the so-called incentive compatibility (IC) and
    % individual rationality (IR) constraints.
    % \begin{itemize}
    %   \item Incentive compatible in perfect Bayesian equilibrium (IC): at any
    %         history, truth-telling is the best response for current and future
    %         stages in expectation over randomness of future types.
    %   \item Ex-post individually rational (IR): for any type realization, the
    %         buyer's overall utility is non-negative.
    % \end{itemize}
    %
    % Formally, define the following notation for the expected buyer utility
    % after stage $t$, given the current history $v_{(1, t)}$ and some bidding
    % strategy $b_{(t + 1, T)}$, i.e.,
    \begin{align}
      U^{b_{(t + 1, T)}}_t(v_{(1, t)}) = \E_{v_{(t + 1, T)}}\bigg[
        \sum_{\tau = t + 1}^T u_{\tau}
          \big(v_{(1, t)}, b_{t + 1}(v_{(1, t + 1)}), \ldots,
               b_{\tau}(v_{(1, \tau)}); v_{\tau}\big)\bigg].
      \label{eq:efu}
    \end{align}
    We use $U_t$ without superscript to denote the expected utility yielded by
    truthful bidding. This term is often called the promised utility
    \cite{li2013dynamic}, since it is the utility that the mechanism promises
    that the buyer will obtain in expectation if he/she behaves truthfully.

    Now we can define IC and IR precisely. We say that the mechanism satisfies
    IC if at any given stage his present utility plus expected future utility
    is maximized by bidding truthfully in the current and future stages. In
    other words, for each time $t \in [T]$, every type history $v_{(1, T)} \in
    \dV_{(1, T)}$, any deviating report at this timestep $v'_t \in \dV_t$ and
    any bidding strategy for the following timesteps  $b_{(t + 1, T)}$, we
    have:
    \begin{align}\label{mech:defic}
      \textbf{IC:} \quad
        u_t(v_{(1, t)}; v_t) + U_t(v_{(1, t)})
          \geq u_t(v_{(1, t - 1)}, v'_t; v_t)
                + U^{b_{(t + 1, T)}}_t(v_{(1, t - 1)}, v'_t);
    \end{align}

    We say that a mechanism satisfies individual rationality if for \emph{every
    history of types} the utility is non-negative:
    \begin{align}\label{mech:defir}
      \textbf{IR:} \quad
        \sum_{\tau = 1}^T u_\tau(v_{(1, \tau)}; v_\tau) \geq 0.
    \end{align}

    It will be convenient to define a simplified notion of IC, called stage-wise
    IC, which we will show to be equivalent to the previous notion. We say that
    a mechanism is stage-wise IC if at any given stage, reporting truthfully at
    the current and future steps dominates deviating in the current step and
    reporting truthfully from that point on:
    \begin{align}\label{mech:ic}
      \textbf{Stage-wise IC:} \quad
        u_t(v_{(1, t)}; v_t) + U_t(v_{(1, t)})
          \geq u_t(v_{(1, t - 1)}, v'_t; v_t) + U_t(v_{(1, t - 1)}, v'_t);
    \end{align}

    \begin{lemma}\label{lem:swic}
      IC (\ref{mech:defic}) is equivalent to stage-wise IC (\ref{mech:ic}).
    \end{lemma}

    The advantage to work with stage-wise IC is that we can focus on the
    current stage and ignore deviations from this point on. Using backwards
    induction, we will argue that IC and stage-wise IC are equivalent.

    It is also convenient to define stage-wise IR, which means that at any
    given stage $t$, the utility is non-negative.
    \begin{align}\label{mech:ir}
      \textbf{Stage-wise IR:} \quad u_t(v_{(1, t)}; v_t) \geq 0.
    \end{align}

    \begin{lemma}\label{lem:swir}
      IR (\ref{mech:defir}) is implied by stage-wise IR (\ref{mech:ir}).
    \end{lemma}

    Clearly stage-wise IR implies IR, but not the other way round. What we will
    show in the next lemma, however, is that if there is a mechanism that is IC
    and IR, we can construct a mechanism that is IC and stage-wise IR with the
    same allocation rule and the same revenue. The idea is that if a mechanism
    is ex-post IR, one can anticipate the payments to earlier stages in such a
    way the buyer is never required to pay more than his actual value. An
    extreme case of this reduction is the one where the buyer pays his bid in
    every round except the last, where the buyer is paid the difference between
    all the bids charged during the mechanism and his payment in the original
    mechanism.

    \begin{lemma}\label{lem:swirconst}
      For any IR direct mechanism $M = \langle x_{(1, T)}, p_{(1, T)} \rangle$,
      there is another direct mechanism $M' = \langle x'_{(1, T)}, p'_{(1, T)}
      \rangle$, such that
      \begin{enumerate}
        \item \label{lem:swicir:p1}
              $\forall t \in [T],~x_t \equiv x'_t$,
        \item \label{lem:swicir:p2}
              $\forall v_{(1, T)} \in \dV_{(1, T)},~
               \sum_{\tau = 1}^T p_\tau(v_{(1, \tau)}) =
               \sum_{\tau = 1}^T p'_\tau(v_{(1, \tau)})$,
        \item \label{lem:swicir:p3}
              $M'$ is stage-wise IR,
        \item \label{lem:swicir:p4}
              and $M'$ is IC, if and only if $M$ is IC.
      \end{enumerate}
    \end{lemma}

    In light of Lemma \ref{lem:swic} to \ref{lem:swirconst}, we can use the
    stage-wise IC (\ref{mech:ic}) as the equivalent definition of IC, and
    assuming that a direct mechanism satisfies stage-wise IR constraint
    (\ref{mech:ir}) is without loss of generality.

    % One may wonder if any IR direct mechanism can be made stage-wise IR by
    % rearranging the payments for different stages. Unfortunately, the following
    % example shows that there is an essential difference between IR and
    % stage-wise IR.
    % {\edit
    % \begin{example}\label{exam:swir}
    %   TODO
    % \end{example}
    % }
    For any mechanism $M$, denote the overall expected revenue by $\Rev(M)$,
    overall expected buyer utility by $\Utl(M)$, and overall expected
    social welfare by $\Eff(M)$, i.e.,
    \begin{align*}
      \Rev(M) &= \E_{v_{(1, T)}}\bigg[
                    \sum_{\tau = 1}^T p_\tau(v_{(1, \tau)})\bigg],  \\
      \Utl(M) &= U_0 = \E_{v_{(1, T)}}\bigg[\sum_{\tau = 1}^T
                                      u_\tau(v_{(1, \tau)}; v_\tau)\bigg],  \\
      \Eff(M) &= \Utl(M) + \Rev(M).
    \end{align*}

    Moreover, define the conditional expected utility with partially realized
    types $v_{(t, t')}$,
    \begin{align*}
      \Utl(M | v_{(t, t')}) = \E_{V_{(1, T)}} \bigg[
                    \sum_{\tau = 1}^T u_\tau(V_{(1, \tau)}; V_\tau) \Big|
                      V_{(t, t')} = v_{(t, t')}\bigg].
    \end{align*}

    A mechanism is {\em deterministic}, if for each stage $t$, $x_t \in \{0,
    1\}^{k_t}$. A mechanism is {\em history-independent}, if both $x_t$ and
    $p_t$ only depend on $v_t$, but not $v_{(1, t - 1)}$.

\section{Bank account mechanisms}\label{sec:bam}

  Designing dynamic mechanism is a hard task, mostly because the design space
  is so large. To remedy this situation, we propose a class of mechanisms,
  called {\em bank account mechanisms}, that are simpler and can be shown to
  contain the revenue optimal mechanism. Moreover, they are based on a simple
  principle that is easy to reason about: the mechanism keeps a bank balance
  for the buyer and the allocation rule depend on the bank balance and not on
  the full history of the mechanism.

  % Intuitively, the balance can be regarded
  % as an explicit way to realize per-stage charge $U_t$ --- the change of
  % balance of each stage specifies two things: (i) the amount of utility to
  % compensate previously promised utilities; and (ii) the amount of newly
  % promised utility starting from this stage.
  %
  % {\edit {\bf Comments:} @Renato: ``per-stage charge $U_t$'' is not introduced
  % until the paragraph after equation (\ref{eq:envic}).}

  Bank account mechanisms were introduced in \cite{mirrokni2016dynamic}. Here
  we show that this framework can be used to reduce the problem of designing
  dynamic mechanisms to the problem of designing a family of parametrized
  single-shot mechanism. To understand the spirit of this reduction, we bring
  the reader attention to the fact that the stage-wise IC constraint looks very
  much like the IC constraint in single-shot mechanism design for period $t$,
  when the quantity $p_t(v_{(1, t)}) - U_t(v_{(1, t)})$ is used instead of
  payments. To see that, rewrite the IC constraint (\ref{mech:ic}) as follows:
  \begin{align*}
    x_t(v_{(1, t)}) \cdot v_{t} - \big(p_t(v_{(1, t)}) - U_t(v_{(1, t)})\big)
    \geq x_t(v_{(1, t - 1)}, v'_t) \cdot v_{t} - \big(p_t(v_{(1, t - 1)}, v'_t)
         - U_t(v_{(1, t - 1)}, v'_t)\big),
  \end{align*}
  which can be considered as the IC condition for some one-shot stage
  mechanism with allocation rule $\hat{x}_{t \vert v_{(1, t - 1)}}(\bcdot) =
  x_t(v_{(1, t - 1)}, \bcdot)$ and payment rule $\hat{p}_{t \vert
  v_{(1, t - 1)}} (\bcdot) := \big(p_t(v_{(1, t - 1)}, \bcdot) -
  U_t(v_{(1, t - 1)}, \bcdot)\big)$.

%Applying the Envelope theorem to this one-shot mechanism, the allocation rule is the sub-gradient of the utility,
%  \begin{align}\label{eq:envic}
%    \frac{\partial \big(u_t(v_{(1, t)}; v_t) + U_t(v_{(1, t)})\big)}
%         {\partial v_t}
%    = \frac{\partial \Big(x_t(v_{(1, t)})v_{t} -
%                          \big(p_t(v_{(1, t)}) - U_t(v_{(1, t)})\big)\Big)}
%           {\partial v_t}
%    = x_t(v_{(1, t)}).
%  \end{align}

  Furthermore, rewrite the stage payment of a dynamic mechanism as
  $p_t(v_{(1, t)}) = \hat{p}_{t \vert v_{(1, t - 1)}}(v_t) + U_t(v_{(1, t)})$.
  This suggests us that the stage payment of any dynamic mechanism can be
  decomposed into the payment $\hat{p}_{t \vert v_{(1, t - 1)}}(v_t)$ of some
  one-shot stage mechanism and a dynamic component $U_t(v_{(1, t)})$, i.e., the
  truthful utility from stages $t + 1$ to $T$.

  As a result, it is without loss of generality to think of any dynamic
  mechanism as a sequence of per-stage IC mechanisms, described by allocation
  rule $\hat{x}_{t \vert v_{(1, t - 1)}}(\bcdot)$ and payment rule
  $\hat{p}_{t \vert v_{(1, t - 1)}}(\bcdot)$, as well as an additional
  per-stage charge $U_t(v_{(1, t)})$ (a temporary utility deficit), which will
  eventually be compensated via future per-stage IC mechanisms.

  By the incentive compatibility of the per-stage mechanism, according to the
  Envelope theorem \cite{rochet1985taxation}, the per-stage allocation rule is
  the sub-gradient of the per-stage utility,
  \begin{align}\label{eq:envic}
    \frac{\partial \big(u_t(v_{(1, t)}; v_t) + U_t(v_{(1, t)})\big)}
         {\partial v_t}
    = \frac{\partial \big(\hat{x}_{t \vert v_{(1, t - 1)}}(v_t) \cdot v_t -
                          \hat{p}_{t \vert v_{(1, t - 1)}}(v_t)\big)}
           {\partial v_t}
    = x_t(v_{(1, t)}).
  \end{align}

  So, given the allocation rule of the per-stage mechanisms, according to our
  reinterpretation, the only flexibility of the seller is to design the
  per-stage charge $U_t(v_{(1, t)})$, the truthful utility for the buyer in
  future stages.

  We are now ready to formally define a bank account mechanism (BAM). Imagine
  that the buyer has a bank account in his name. In each stage, before the
  buyer reports his type, the mechanism moves a certain amount of money from
  the account to the seller (we call it spend), then the buyer reveals his
  type and a single-shot IC mechanism is run. In the end of the mechanism,
  depending on the buyer's reported type, a certain amount of money is
  deposited in the bank account. The amount deposited still belongs to the
  buyer, so it doesn't correspond to a payment at the current step. The right
  way to think about a deposit is an amount of money that the buyer is setting
  aside so that the mechanism can spend.

  Intuitively, the balance can be regarded as an explicit way to realize
  per-stage charge $U_t$ --- the change of balance of each stage specifies two
  things: (i) the amount of utility to compensate previously promised utilities;
  and (ii) the amount of newly promised utility starting from this stage.

  First we define a generic bank account mechanism and then we discuss which
  properties it needs to have to satisfy IC and IR.

  \begin{figure}
    \centering
    \includegraphics[height=0.2\textwidth]{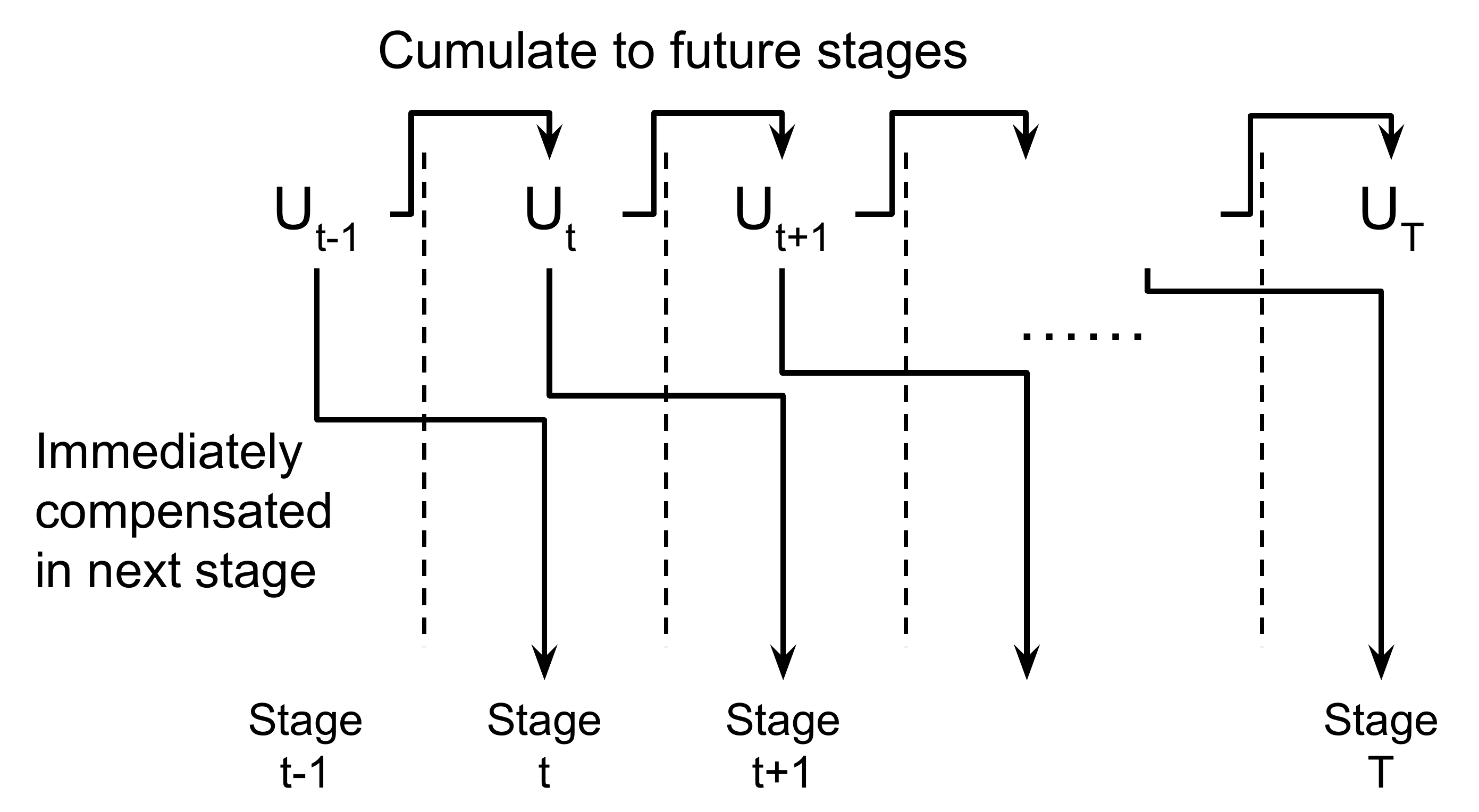}%
    ~~~~~~~%
    \includegraphics[height=0.2\textwidth]{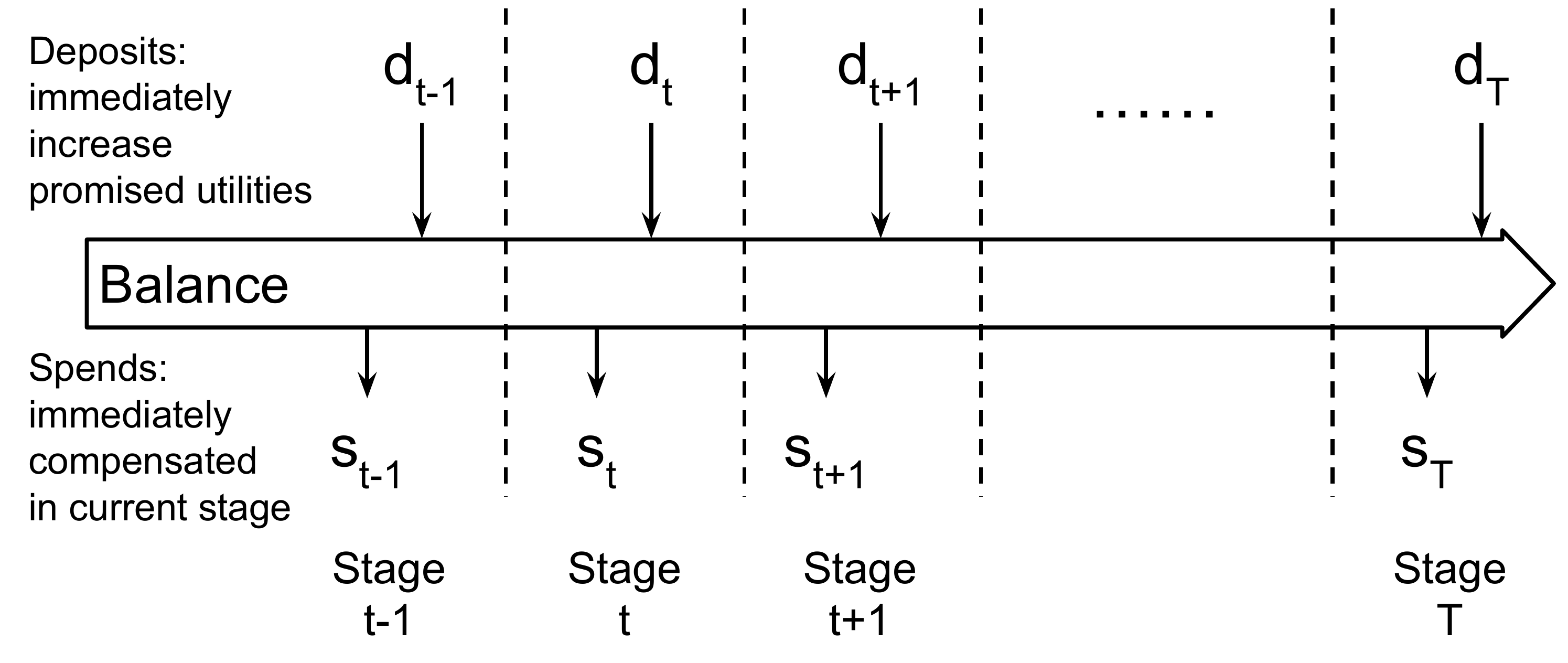}
    \caption{Describing dynamic mechanisms with $U_t$'s (left) incurs a
             recursive structure over $U_t$. Introducing bank account balance
             (right) helps us to get rid of such recursions, i.e., the deposits
             $d_t$ and spends $s_t$ are rather independent with each other.}
    \label{fig:balance}
  \end{figure}

%   The balance has the same syntactic form as the per stage charge $U_t(\cdot)$, but with different interpretations:
%  can be regarded as a more succinct way to represent
%
%   The balance plays a key role in our bank account mechanism due
%  to the following two aspects: (1) the outgoing flow from it specifies the
%  amount of promised utility to be realized in each stage; (2) the incoming
%  flow to it identifies the quota (or limit) on future outgoing flows to
%  prevent the mechanism from violating the IR constraint.

  \begin{definition}[Bank Account Mechanism]\label{def:bam}
    A bank account mechanism $B$ is a tuple $\langle z_{(1, T)},\ab q_{(1, T)},
    \ab \bal_{(1, T)},\ab d_{(1, T)}, s_{(1, T)} \rangle$, where for each $t$,
    \begin{itemize}
      \item allocation rule $z_t : \Real_+ \times \dV_t \rightarrow
            \dX_t$, that maps balance and stage type to stage allocation,
      \item payment rule $q_t : \Real_+ \times \dV_t \rightarrow \Real_+$,
            that maps balance and stage type to stage payment,
      \item balance function $\bal_t: \dV_{(1,t-1)} \rightarrow \Real_+$ is
            defined recursively by the following equation,
            \begin{align}\label{eq:defbal}
              \forall v_t,~\bal_{t + 1} = \bal_t - s_t(\bal_t) +
                d_t(\bal_t, v_t),
            \end{align}
            where $\bal_1 = 0$. Mathematically, $\bal_t$ is a function of
            history $v_{(1, t - 1)}$. We will often refer to
            $\bal_t(v_{(1, t - 1)})$ simply as $\bal_t$ and think of it as a
            variable that is updated in the course of the mechanism as the
            types are revealed.
      \item deposit policy $d_t : \Real_+ \times \dV_t \rightarrow \Real_+$,
            that maps balance and stage type to a non-negative real that is to
            be added to the current balance,
      \item spend policy $s_t : \Real_+ \rightarrow \Real$, that maps balance
            to a real (no more than balance), that is to be subtracted from the
            current balance.
    \end{itemize}
  \end{definition}
  \begin{figure}
    \centering
    \includegraphics[width=0.45\textwidth]{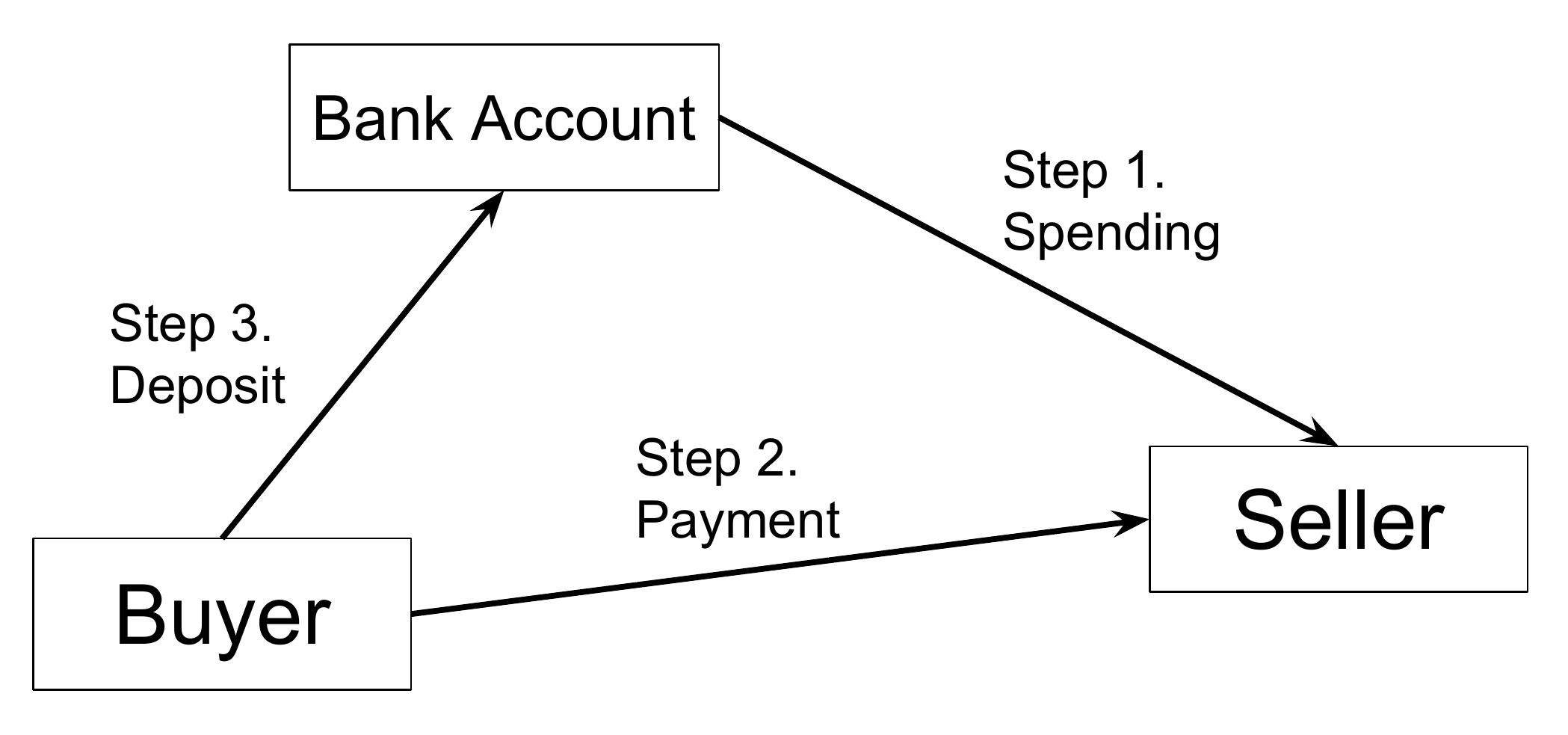}
    \caption{Transfer in bank account mechanisms. Revenue comes from both
             Spending (Step 1) and Payment (Step 2). In principle, the balance
             in the bank account belongs to the buyer, but the seller can spend
             it according to policy $s_t$.}
    \label{fig:bam}
  \end{figure}
  At stage $t$, a bank account mechanism works as follows, (see Figure
  \ref{fig:bam})
  \begin{enumerate}
    \item \label{stp:spd}
          The seller spends the balance by $s_t(\bal_t)$, before learning the
          buyer's stage type $v_t$. ($s_t(\bal_t)$ directly transfers to the
          seller's revenue.)
    \item \label{stp:all}
          Upon receiving buyer's report $v_t$, the seller allocates according
          to $z_t(\bal_t, v_t)$ to the buyer and charges $q_t(\bal_t, v_t)$.
    \item \label{stp:dps}
          The buyer deposits an amount of $d_t(\bal_t, v_t)$ to the bank
          account.
  \end{enumerate}

  It is useful to describe how a bank account mechanism maps to a direct
  revelation mechanism. Given a description of a BAM in terms of $z_{(1, T)},
  \ab q_{(1, T)},\ab \bal_{(1, T)},\ab d_{(1, T)}, s_{(1, T)}$ we map it to the
  direct mechanism described by $x_{(1,T)}, p_{(1,R)}$.

  For the allocation rule, we define:
  \begin{align*}
    x_t(v_{(1, t)}) = z_t\big(\bal_t(v_{(1, t - 1)}), v_t\big).
  \end{align*}
  The payment to the seller is given by:
  \begin{align*}
    p_t(v_{(1, t)}) = s_t\big(\bal_t(v_{(1, t - 1)})\big)
                        + q_t\big(\bal_t(v_{(1, t - 1)}), v_t\big)
  \end{align*}
  While the deposit $d_t(\bal_t, v_t)$ is not present in the description above,
  it affects the mechanisms in the sense that it dictates how the balance is
  updated from one step to the other.\\

  Consider the following example that illustrates the bank account mechanism.
  \begin{example}\label{example:bam}
    For the $T = 2$, one item per stage case. Let the valuation of each item
    be independently identically drawn from the {\em equal revenue
    distribution}, i.e.,
    \begin{align*}
      \F_1 = \F_2 = \F,~
      \F(v) = \left\{\begin{array}{ll}
        0, & v \leq 1  \\
        1 - 1 / v, & 1 < v < v_{\max}  \\
        1, & v \geq v_{\max}
      \end{array}\right..
    \end{align*}

    The following is a bank account mechanism for this setting.
    \begin{itemize}
      \item Starting with $\bal_1 = 0$, spend $0$ from the balance, i.e.,
            $s_1(0) = 0$ (Step \ref{stp:spd}).
      \item The first item is sold at posted-price $1$ (Step \ref{stp:all}).
      \item The buyer deposits money of amount $v_1 - 1$ to the bank account,
            i.e., $d_1(0, v_1) = v_1 - 1$ (Step \ref{stp:dps}).
      \item All the balance so far are spent at the beginning of stage $2$,
            i.e., $s_2(\bal_2) = \bal_2$ (Step \ref{stp:spd}).
      \item The second item is sold at posted-price $v_{\max} / e^{v_1 - 1}$
            (Step \ref{stp:all}).
      \item The buyer deposit nothing at stage $2$, i.e., $d_1(\bal_2,v_2) = 0$
            (Step \ref{stp:dps}).
    \end{itemize}

    First of all, one can verify that $B$ is IC and IR: it is straightforward
    to verify that $B$ is IR and IC for stage $2$. For stage $1$, (1) reporting
    $v'_1 < 1$ yields utility $0$, since the buyer neither gets the first item
    ($v'_1$ is less than the reserve), nor gets the second item (the reserve is
    higher than $v_{\max}$); (2) reporting any $v'_1 \geq 1$ makes no
    difference to the buyer than truthfully reporting,
    \begin{align*}
      & v_1 - 1 - (v'_1 - 1) + \E_{v_2 \geq r(v'_1)} \big(v_2 - r(v'_1)\big)
        = v_1 - v'_1 + \ln v_{\max} - \ln r(v'_1)  \\
      &~= v_1 - v'_1 + \ln v_{\max} - \ln v_{\max} + \ln e^{v'_1 - 1}
        = v_1 - 1.
    \end{align*}

    In fact, it is not hard to show that the bank account mechanism described
    above is revenue optimal among all dynamic mechanisms, with optimal revenue
    $2 + \ln\ln v_{\max}$, while any history independent mechanism can obtain
    revenue at most $2$.
  \end{example}

  To sum up, the spend $s_t(\bal_t)$ in Step \ref{stp:spd} can be considered as
  a prepayment for items to be sold in stage $t$. Then, based on the prepayment,
  the seller provides a per-stage mechanism (allocation and payment) to the
  buyer for stage $t$ in Step \ref{stp:all}. The prepayment and mechanism are
  carefully chosen so as to guarantee stagewise IC. Finally, the buyer deposits
  part of his/her utility to the bank account balance in Step \ref{stp:dps} to
  enable prepayments of the future stages.

  \subsection{Conditions for IC and IR bank account mechanisms}

    % By Definition \ref{def:bam}, any bank account mechanism is a direct
    % mechanism. However, the bank account mechanisms specified by Definition
    % \ref{def:bam} do not always satisfy IC and IR conditions. From the previous
    % observation, if (i) for any fixed balance, the stage allocation and payment
    % rules, $\langle z_t(\bal_t, \bcdot), q_t(\bal_t, \bcdot) \rangle$, are per
    % stage incentive compatible (see equation (\ref{cond:sic}) below), and (ii)
    % the prepayment and per-stage mechanism induced by different balances are
    % indifferent to the buyer (see equation (\ref{cond:spd}) below), then IC is
    % satisfied. Meanwhile, as long as the balance is not too large (see equation
    % (\ref{cond:ir}) below), IR is satisfied as well.
    %
    % Formally, the following lemma gives a sufficient condition for a bank
    % account mechanism to be IC and IR.

    We discuss now sufficient conditions for a bank account mechanism to
    satisfy IC and IR constraints. At any given stage $t$, the type $v_t$
    reported can affect the agent in two ways: it will influence the utility
    he/she gets in the single-shot mechanism defined by $z_t(\bal_t, \bcdot),
    q(\bal_t, \bcdot)$ and it will affect how much is deposited in his bank
    account, which will in turn affect how much is spent in future rounds. If
    somehow we can design the mechanism such that the agent is indifferent
    about how much we deposit in his/her account and spend from his/her account
    (i.e., indifferent about the actual balance he/she has), then we can focus
    on making the stage mechanism truthful.

    Define the utility of the $\langle z, q \rangle$-mechanism at stage $t$ as
    when the type is $v_t$ and the reported type is $v'_t$ as:
    \begin{align}\label{eq:hu}
      \hu_t(\bal_t, v'_t; v_t) =
          z_t(\bal_t, v'_t) \cdot v_t - q_t(\bal_t, v'_t).
    \end{align}
    The expected utility in stage $t$ of the corresponding direct mechanism is:
    \begin{align}
      \E_{v_t} \big[u_t(v_{(1, t)}; v_t)\big] =
          - s_t(\bal_t) + \E_{v_t} \big[\hu_t(\bal_t, v_t; v_t)\big].
    \end{align}
    Our next lemma shows that if the previous expression is independent of the
    bank balance $\bal_t$, it is enough for the $\langle z, q \rangle$-mechanism
    to be truthful as a single-shot mechanism for every bank balance $\bal_t$.

    % {\edit {\bf Comments:} @Renato: I'm not sure if I understand the paragraphs
    % above correctly. Then I write down what I think on this lemma.

    % The stuffs on stage mechanism --- I understand.

    % The deposits control the maximum amount that the seller might spend,
    % therefore ensure ex-post IR (if the deposits are properly designed).

    % Then if how much we spend from the buyer's account is indifferent to the
    % buyer --- the amount of spend equals to the increased expected stage
    % utility comparing with the base utility (the expected stage utility with
    % $0$ spend), i.e., $s_t(\bal_t) = \E_{v_t} \big[\hu_t(\bal_t, v_t; v_t)
    % - \hu_t(0, v_t; v_t)\big]$ --- the bank account mechanism is IC and IR.

    % The reason we use (\ref{cond:spd}) instead, is that there is no guarantee
    % that $\bal_t$ can reach $0$ under some history.
    % }

    \begin{lemma}\label{lem:sconic}
      Bank account mechanism $B$ is IC, if the following two conditions are
      satisfied,
      \begin{align}
        & \forall t \in [T],~\bal_t, \bal'_t \geq 0,~v_t, v'_t \in \dV_t,
          \nonumber  \\
        & \hu_t(\bal_t, v_t; v_t) \geq \hu_t(\bal_t, v'_t; v_t),
          \label{cond:sic}  \\
        & s_t(\bal_t) - s_t(\bal'_t) =
            \E_{v_t} \big[\hu_t(\bal_t, v_t; v_t)
              - \hu_t(\bal'_t, v_t; v_t)\big].
          \label{cond:spd}
      \end{align}
    \end{lemma}

    \begin{lemma}\label{lem:sconir}
      Bank account mechanism $B$ is IR, if the following condition is satisfied,
      \begin{align}
        \hu_t(\bal_t, v_t; v_t) \geq d_t(\bal_t, v_t). \label{cond:ir}
      \end{align}
    \end{lemma}
    %
    % \begin{lemma}\label{lem:scon}
    %   Let $\hu_t(\bal_t, v'_t; v_t)$ denote the stage utility,
    %   \begin{equation}
    %     \hu_t(\bal_t, v'_t; v_t) = z_t(\bal_t, v'_t) \cdot v_t -
    %       q_t(\bal_t, v'_t).  \label{eq:hu}
    %   \end{equation}
    %   $B$ is IC and IR, if the following set of conditions are satisfied,
    %   \begin{align}
    %     & \forall t \in [T],~\bal_t, \bal'_t \geq 0,~v_t, v'_t \in \dV_t,
    %       \nonumber  \\
    %     & \hu_t(\bal_t, v_t; v_t) \geq \hu_t(\bal_t, v'_t; v_t),
    %       \label{cond:sic}  \\
    %     & s_t(\bal_t) - s_t(\bal'_t) =
    %         \E_{v_t} \big[\hu_t(\bal_t, v_t; v_t)
    %           - \hu_t(\bal'_t, v_t; v_t)\big],
    %       \label{cond:spd}  \\
    %     & \hu_t(\bal_t, v_t; v_t) \geq d_t(\bal_t, v_t), \label{cond:ir}
    %   \end{align}
    %   In particular, (\ref{cond:sic})(\ref{cond:spd}) $\Infer$ IC;
    %   (\ref{cond:ir}) $\Infer$ IR.
    % \end{lemma}
    In the remainder of this paper, we use {\em BAMs} to refer to bank account
    mechanisms satisfying (\ref{cond:sic})-(\ref{cond:ir}). Given the previous
    lemmas, we refer (\ref{cond:sic}) and (\ref{cond:spd}) as IC constraint for
    BAMs, and (\ref{cond:ir}) as IR constraint for BAMs.

    % \songtodo{I beleve the following is equivalent to something that I added
    % above. Can you check if it is the case. If so, can you remove the remainder
    % of 3.1.}
    %
    % We say that a direct mechanism $M$ implements a BAM $B$ (or induced by $B$),
    % if $M$ has the same allocation and payment rules, i.e.,
    % \begin{align}\label{prf:scon:constr}
    %   \forall t \in[T], v_{(1, T)} \in \dV_{(1, T)},~&
    %   \left\{\begin{array}{l}
    %     x_t(v_{(1, t)}) = z_t(\bal_t, v_t)  \\
    %     \displaystyle
    %     \sum_{\tau = 1}^T p_\tau(v_{(1, \tau)}) = \sum_{\tau = 1}^T
    %                 \Big(q_\tau(\bal_\tau, v_\tau) + s_\tau(\bal_{\tau})\Big)
    %   \end{array}\right.  \\
    %   \Infer~&
    %   \sum_{\tau = 1}^T u_\tau(v_{(1, \tau)}; v_\tau)
    %   = \sum_{\tau = 1}^T \big(\hu_\tau(\bal_\tau, v_\tau; v_\tau)
    %                            - s_\tau(\bal_\tau)\big). \label{eq:utlequiv}
    % \end{align}
    %
    % Then the proof of Lemma \ref{lem:scon} is to verify that if $B$ satisfies
    % the IC and IR constraints (\ref{cond:sic})-(\ref{cond:ir}), then $M$ is IC
    % and IR.
    %
    % We omit the proof details. All omitted proofs in our paper can be found in
    % Appendix.

    % The following lemma states that assuming a BAM satisfies stage-wise IR
    % constraint (\ref{mech:ir}) is without loss of generality.
    %
    % \begin{lemma}\label{lem:swirimp}
    %   Any BAM can be implemented by a stage-wise IR direct mechanism.
    %   % Moreover, if $\min \bal_{T + 1} = 0$, $M$ is also stage-wise NPT.
    % \end{lemma}

  \subsection{It is without loss of generality to focus on Bank Account
              Mechanisms}

    The reason why we can focus on bank account mechanisms instead of general
    dynamic mechanisms (other than the fact that they are simple) is formally
    justified by the following theorem.
    \begin{theorem}\label{thm:rep}
      For any direct mechanism $M$, there is a BAM $B$, such that,
      \begin{align*}
        \Rev(B) \geq \Rev(M),~\Utl(B) = \Utl(M).
      \end{align*}
      In addition, if $M$ is deterministic, $B$ is deterministic.
    \end{theorem}

    Before conducting the proof of this theorem, we first need to understand
    that the set of BAMs is a strict subset of the set of direct mechanisms. In
    fact, for any BAM that defines the same balance for two different
    histories, $\bal_{t + 1}(v_{(1, t)}) = \bal_{t + 1}(v'_{(1, t)})$, the
    submechanisms after the two histories must be identical, since balance is
    the only history related information stored in a BAM. This is not the case
    for general direct mechanisms, which may be path-dependent.

    Our idea is to conduct a two-step reduction from $M$ to $B$: firstly, from
    any direct mechanism $M$ to a {\em symmetric} direct mechanism $M'$. We say
    that a mechanism is symmetric if whenever conditioned on any two histories
    at time $t$, if the expected utility is the same then the allocation and
    payment rules are also the same. Formally:

    \begin{definition}[Symmetric Mechanism]\label{def:symm}
      Given any direct mechanism $M$ define the equivalence relation
      $v_{(1, t - 1)} \sim v'_{(1, t - 1)}$ between histories as:
      \begin{align}\label{eq:equiv}
        v_{(1, t)} \sim v'_{(1, t)} \iff
          \Utl(M | v_{(1, t)}) = \Utl(M | v'_{(1, t)}).
      \end{align}
      We say that mechanism $M$ is {\em symmetric}, if for any pair of
      equivalent histories, $v_{(1, t - 1)} \sim v'_{(1, t - 1)}$, the
      corresponding submechanisms are identical, i.e.,
      \begin{align*}
        \forall t \leq t' \leq T, v_{(t, t')} \in \dV_{(t, t')},~
        \left\{\begin{array}{l}
          x_{t'}(v_{(1, t - 1)}, v_{(t, t')})
            = x_{t'}(v'_{(1, t - 1)}, v_{(t, t')})  \\
          p_{t'}(v_{(1, t - 1)}, v_{(t, t')})
            = p_{t'}(v'_{(1, t - 1)}, v_{(t, t')})
        \end{array}\right..
      \end{align*}
    \end{definition}

    Our first step is to show that for any mechanism, it is possible to
    construct a symmetric mechanism where the utility is the same and the
    revenue is at least as large:

    \begin{lemma}\label{lem:mechtrans}
      For any direct mechanism $M$, there is a symmetric direct mechanism $M'$
      such that
      \begin{align*}
        \Utl(M') = \Utl(M),~\Rev(M') \geq \Rev(M).
      \end{align*}
      In particular, if $M$ is deterministic, $M'$ is also deterministic.
    \end{lemma}

    The second step in the reduction is to show that any symmetric direct
    mechanism can be converted to a BAM.

    \begin{lemma}\label{lem:symm}
      For any symmetric direct mechanism $M'$, there is a Core BAM $B$ such
      that
      \begin{align*}
        \Utl(B) = \Utl(M'),~\Rev(B) = \Rev(M').
      \end{align*}
    \end{lemma}

    In fact, we will show that any mechanism can be converted to a specific
    type of BAM, which we will denote {\em Core BAM}. The rest of this section
    is devoted to proving the previous lemmas, which taken together constitute
    a proof of Theorem \ref{thm:rep}. We will start by defining the concept of
    Core BAM which will be central to our analysis.

    Let $\dH = \dH_0 \cup \dH_1 \cup \cdots \cup \dH_T$ be the set of all
    histories, where $\dH_t = \dV_{(1, t)}$ is the set of all histories of
    length $t$ ($0 \leq t \leq T$). In particular, $\dH_0 = \{\emptyset\}$.

    For any $0 \leq t \leq T$, Let $g_t$ be a function that maps a history of
    length $t$ to a real number, and $y_t$ (for $t \geq 1$) be a function that
    maps a history of length $t$ to a stage allocation.
    \begin{align*}
      & g_t : \dH_t \rightarrow \Real,~\forall 0 \leq t \leq T;  \\
      & y_t : \dH_t \rightarrow \dX_t,~\forall t \in [T].
    \end{align*}

    % {\edit {TODO:} @Song, use $\partial{v_t}$ instead of $y_t$.}
    %
    % \songtodo{See my comment later in this write-up about defining Core BAM
    % using only $g$ and just taking $y$ directly to be $\partial_{v_t} g$.}
    %
    % {\edit {\bf Comment:} @Renato: I agree with you that we can use $g$ and
    % $\partial_{v_t}g$ instead. The reason that I use $g$ and $y$ is that $y$ is
    % not fully determined by $g$ when $g$ is not always differentiable ($y$ can
    % be any sub-gradient of $g$ at the indifferentiable points). Then the
    % alternative way I think could be adding a note here to say that it does not
    % matter much even if we have multiple choices for $\partial_{v_t}g$.}

    Intuitively, we will construct a BAM, called core BAM, that yields the same
    conditional utility as $g$ and the same allocation as $y$, where $g$ and
    $y$ denote $g_0, \ldots, g_T$ and $y_1, \ldots, y_T$, respectively.
    \begin{definition}[Core BAM]\label{def:corebam}
      Consider the following construction of bank account mechanism based on
      $g$ and $y$, denoted as $B^{g, y}$.
      \begin{align}
        & \bal_{t + 1}(v_{(1, t)}) = g_t(v_{(1, t)}) - \mu_t,
          \label{eq:bbal}  \\
        & z_t\big(\bal_t(v_{(1, t - 1)}), v_t\big) = y_t(v_{(1, t)}),
          \label{eq:ball}  \\
        & q_t(\bal_t, v_t) = z_t(\bal_t, v_t) \cdot v_t
                             - \int_\Zero^{v_t} z_t(\bal_t, v) \ud v,
          \label{eq:bpay}  \\
        & d_t(\bal_t, v_t) = \hu_t(\bal_t, v_t; v_t),
          \label{eq:bdps}  \\
        & s_t(\bal_t) = \bal_t + d_t(\bal_t, v_t) - \bal_{t + 1},
          \label{eq:bspd}
      \end{align}
      where for $t < T$, $\mu_t = \inf_{v_{(1, t)}} g_t(v_{(1, t)})$;
      $\mu_T = \min\big\{0, \inf_{v_{(1, T)}} g_T(v_{(1, T)})\big\}$.

      When the construction above is a valid BAM, we call it a {\em core BAM}.
    \end{definition}

    % {\edit {\bf Comment:} @Renato: Agreed. This paragarph here is trying to
    % give some intuition of the construction above. However, it gives little new
    % information.}
    %
    % \songtodo{Song, I want to remove the following paragraph as I don't see any
    % new informaiton being given there. Do you agree ? Paragraph is:
    % The mechanism $B^{g, y}$ is a bank account mechanism with balance function
    % $\bal_t$ and allocation rule $z_{(1, T)}$ uniquely determined by $g$ and
    % $y$, with payment rule $q_t$ so constructed to satisfy IC per stage (given
    % the allocation rule is increasing), with deposit policy $d_t$ constructed
    % in a way that deposits the largest possible amount into the balance subject
    % to IR constraint (\ref{cond:ir}), and with spend policy $s_t$ defined by
    % balance and deposit functions according to the balance update formula
    % (\ref{eq:defbal}).}

    For the construction in Definition \ref{def:corebam} to be valid BAM we
    need certain conditions to be satisfied. For equation (\ref{eq:ball}) to be
    well-defined, we need that $y_t(v_{(1, t)}) = y_t(v'_{(1, t - 1)}, v_t)$
    whenever $g_t(v_{(1, t - 1)}) = g_t(v'_{(1, t - 1)})$, otherwise equation
    (\ref{eq:ball}) would imply different values of $z_t$ for the same bank
    balance. Also for IC constraint (\ref{cond:sic}) to be satisfied, we need
    the allocation function  $y_t(v_{(1, t)})$ to be weakly increasing in $v_t$.
    In Theorem \ref{thm:corebam} we will provide necessary and sufficient
    conditions on $g_t$ and $y_t$ for the mechanism to be a valid core BAM.

    We also want to remind that reader that our theorems hold in the case where
    more than one item is auctioned in each stage. In the first time the reader
    reads our theorems, we recommend him to focus on the single-item per round
    case and interpret the integrals as one-variable integrals. However, they
    hold in general multi-dimensional case as well, by interpreting the
    integration in (\ref{eq:bpay}) as a multidimensional path-integral $\Zero$
    to $v_t$. If $z_t(\bal_t, v)$ is the sub-gradient of some multi-dimensional
    valued input, then the value of the integral is independent of the
    integration path. As we will see in Theorem \ref{thm:corebam}, a necessary
    condition for the mechanism $B^{g, y}$ to be a Core BAM is that $y_t$ is
    the sub-gradient of $g_t$.

    Core BAM is a key construction throughout this paper because it is without
    loss of generality to restrict to Core BAMs, as we will prove that Theorem
    \ref{thm:rep} is still true even if we replace BAM by Core BAM in the
    statement of the theorem.
    % \begin{definition}[Symmetric Mechanism]\label{def:symm}
    %   A direct mechanism $M$ is {\em symmetric}, if for any pair of equivalent
    %   histories, $v_{(1, t - 1)} \sim v'_{(1, t - 1)}$, the corresponding
    %   submechanisms are identical, i.e.,
    %   \begin{align*}
    %     \forall t \leq t' \leq T, v_{(t, t')} \in \dV_{(t, t')},~
    %     \left\{\begin{array}{l}
    %       x_{t'}(v_{(1, t - 1)}, v_{(t, t')})
    %         = x_{t'}(v'_{(1, t - 1)}, v_{(t, t')})  \\
    %       p_{t'}(v_{(1, t - 1)}, v_{(t, t')})
    %         = p_{t'}(v'_{(1, t - 1)}, v_{(t, t')})
    %     \end{array}\right.,
    %   \end{align*}
    %   where the equivalence relation between any two histories are defined as
    %   follows,
    %   \begin{align}\label{eq:equiv}
    %     v_{(1, t)} \sim v'_{(1, t)} \iff
    %       \Utl(M | v_{(1, t)}) = \Utl(M | v'_{(1, t)}).
    %   \end{align}
    % \end{definition}

    % With the above definitions, Theorem \ref{thm:rep} follows from Lemma
    % \ref{lem:mechtrans} and Lemma \ref{lem:symm}.
    % \begin{lemma}\label{lem:mechtrans}
    %   For any direct mechanism $M$, there is a symmetric direct mechanism $M'$
    %   such that
    %   \begin{align*}
    %     \Utl(M') = \Utl(M),~\Rev(M') \geq \Rev(M).
    %   \end{align*}
    %   Particularly, if $M$ is deterministic, $M'$ is also deterministic.
    % \end{lemma}

    Now we are ready to prove the lemmas. Both of the proofs are constructive.
    The idea to prove Lemma \ref{lem:mechtrans} is sketched as follows:
    \begin{enumerate}
      \item Given a direct mechanism $M$, we construct equivalence classes
            defined by (\ref{eq:equiv}).
%      \item  For any two equivalent histories,
%            $v_{(1, t - 1)}, v'_{(1, t - 1)}$, such that
%             $\Utl(M | v_{(1, t)}) = \Utl(M | v'_{(1, t)}).$
      \item For each equivalence class, select a {\em representative history}
            $v^*_{(1, t - 1)}$ such that the revenue from the
            submechanism at this history dominates the expected revenue from
            the submechanisms at the histories in this equivalence class.
            % \renato{which maximizes the revenue, breaking
            % ties arbitrarily.} {\edit {\bf Comments:} @Renato: I agree that
            % choosing maximizers here greatly helps readers to understand the
            % idea. However, we don't have rigorous proofs to show such maximizers
            % exist.}
      \item At history $v_{(1, t - 1)}$ and stage type $v_t$, $M'$ simulates
            what $M$ does at history $v^*_{(1, t - 1)}$ with stage type $v_t$.
            % \songtodo{Song, from this point on, isn't it easier to do things by
            % backwards induction, i.e., from the last stage to the first.}
            %
            % {\edit {\bf Comments:} @Renato: I'm not sure if the original proof
            % would directly apply in backwards induction. Since during the
            % backwards construction process, changes made on higher levels (with
            % smaller $t$) could re-define the equivalence classes on lower levels
            % (with larger $t$). Then the representative history chosen for some
            % lower level may not dominates the expected revenue in the
            % corresponding equivalence class anymore. (Since we don't have
            % guarantees on being able to choose the maximizer in each class as
            % the representative history yet.)}
      \item However, applying such simulation to histories at stage $t$ will
            destroy the equivalence relation in stage $t + 1$ since it changes
            the utilities in stage $t + 1$.
      \item Our idea is to replace each history with its representative history
            inductively, stage by stage (so that when the replacement procedure
            is applied up to stage $t$, then equivalence relation in stage
            $t + 1$ is determined).
      \item The so-constructed mechanism is symmetric and we verify that it
            satisfies the desirable properties.
    \end{enumerate}

    The representative history in an equivalent class is a history in this
    class such that the seller's revenue at this history is no less than the
    expected revenue in this equivalence class. Clearly, such a history always
    exists. Intuitively, by replacing all histories in an equivalent class with
    such a representative history, one can ensure that the seller's expected
    revenue weakly increases while the buyer's utility remains the same.

    Instead of directly proving Lemma \ref{lem:symm}, we prove the following
    stronger version: for each symmetric direct mechanism, one can construct a
    Core BAM with the same overall outcome.
    \begin{lemma}\label{lem:symmstrong}
      For any symmetric direct mechanism $M$, $B^{g, y}$ is a BAM, if
      \begin{align}\label{eq:lemsymmstat}
        \forall 0 \leq t \leq T, h_t \in \dH_t, g_t(h_t) = \Utl(M | h_t),~
        \forall t \in [T], h_t \in \dH_t, y_t(h_t) = x_t(h_t).
      \end{align}

      Moreover, $B^{g, y}$ has the same overall allocation and payment, as a
      result,
      \begin{align*}
        \Utl(B^{g, y}) = \Utl(M) = g_0(\emptyset),~\Rev(B^{g, y}) = \Rev(M).
      \end{align*}
    \end{lemma}

    Then the proof of this lemma is to verify that so-constructed $B^{g, y}$ is
    indeed a BAM.
    %
    % \input{prf-lem-symm.tex}
    %
    % \begin{proof}[of Theorem \ref{thm:rep}]
    %   Immediately from Lemma \ref{lem:mechtrans} and Lemma \ref{lem:symm}.
    % \end{proof}

  % \subsection{Characterization of core BAMs {\edit I suggest to remove this
  %   title.}}

    % \songtodo{Hi Song, I was wondering if any of the results in this section
    % are used in subsequent section. This section is very nice, but if the
    % results are not crucial, I think it delays the reader to get to the main
    % results. I'd condense it in half a page and send the rest to the appendix.
    % Do you agree?}
    %
    % {\edit {\bf Comment:} @Renato, I agree with to condense, although the
    % theorem 3.11 is used in Section 6 (so that when writing the program, $g$ is
    % the only variable).}

    % {\edit {\bf TODO:} @Song, condense this subsection in half a page and send
    % the rest to the appendix.}

    In light of Theorem \ref{thm:rep}, it is without loss of generality to
    consider core BAMs only for the purpose of this paper. Most of our results
    (approximation and computation) rely on the construction of core BAMs.
    However, as mentioned, the construction $B^{g, y}$ is not guaranteed to be
    a core BAM for an arbitrary pair of $g$ and $y$. At the end of this section,
    we characterize the set of pairs of $g$ and $y$ such that $B^{g, y}$ is a
    Core BAM (Theorem \ref{thm:corebam}). The characterization will be used to
    construct feasibility constraints for finding the optimal Core BAMs in
    Section \ref{sec:fptas}.
    \begin{theorem}\label{thm:corebam}
      $B^{g, y}$ is a core BAM, if and only if for any $t \in [T]$,
      \begin{itemize}
        \item $y_t$ is the sub-gradient of $g_t$ with respect to $v_t$ with
              range being $\dX_t$, i.e., $y_t(v_{(1, t)}) =
              \partial g_t(v_{(1, t)}) / \partial v_t$ and $y_t(v_{(1, t)}) \in
              \dX_t$.
        \item $g_t$ is {\em consistent}, {\em symmetric}, convex in $v_t$ and
              weakly increasing in $v_t$.
      \end{itemize}

      Consistency and symmetry are defined as follows,
      \begin{align*}
                           & \forall t \in [T],
                             v_{(1, t)}, v'_{(1, t)} \in \dV_{(1, t)}  \\
        \text{consistent:} & \quad g_{t - 1}(v_{(1, t - 1)})
                             - \E_{v_t}\big[g_t(v_{(1, t)})\big] = c_t,  \\
        \text{symmetric:}  & \quad g_{t - 1}(v_{(1, t - 1)})
                                        = g_{t - 1}(v'_{(1, t - 1)})
                             \Infer g_t(v_{(1, t)}) = g_t(v'_{(1, t - 1)}, v_t),
      \end{align*}
      where $c_t$ is a constant.
    \end{theorem}

    As a result, it is without loss of generality to study BAMs and Core BAMs
    in the rest of the paper.

    Figure \ref{fig:hrch} visualizes our findings so far. For more details and
    intuitions behind, we refer the readers to Appendix (along with the proof
    of Theorem \ref{thm:corebam}).

    \begin{figure}
      \centering
      \includegraphics[width=0.55\textwidth]{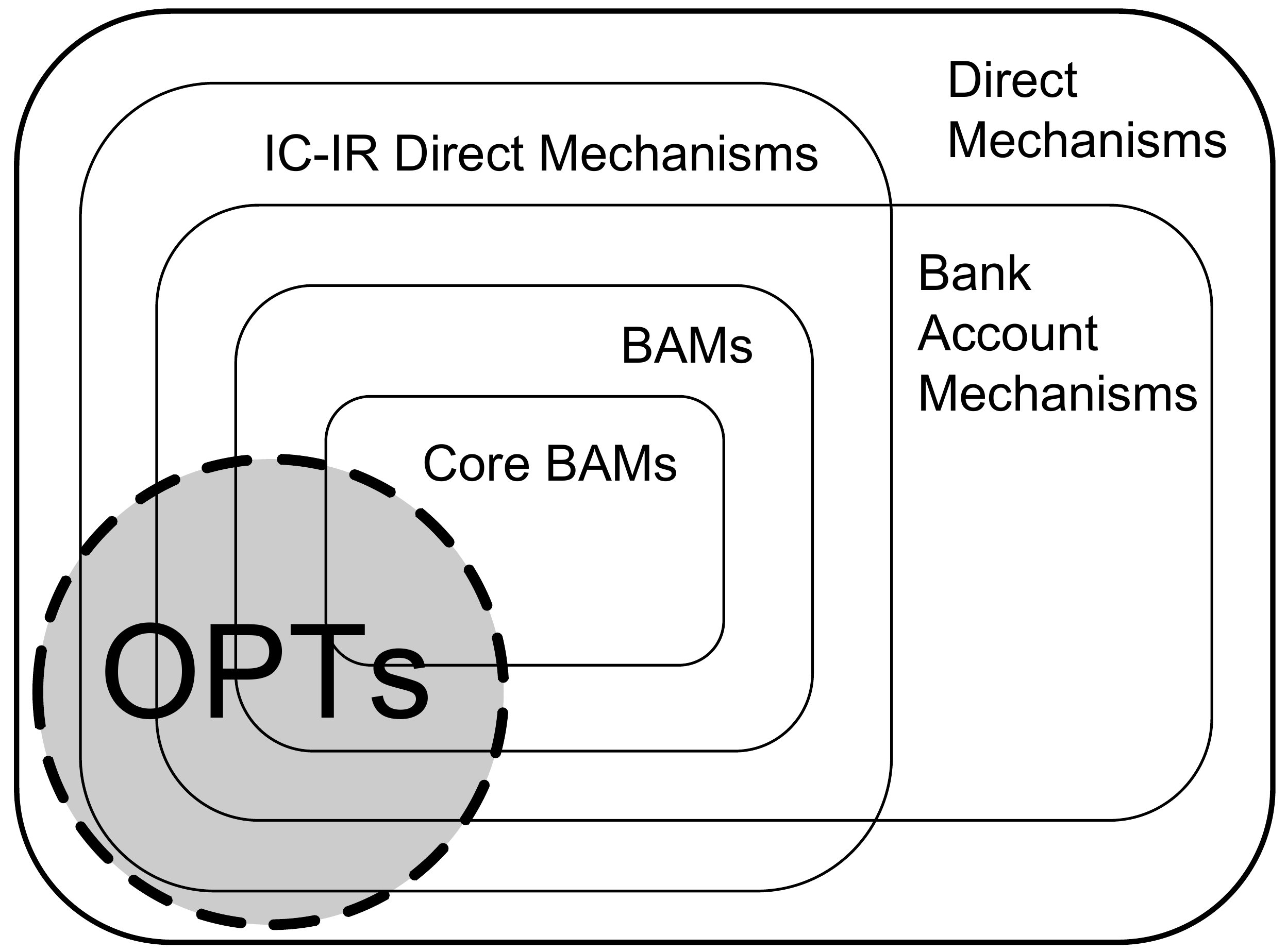}
      \caption{The intersection of OPTs and Core BAMs is
               non-empty.}
      \label{fig:hrch}
    \end{figure}

% Main Results
\section{Approximation of Optimal Mechanism}\label{sec:3app}

  By Theorem \ref{thm:rep}, the revenue optimal bank account mechanism is
  optimal among any mechanism. However, the optimal bank account mechanism can
  still be complicated. This observation is formalized in a later section.
  Later in this section, we will use an example to informally illustrate this
  observation.

  \subsection{A simple BAM that $3$-Approximates the optimal revenue}

    In this section we present a very simple BAM whose revenue is a
    $3$-approximation of the revenue of the revenue optimal mechanism. By
    Theorem \ref{thm:rep} we know that the optimal revenue is achieved by a
    BAM, so our strategy will be to first provide an upper bound on the maximum
    revenue that a BAM mechanism can extract. Then we will craft three
    individual mechanisms achieving one of the components of the upper bound as
    revenue. Then we will argue that their combination is a constant
    approximation to the optimal revenue.

    We begin with an upper bound on the revenue of a BAM. The following theorem
    is a version of Theorem 3 in our previous paper where bank accounts were
    introduced \cite{mirrokni2016dynamic}.

    \begin{theorem}\label{thm:bl}
      For any BAM $B$,
      \begin{align*}
        \Rev(B) \leq \Rev(\SMA) +
          \E_{v_{(1, T)}} \bigg[\sum_{\tau = 1}^T s_\tau(\bal_\tau)\bigg].
      \end{align*}
      where $\SMA$ denotes the optimal history-independent mechanism, i.e.,
      within each stage, run the single-stage optimal mechanism.
    \end{theorem}

    % \songtodo{Song, can you add a proof of the Thm above for completeness in the
    % appendix. I believe that the proof is almost direct, so it should be simple
    % to add.}
    %
    % {\edit {\bf Comments:} @Renato, the proof is in page \pageref{app:thm:bl},
    % which was incorrectly titled as Proof of Lemma \ref{thm:bl} in the previous
    % version.}

    There is a simple mechanism that achieves the first part of the upper bound,
    which is the optimal history-independent mechanism itself. In the case where
    a single item is sold per stage, this is achieved by running Myerson's
    optimal auction in each stage. The main challenge is to come up with a
    mechanism whose revenue approximates the sum of spends of \emph{any} BAM.

    We start by providing an bound on the the sum of spends $\sum_{\tau = 1}^T
    s_\tau(\bal_\tau)$ for any history $v_{(1,T)}$. Denote
    \begin{align*}
      \Val_t = \int_{\dV_t} \One \cdot v \ud \F(v),
    \end{align*}
    then

    \begin{lemma}\label{lem:3appspddpsupb}
      For any BAM, the spend and deposit in any stage can be bounded as follows:
      \begin{align*}
        & s_t(\bal_t) \leq \min \{\bal_t, \Val_t\}, \\
        & d_t(\bal_t, v_t) \leq \One \cdot v_t.
      \end{align*}
    \end{lemma}

    \begin{lemma}\label{lem:sbd}
      For any BAM $B$, and any history $v_{(1, t)} \in \dV_{(1, t)}$ of length
      $t \in [T]$. The sum of spends of $B$ on history $v_{(1, t)}$ can be
      bounded as follows,
      \begin{align*}
        \sum_{\tau = 1}^t s_\tau(\bal_\tau)
            \leq \sum_{\tau = 1}^t s^*_\tau(\bal^*_\tau),
      \end{align*}
      where $s^*_\tau$ and $d^*_\tau$ are the spend and deposit policies of a
      bank account mechanism $B^*$, defined as follows,
      \begin{align*}
        & s^*_\tau(\bal^*_\tau) = \min\{\bal^*_\tau, \Val_\tau\},  \\
        & d^*_\tau(\bal^*_\tau, v_\tau) = \One \cdot v_\tau,
      \end{align*}
      and $\bal^*_\tau$ is the balance of $B^*$ updated accordingly.
    \end{lemma}

    The intuition behind Lemma \ref{lem:sbd} is: (i) for any fixed history
    $v_{(1, t)}$, increasing any stage deposit $d_\tau$ into the bank account
    never disables the seller to spend from the bank account at any stage as
    before, so we can set each of the stage deposit $d_\tau$ to reach its upper
    bound $\One \cdot v_\tau$; (ii) after increasing the stage deposits, they
    become a sequence of constant incoming flows to the bank account, so taking
    money out from the bank account as much as possible at each stage is clearly
    the optimal way to maximize the cumulative spends until any stage.

    The takeaway from Lemma \ref{lem:sbd} is that in order to maximize the sum
    of spends, we should try to keep the balance as large as possible by
    depositing as much as possible and at the same time we should try to spend
    as much as possible. Instead of trying to find a BAM with those two
    properties simultaneously, we define two mechanisms: the first one has the
    largest possible deposit in each stage. This is called the {\em
    give-for-free} mechanism where all items are given to the buyer for free at
    every stage. Here the utility of the buyer is maximized and all the utility
    is deposited in the bank account. The second mechanism seeks to spend as
    much as possible from the balance in each round. This is called {\em
    posted-price-for-the-grand-bundle mechanism}. This mechanism will post a
    price for the grand-bundle ensuring that the buyer will have enough utility
    in expectation so that he is comfortable with a large spend from his/her
    bank account.

    Then, by the approach stated earlier, the uniform randomization of these
    three mechanisms (optimal history-independent mechanism, give-for-free
    mechanism, and posted-price-for-the-grand-bundle mechanism)
    $3$-approximates the optimal revenue.

    To formalize our reasoning so far, let $\SMA = \langle x^\SM_{(1, T)},
    p^\SM_{(1, T)} \rangle$ denote the optimal history independent mechanism,
    i.e., within each stage, separately run the single-stage optimal
    mechanism\footnote{If there is only one item for sale for each stage, then
    $\SMA$ denotes the mechanism that runs separate Myerson auction for each
    stage.}.

    Let $\langle x^\PM_t(\theta, \bcdot), p^\PM_t(\theta, \bcdot) \rangle$ be
    the posted-price-for-the-grand-bundle mechanism with parameter $\theta$ for
    stage $t$, i.e.,
    \begin{align}\label{eq:parapp}
      x^\PM_t(\theta, v_t) = \One \cdot \I[v_t \cdot \One \geq r_t(\theta)],~
      p^\PM_t(\theta, v_t) =
                  r_t(\theta) \cdot \I[v_t \cdot \One \geq r_t(\theta)],
    \end{align}
    where $r_t(\theta)$ is the posted-price such that the expected buyer
    utility yielded by the single-stage posted-price-for-the-grand-bundle
    mechanism is exactly $\theta$, formally,
    \begin{align}
      r_t(\theta) = r,~s.t.,~
        \int_{\dV_t} (v \cdot \One - r) \cdot \I[v \cdot \One \geq r] \ud v
          = \theta,
      \Infer \E_{v_t}\big[u^\PM_t(\theta, v_t; v_t)\big] = \theta.
      \label{eq:pppu}
    \end{align}

    Since the function $\int_{\dV_t} (v \cdot \One - r) \cdot \I[v \cdot \One
    \geq r] \ud v$ is continuous in $\theta$, a price $r_t(\theta)$ exists for
    every $\theta$ in the interval $[0, \Val_t]$. Note that the integration
    above is the expected utility when the items are sold as a grand bundle at
    posted-price $r$ in stage $t$.
    \begin{framed}
      \begin{mechanism}[$3$-Approximation BAM]\label{mech:3app}
        Consider the following BAM, $B$ --- the uniform randomization of the
        optimal history-independent mechanism $\SMA$, the
        posted-price-for-the-grand-bundle mechanism (with parameter
        $3s_t(\bal_t)$ for each stage $t$), and the give-for-free mechanism.

        Formally defined as follows, where $\bal_t$ is defined by $d_t$ and
        $s_t$ according to the balance update formula (\ref{eq:defbal}).
        \begin{align}
          z_t(\bal_t, v_t) &= \frac13 \Big(x^{\SM}_t(v_t) + \One +
                                  x^{\PM}_t\big(3s_t(\bal_t), v_t\big)\Big),
          \label{eq:3all}  \\
          q_t(\bal_t, v_t) &= \frac13 \Big(p^{\SM}_t(v_t) +
                                  p^{\PM}_t\big(3s_t(\bal_t), v_t\big)\Big),
          \label{eq:3pay}  \\
          d_t(\bal_t, v_t) &= \frac{\One \cdot v_t}3, \label{eq:3dps}  \\
          s_t(\bal_t) &= \min\{\bal_t, \Val_t / 3\}, \label{eq:3spd}
        \end{align}
        where $\Val_t = \int_{\dV_t} \One \cdot v \ud \F_t(v)$.
      \end{mechanism}
    \end{framed}

    Such a uniform randomization ensures that for any core BAM $B'$, the
    following properties hold simultaneously.
    \begin{itemize}
      \item The one third fraction on the optimal history-independent mechanism
            ensures enough revenue from per stage payment $q_t$:
            \begin{align*}
              \forall t \in [T], \bal_t, \bal'_t \geq 0, v_t \in \dV_t,~
                q_t(\bal_t, v_t) \geq \frac13 q'_t(\bal'_t, v_t).
            \end{align*}
      \item The one third fraction on the posted-price-for-the-grand-bundle
            mechanism with parameter\footnote{Since it is only one third
            fraction on the posted-price-for-the-grand-bundle mechanism, having
            parameter $3s_t(\bal_t)$ ensures the buyer utility yielded from
            this one third fraction being exactly $s_t(\bal_t)$.} $3s_t(\bal_t)$
            ensures enough spend $s_t$ from the bank account in each stage:
            %
            % \songtodo{I still don't understand how this is possible. I in fact
            % believe this bullet is not true.}
            %
            % {\edit {\bf Comments:} @Renato, sorry for keeping making typos.
            % This version should be correct.}
            \begin{align*}
              \forall t \in [T], \bal_t \geq 0,~
                  s_t(\bal_t) \geq \frac13 s'_t(\bal_t).
            \end{align*}
      \item The one third fraction on the give-for-free mechanism ensures
            enough deposit $d_t$ into the bank account in each stage:
            \begin{align*}
              \forall t \in [T], \bal_t, \bal'_t \geq 0, v_t \in \dV_t,~
                d_t(\bal_t, v_t) \geq \frac13 d'_t(\bal'_t, v_t).
            \end{align*}
    \end{itemize}
    % Mechanism \ref{mech:3app} can be seen as a uniform randomization of three
    % mechanisms: optimal history-independent mechanism, the give-for-free
    % mechanism, and the posted-price-for-the-grand-bundle mechanism with
    % $r_t(3s_t(\bal_t))$ for each $t$. Each of them guarantees that the revenue
    % from $q_t$ is no less than one third of $\Rev(\SMA)$, the deposit is no
    % less than one third of any core BAM, and the total spend is also no less
    % than one third of any core BAM, respectively.

    \begin{lemma}\label{lem:3appbam}
      Mechanism \ref{mech:3app} is a BAM.
    \end{lemma}

    \begin{theorem}\label{thm:3app}
      Mechanism \ref{mech:3app} $3$-approximates the optimal revenue.

      In particular, if $\SMA$ is given, then such BAM can be explicitly
      constructed.
    \end{theorem}

    % {\edit {\bf TODO:} @Song, split the proof.}

    % Proof of Theorem 3app

\begin{proof}[Proof of Theorem \ref{thm:3app}]

  By Theorem \ref{thm:bl}, we have,
  \begin{align}\label{eq:3apps1}
    \Rev(B^\OPT) \leq \Rev(\SMA) + \E_{(1, T)}\bigg[
      \sum_{\tau = 1}^T s^\OPT_\tau(\bal^\OPT_\tau)\bigg].
  \end{align}

  % For any $v_{(1, T)} \in \dV_{(1, T)}$, consider the deposit sequence of
  % $B^\OPT$: $d^\OPT_1(\bal^\OPT_1, v_1)$, $\ldots$, $d^\OPT_T(\bal^\OPT_T,
  % v_T)$, and the spend sequence of $B^\OPT$: $s^\OPT_1(\bal^\OPT_1)$, $\ldots$,
  % $s^\OPT_T(\bal^\OPT_T)$. Note that they are bounded by (\ref{eq:svb}) and
  % (\ref{eq:dvb}). Then
  
  By Lemma \ref{lem:sbd}, we have,
  \begin{align}\label{eq:3apps2}
    \sum_{\tau = 1}^T s^\OPT_\tau(\bal^\OPT_\tau)
      \leq \sum_{\tau = 1}^T s^*_\tau(\bal^*_\tau).
  \end{align}

  By the definition of $B$ (see (\ref{eq:3dps}) and (\ref{eq:3spd})), we have
  that $\bal_t = \bal^*_t / 3$, $d_t = d^*_t / 3$ and $s_t = s^*_t / 3$. Hence
  by the definition of $B$ (see (\ref{eq:3pay})),
  \begin{align}\label{eq:3apps3}
    \Rev(B) =& \E_{v_{(1, T)}} \bigg[
      \sum_{\tau = 1}^T \big(q_\tau(\bal_\tau, v_\tau)
                            + s_\tau(\bal_\tau)\big)\bigg] \nonumber  \\
         \geq&~ \frac13\E_{v_{(1, T)}}\bigg[\sum_{\tau = 1}^T
                                                p^\SM_\tau(v_\tau)\bigg]
              + \frac13\E_{v_{(1, T)}}\bigg[
                \sum_{\tau = 1}^T s_\tau(\bal_\tau)\big)\bigg] \nonumber  \\
         =&~ \frac13\bigg(\Rev(\SMA) + \E_{(1, T)}\bigg[
                \sum_{\tau = 1}^T s^*_\tau(\bal^*_\tau)\bigg]\bigg).
  \end{align}

  Combining (\ref{eq:3apps1})(\ref{eq:3apps2})(\ref{eq:3apps3}) above, we
  complete the proof,
  \begin{align*}
    \Rev(B) &\geq \frac13\bigg(\Rev(\SMA) + \E_{(1, T)}\bigg[
                    \sum_{\tau = 1}^T s^*_\tau(\bal^*_\tau)\bigg]\bigg)  \\
            &\geq \frac13\bigg(\Rev(\SMA) + \E_{(1, T)}\bigg[
                    \sum_{\tau = 1}^T s^\OPT_\tau(\bal^\OPT_\tau)\bigg]\bigg)
            \geq \frac13 \Rev(B^\OPT).
  \end{align*}
\end{proof}

    In particular, for the case where only one item is sold in each stage, the
    optimal history-independent mechanism is the separate Myerson auction for
    each stage. So we can explicitly construct Mechanism \ref{mech:3app} in this
    case.

    Also note that, the optimal history-independent mechanism is unknown for
    general cases. However, there have been many recent results concerning
    approximately optimal mechanism for certain valuation assumptions
    \cite{hart2012approximate,li2013revenue,babaioff2014simple,cai2016duality}.
    The following corollary allows us to transform these approximately optimal
    mechanisms into an approximately optimal mechanism in our setting.

%    Note that $\SMA$ is unknown for general case, but there are lots of results
%    on approximate optimal mechanism with some reasonable assumptions on $\F_t$
%    [cite]. The following corollary relaxes the approximation ratio, but
%    provides explicitly constructible mechanisms with approximation guarantees.
    \begin{corollary}
      If there is a history-independent mechanism $M$ approximates $\SMA$, i.e.,
      $\Rev(M) \ab \geq \alpha \Rev(\SMA)$,  one can explicitly construct a BAM
      $B$ that approximates the optimal dynamic mechanism $B^\OPT$, i.e.,
      \begin{align*}
        \Rev(B) \geq \frac{\alpha}{2\alpha + 1} \Rev(B^\OPT).
      \end{align*}
      $B$ is a randomization of $M$, the give-for-free mechanism, and the
      posted-price-for-the-grand-bundle with
      $r_t\big((2 + 1 / \alpha)s_t(\bal_t)\big)$ for each $t$, on weights
      $1 / (2\alpha + 1)$, $\alpha / (2\alpha + 1)$, and $\alpha /
      (2\alpha + 1)$, respectively.
    \end{corollary}

  \subsection{Global connectivity is necessary to achieve any constant
              approximations}

    % \songtodo{Two things, Song. The first one is that we talk about local
    % mechanisms without defining what local means. Maybe we want to start with a
    % definition of local. Second is that it might be better to move this part about
    % local mechanisms to the end of section 4, or at least after the $3$-approx. In
    % general, I tend to like having positive results before. Maybe you can add a
    % a paragraph here with a pointer to later. But in the end, it is your choice.
    % If you prefer, feel free to keep it here.}
    %
    % {\edit {\bf TODO:} @Song add the definition of local mechanisms, and move
    % this part to the end of section 4.}

    We show by example that the bank account structure (or some other
    structure that summarizes the entire history) is necessary to achieve
    a constant approximation as the complementary to the $3$-approximation
    results.

    A {\em local mechanism} is a submechanism that only involves a subset of
    stages regardless of the valuation distribution, and is completely
    independent with what happens in other stages. For example, the stage
    mechanisms of a history-independent mechanism are local mechanisms, which
    operates independently with each other.
    \begin{observation}\label{obsv:conn}
      Any dynamic mechanism that can be decomposed into several independent
      local mechanisms cannot be optimal --- in fact, it cannot guarantee any
      constant approximation.

      In particular, it is also true for the mechanisms that are distributions
      over such mechanisms.
    \end{observation}

    % Here, independent local mechanisms is a generalization of stage mechanisms to
    % several consecutive stages. The execution of each local mechanism is
    % independent of the outcomes of other local mechanisms.

    The observation above informally states the fact that a global structure
    such as the ``bank account'' introduced in this paper is essential to
    preserve the optimal revenue.

    \begin{example}[Necessity of global connectivity]\label{example:hard}
      Consider a dynamic mechanism $M$ that can be decomposed into two
      independent local mechanisms, i.e., there exists $T'$, $1 < T' < T$, such
      that the outcomes of stages $T' + 1$ to $T$ are independent with what has
      happened in the first $T'$ stages.

      We construct the following valuation distributions,
      \begin{itemize}
        \item the first and the $(T' + 1)$-th items are i.i.d. equal-revenue
              distributions (see Example \ref{example:bam});
        \item all other items have $0$ valuations.
      \end{itemize}

      By our assumption, $M$ treats the two items with equal-revenue valuations
      independently, and can extract revenue at most $1$ from each such stage
      without violating the ex-post IR constraint. (Because all other zero-valued
      items can be ignored.)

      However, for example \ref{example:bam}, we have shown a BAM with revenue
      $2 + \ln \ln v_{\max} \gg 2$ (as $v_{\max}$ tends to infinity) on a similar
      instance. To adjust the mechanism introduced in Example \ref{example:bam}
      to this example, let $B$ be a BAM as follows,
      \begin{align*}
        & \bal_1 = 0, z_1(\bal_1, v_1) = 1, q_1(\bal_1, v_1) = 1,
          d_1(\bal_1, v_1) = v_1 - 1, s_1(\bal_1) = 0;  \\
        & z_{T' + 1}(\bal_{T' + 1}, v_{T' + 1})
            = \I[v_{T' + 1} \geq v_{\max} / e^{\bal_{T' + 1}}],  \\
        & q_{T' + 1}(\bal_{T' + 1}, v_{T' + 1})
            = z_{T' + 1}(\bal_{T' + 1}, v_{T' + 1})
                \cdot v_{\max} / e^{\bal_{T' + 1}},  \\
        & d_{T' + 1}(\bal_{T' + 1}, v_{T' + 1}) = 0,
          s_{T' + 1}(\bal_{T' + 1}) = \min\{\ln v_{\max}, \bal_{T' + 1}\};  \\
        & \forall t \neq 1, T' + 1,~z_t(\bal_t, v_t) = 0, q_t(\bal_t, v_t) = 0,
          d_t(\bal_t, v_t) = 0, s_t(\bal_t) = 0.  \\
        \Infer &
        \Rev(B) = \E_{v_1,v_{T' + 1}}[q_1(v_1) + q_{T' + 1}(v_{(1, T' + 1)})
                      + s_1(\bal_1) + s_{T' + 1}(\bal_{T' + 1})]  \\
        &~~~~   = 2 + \E_{v_1}\big[\min\{v_1 - 1, \ln v_{\max}\}\big]
                = 2 + \ln \ln v_{\max}.
      \end{align*}

      In other words, the gap between BAM (hence the optimal mechanism) and
      independent local mechanisms can be arbitrarily far on such instances.
    \end{example}

    % In the remainder of this section, we show a simple bank account mechanism
    % that can approximate the optimal mechanism up to a constant factor of $3$,
    % given access to optimal history independent mechanisms.

\section{Constant approximation via deterministic BAM}\label{sec:dapp}

  % Subsection: Constant Approximation of Deterministic BAM

In this section, we prove the existence of deterministic BAM that guarantees
constant approximation of the optimal revenue.

However, since there is an arbitrarily large gap between deterministic and
optimal mechanisms for non-dynamic multidimensional case (even for $2$-item,
$1$-buyer case \cite{hart2013menu}), there is no general approximation
guarantee of deterministic mechanisms. Instead, we prove that if there is a
reasonably good (constant) approximation of the optimal history-independent
mechanism, there is also a constant approximation of the optimal dynamic
mechanism.
\begin{theorem}\label{thm:dbam}
  If there is a deterministic (history-independent) mechanism $M$ that is a
  constant approximation of the optimal history-independent mechanism $\SMA$,
  i.e., $\Rev(M) \geq \alpha \Rev(\SMA)$, then one can construct a
  deterministic BAM $B$ based on $M$, such that $B$ is a constant
  approximation to the revenue of the optimal (possibly randomized) dynamic
  mechanism $B^\OPT$, i.e.,
  \begin{align}\label{eq:dbamapp}
    \Rev(B) \geq \frac{\alpha}{4\alpha + 1} \Rev(B^\OPT).
  \end{align}
  In particular, for the one item per stage case, $\alpha = 1$ (since in this
  case, $\SMA$ is just a posted-price mechanism), and $B$ guarantees a
  $5$-approximation.
\end{theorem}

% \songtodo{If the distributions are irregular, Myerson requires ironing, which
% is randomized. I believe we have $\alpha=1$ only if all distributions are
% regular. In general, I think we have $\alpha=2$ by the result of
% Hartline-Roughgarden, I think, but it is good to check.}
%
% {\edit {\bf Comment:} @Renato, for single-buyer case, Myerson's auction is a
% posted-price auction, where $r_M \in \arg\max_{r} r(1 - \F(r))$.}

\begin{proof}[Proof of Theorem \ref{thm:dbam}]
  First of all, if $\Rev(\SMA) \geq \Rev(B^\OPT) / (4\alpha + 1)$, then
  $M$ is a good approximation, $\Rev(M) \geq \alpha \Rev(\SMA) \geq
  \alpha \Rev(B^\OPT) / (4\alpha + 1)$.

  Otherwise, suppose $\Rev(\SMA) < \Rev(B^\OPT) / (4\alpha + 1)$, then by
  Theorem \ref{thm:bl},
  \begin{align}\label{eq:dbamtob}
    \E_{v_{(1, T)}}\bigg[\sum_{\tau = 1}^T s^\OPT_\tau(\bal^\OPT_\tau)\bigg]
      \geq \frac{4\alpha}{4\alpha + 1} \Rev(B^\OPT).
  \end{align}
  Then we prove that there is a deterministic BAM $B$ that approximates the
  expected overall spend of $B^\OPT$, i.e.,
  \begin{align}
    \E_{v_{(1, T)}}\bigg[\sum_{\tau = 1}^T s_\tau(\bal_\tau)\bigg] \geq
    \frac14
    \E_{v_{(1, T)}}\bigg[\sum_{\tau = 1}^T s^\OPT_\tau(\bal^\OPT_\tau)\bigg]
  \end{align}

  Recall that the randomization over the give-for-free mechanism and the
  posted-price-for-the-grand-bundle mechanism (each with probability $1 / 3$)
  guarantees a $3$-approximation to $\E_{v_{(1, T)}}\big[\ab\sum_{\tau = 1}^T
  s^*_\tau(\bal^*_\tau)\big]$, (see inequality (\ref{eq:3apps3})).

  Similarly the following construction of $B^{\dag}$, a uniform randomization
  of the give-for-free mechanism and the posted-price-for-the-grand-bundle
  mechanism (with posted-price at $r^{\dag}_t$ for each $t$), always spends
  half of the maximum of overall spend.
  \begin{align}
    & z^\dag_t(\bal^\dag_t, v_t) = \frac12 \Big(\One +
                      x^\PM_t\big(2s^\dag_t(\bal^\dag_t), v_t\big)\Big),\quad
    q^\dag_t(\bal^\dag_t, v_t) =
                     \frac12 p^\PM_t\big(2s^\dag_t(\bal^\dag_t), v_t\big),
      \nonumber  \\
    & d^\dag_t(\bal^\dag_t, v_t) = \frac{\One \cdot v_t}2, \quad
    s^\dag_t(\bal^\dag_t) = \min\{\bal^\dag_t, \Val_t/2\},  \nonumber  \\
    \Infer& \forall v_{(1, T)} \in \dV_{(1, T)},~
      \sum_{\tau = 1}^T s^\dag_\tau(\bal^\dag_\tau) =
          \frac12 \sum_{\tau = 1}^T s^*_\tau(\bal^*_\tau).
  \end{align}

  Combining (\ref{eq:3apps2}), $B^{\dag}$ also guarantees a $2$-approximation
  to the expected overall spend of the $B^\OPT$.
  \begin{align}
    \E_{v_{(1, T)}}\bigg[\sum_{\tau = 1}^T s^\dag_\tau(\bal^\dag_\tau)\bigg]
    = \frac12\E_{v_{(1, T)}}\bigg[\sum_{\tau = 1}^T s^*_\tau(\bal^*_\tau)\bigg]
    \geq \frac12
      \E_{v_{(1, T)}}\bigg[\sum_{\tau = 1}^T s^\OPT_\tau(\bal^\OPT_\tau)\bigg].
    \label{eq:mixspdlb}
  \end{align}

  Unfortunately, $B^\dag$ is not deterministic, since it is a mixture of
  two deterministic mechanism, i.e., $z^\dag_t(\bal^\dag_t, v_t)\in\{1/2, 1\}$.
  Then the rest of the proof focuses on the derandomization of $B^\dag$ to get
  a deterministic BAM that guarantees a constant approximation to the expected
  overall spend of $B^\dag$.

  Consider BAM\footnote{We omit the proof that $B^\sgm$ is a BAM, since it is
  almost the same with proving Mechanism \ref{mech:3app} is a BAM in the proof
  of Theorem \ref{thm:3app}.} $B^\sgm$ defined by a binary string $\sgm \in
  \{0, 1\}^T$ as follows,
  \begin{align*}
    \text{if}~\sgm_t = 0:~& z^\sgm_t(\bal^\sgm_t, v_t) = \One,~
                            q^\sgm_t(\bal^\sgm_t, v_t) = 0,~
                            d^\sgm_t(\bal^\sgm_t, v_t) = \One \cdot v_t,~
                            s^\sgm_t(\bal^\sgm_t)      = 0;  \\
    \text{if}~\sgm_t = 1:~&
    z^\sgm_t(\bal^\sgm_t, v_t) = x^\PM_t\big(s^\sgm_t(\bal^\sgm_t), v_t\big),~
    q^\sgm_t(\bal^\sgm_t, v_t) = p^\PM_t\big(s^\sgm_t(\bal^\sgm_t), v_t\big),
      \\ &
    d^\sgm_t(\bal^\sgm_t, v_t) = 0,~
    s^\sgm_t(\bal^\sgm_t)      = \min\{\bal^\sgm_t, \Val_t\}.
  \end{align*}

  Intuitively, if $\sgm_t = 0$, then $B^\sgm$ runs the give-for-free mechanism
  and deposits to the balance; if $\sgm_t = 1$, then $B^\sgm$ spends the
  balance and runs the posted-price-for-the-grand-bundle mechanism with
  parameter equal to the spend $s^\sgm_t(\bal^\sgm_t)$.

  We then argue that for binary string $\sgm$ chosen at uniformly random,
  $\E_{\sgm}[\Rev(B^\sgm)] \geq \Omega(\Rev(B^\OPT))$. Thus at least one out of
  them guarantees a constant approximation to the optimal revenue.

  Consider the following lemma.
  \begin{lemma}\label{lem:pairb}
    \begin{align}\label{eq:pairb}
      \forall t \in [T], v_{(1, T)} \in \dV_{(1, T)},~
        \E_\sgm \Big[s^\sgm_t(\bal^\sgm_t)\Big]
          \geq \frac14 s^*_t(\bal^*_t).
    \end{align}
  \end{lemma}

  Hence
  \begin{align*}
      & \E_{\sgm}\big[\Rev(B^\sgm)\big]
    = \E_{\sgm, v_{(1, T)}}\bigg[\sum_{\tau = 1}^T
          \Big(q^\sgm_\tau(\bal^\sgm_\tau, v_\tau)
               + s^\sgm_\tau(\bal^\sgm_\tau)\Big)\bigg]
    \geq \E_{\sgm, v_{(1, T)}}\bigg[\sum_{\tau = 1}^T
          s^\sgm_\tau(\bal^\sgm_\tau)\bigg]   \\
    \geq~& \frac14 \E_{v_{(1, T)}}\bigg[\sum_{\tau = 1}^T
          s^*_\tau(\bal^*_\tau)\bigg]
    \geq \frac14 \E_{v_{(1, T)}}\bigg[\sum_{\tau = 1}^T
          s^\OPT_\tau(\bal^\OPT_\tau)\bigg],
  \end{align*}
  where the last inequality is from (\ref{eq:3apps2}).

  Therefore there exists one $\sgm^* \in \{0, 1\}^T$, such that
  \begin{align*}
    \Rev(B^{\sgm^*})
      \geq \E_{\sgm}\big[\Rev(B^\sgm)\big]
      \geq \frac14 \E_{v_{(1, T)}}\bigg[\sum_{\tau = 1}^T
                                        s^\OPT_\tau(\bal^\OPT_\tau)\bigg],
  \end{align*}
  and by assumption (\ref{eq:dbamtob}), we conclude that
  \begin{align*}
    \Rev(B^{\sgm^*})
      \geq \frac14 \E_{v_{(1, T)}}\bigg[\sum_{\tau = 1}^T
                                        s^\OPT_\tau(\bal^\OPT_\tau)\bigg]
      \geq \frac{\alpha}{4\alpha + 1}\Rev(B^\OPT).
  \end{align*}

  Note that for any $\sgm$, $B^\sgm$ is a deterministic BAM by construction.
  So in summary, at least one out of these $2^T + 1$ deterministic mechanisms,
  $\{M\} \cup \big\{B^\sgm | \sgm \in \{0, 1\}^T\big\}$, guarantees the desired
  constant approximation to the optimal revenue.
\end{proof}

\section{Computing Optimal Bank Account Mechanism via Dynamic Programming}
\label{sec:fptas}

  In this section, we put forward a dynamic programing algorithm and FPTAS to
  compute the optimal mechanism for the discrete type case. We
  show a proof for the one item per stage case (see Appendix
  \ref{app:thm:dp}), and argue that it can be directly generalized to
  the multi-item per stage case by using the method from
  \cite{cai2012algorithmic}.

  \begin{theorem}\label{thm:dp}
    The optimal bank account mechanism can be computed through a dynamic
    programming algorithm.

    Moreover, for any $\epsilon > 0$, there is an FPTAS to achieve an
    $\epsilon$-approximation (multiplicative) of the optimal revenue for the
    one item per stage case with discrete valuation distributions for each
    stage.
  \end{theorem}

  \begin{corollary}
    For any $\epsilon > 0$, there is an FPTAS to achieve an
    $\epsilon$-approximation (multiplicative) of the optimal revenue for the
    general (multi-item per stage) case with discrete valuation distributions.
  \end{corollary}

  \begin{proof}[Proof sketch of Theorem \ref{thm:dp}]
    We outline the proof idea here and leave the complete proof to Appendix
    \ref{app:thm:dp}.

    By Theorem \ref{thm:rep}, we assume that the optimal mechanism is a core
    BAM $B^{g, y}$. By Theorem \ref{thm:corebam} and Lemma \ref{lem:symmstrong},
    it is without loss of generality to assume that $g_t(h_t) = \Utl(B | h_t)$.
    Since $y_t$ is the sub-gradient of $g_t$ and $g_t$ is symmetric, we denote
    $z_t(\bal_t, v_t)$ as $y_t\big(g_t(v_{(1, t - 1)}), v_t\big)$, by
    overloading the definition of $y_t$.

    We claim that the optimal mechanism can be computed via a dynamic program
    with state function $\phi_t(\xi)$, which is the optimal revenue for the
    sub-problem (consists of stages $t + 1$ to $T$) but with additional
    constraint that the promised utility must equal to $\xi$ $(U_t = \xi)$.

    Formally,
    \begin{align*}
      \phi_t(\xi) = \max_{B : g_t(v_{(1, t)}) = \xi} \E_{v(t + 1, T)}\bigg[
        \sum_{\tau = t + 1}^T y_t\big(g_t(v_{(1, \tau - 1)}), v_\tau\big)
                                \cdot v_\tau  - \bal_{T + 1}\bigg].
    \end{align*}

    It can be computed from $\phi_{t + 1}(\xi)$ by the following program.
    \begin{align}
      \mathrm{maximize}   \quad & \E_{v_t}\Big[
        y_t(\xi, v_t) \cdot v_t + \phi_{t+1}\big(g(\xi, v_t)\big)\Big]
        \label{prog:conv}  \\
      \mathrm{subject~to} \quad & \hu_t(\xi, v_t; v_t) \geq
        \hu_t(\xi, v'_t; v_t),~\forall v_t \in \dV_t \nonumber  \\
      & g(\xi, v_t) = \xi + \hu_t(\xi, v_t; v_t)
          - \E_{v'_t}\big[\hu_t(\xi, v'_t; v'_t)\big] \geq 0,
        ~\forall v_t \in \dV_t \nonumber
    \end{align}

    Then $\max_{\xi \geq 0} \phi_0(\xi) = \max_B \Rev(B)$ is the optimal
    revenue, and function $\phi_0(\xi)$ can be computed via dynamic program
    starting from $\phi_T(\xi) = -\xi$. For any $t$, $\phi_t(\xi)$ is a concave
    function, and once the optimal revenue is determined, we can recover the
    entire mechanism by solving program (\ref{prog:conv}) for each stage.

    For discrete type space cases, program (\ref{prog:conv}) can be written as
    an LP that can be efficiently solved by techniques introduced in
    \cite{cai2012algorithmic}. Moreover, since $\phi_t(\xi)$ is concave, it can
    be bounded by two concave piece-wise linear functions within any constant
    $\delta > 0$. Since the number of pieces of the bounds is at most $O(N /
    \delta)$\footnote{$N$ is the input size}, this dynamic algorithm admits an
    FPTAS.
  \end{proof}
  We conjecture that the computation of the exact solution is likely to be
  hard, because in general cases, the description of the BAM could be
  exponentially large (since there could be exponentially many different values
  of balance).

  In addition, if one wants to further generalize the results for continuous
  valuation distributions by discretization of the valuation domain, one needs
  to relax the IC and IR constraints to $\epsilon$-IC and $\epsilon$-IR, which
  may be an interesting future direction.

% Bibliography
\bibliographystyle{apalike}
\bibliography{acmsmall-sample-bibfile}

\newpage

% Appendix
\appendix
\section*{APPENDIX}
\setcounter{section}{0}

% Appendix

\section{Omitted proofs}

\subsection{Proof of Lemma \ref{lem:swic}}\label{app:lem:swic}
% Proof of Lemma Stage-wise IC
\begin{proof}[Proof of Lemma \ref{lem:swic}]
  {\bf IC $\Infer$ stage-wise IC:}

  Let $b_{(t + 1, T)}$ be the truthful bidding strategy from stage $t + 1$, we
  immediately have IC (\ref{mech:defic}) $\Infer$ stage-wise IC (\ref{mech:ic}).

  {\bf IC $\Reduce$ stage-wise IC:}

  Let $b_t(v_{(1, t)})=v'_t $ be some bidding strategy, and take expectation on
  the stage-wise IC constraint (\ref{mech:ic}) over $v_t$, we have
  \begin{align*}
      \E_{v_t}\big[u_t(v_{(1, t)}; v_t) + U_t(v_{(1, t)})\big]
      \geq \E_{v_t}\Big[u_t\big(v_{(1, t - 1)}, b_t(v_{(1, t)}); v_t\big)
            + U_t\big(v_{(1, t- 1)}, b_t(v_{(1, t)})\big)\Big],
  \end{align*}
  In other words,
  \begin{align*}
    U_{t - 1}(v_{(1, t - 1)}) \geq U^{b_t}_{t - 1}(v_{(1, t - 1)}).
  \end{align*}

  Combining (\ref{mech:ic}):
  \begin{align*}
    u_t(v_{(1, t)}; v_t) + U_t(v_{(1, t)}) \geq
      u_t(v_{(1, t - 1)}, v'_t; v_t) + U_t(v_{(1, t - 1)}, v'_t)
  \end{align*}
  we have
  \begin{align*}
    u_t(v_{(1, t)}; v_t) + U_t(v_{(1, t)}) \geq
      u_t(v_{(1, t - 1)}, v'_t; v_t) + U^{b_{t + 1}}_t(v_{(1, t - 1)}, v'_t).
  \end{align*}

  Take expectation on the last inequality above over $v_t$, we have
  \begin{align*}
    U_{t - 1}(v_{(1, t - 1)}) \geq U^{b_{(t, t + 1)}}_{t - 1}(v_{(1, t - 1)})
  \end{align*}

  Repeat the same argument above, we have
  \begin{align*}
    u_t(v_{(1, t)}; v_t) + U_t(v_{(1, t)}) \geq
      u_t(v_{(1, t - 1)}, v'_t; v_t)
      + U^{b_{(t + 1, t + 2)}}_t(v_{(1, t - 1)}, v'_t)
          \Infer \cdots \Infer \text{(\ref{mech:defic})}.
  \end{align*}

  This proves the equivalence between IC and stage-wise IC.
\end{proof}

\subsection{Proof of Lemma \ref{lem:swir}}\label{app:lem:swir}
% Proof of Lemma Stage-wise IR
\begin{proof}[Proof of Lemma \ref{lem:swir}]
  To show that stage-wise IR implies IR, sum up the stage-wise IR constraints
  (\ref{mech:ir}) from stage $1$ to $T$, we directly have the IR constraint
  (\ref{mech:defir}).
\end{proof}

\subsection{Proof of Lemma \ref{lem:swirconst}}\label{app:lem:swirconst}
% Proof of Lemma Stage-wise IR Construction
\begin{proof}[Proof of Lemma \ref{lem:swirconst}]
  Recall the construction of $M'$ is as follows,
  \begin{itemize}
    \item \begin{align*}
            \forall t \in [T], v_{(1, t)} \in \dV_{(1, t)},~
              x'_t(v_{(1, t)}) = x_t(v_{(1, t)});
          \end{align*}
    \item \begin{align*}
            \forall t < T, v_{(1, t)} \in \dV_{(1, t)},~&
              p'_t(v_{(1, t)}) = x'_t(v_{(1, t)}) \cdot v_t;  \\
            \forall v_{(1, T)} \in \dV_{(1, T)},~&
              p'_T(v_{(1, T)}) = \sum_{\tau = 1}^T p_\tau(v_{(1, \tau)})
                        - \sum_{\tau = 1}^{T - 1} p'_\tau(v_{(1, \tau)}).
          \end{align*}
  \end{itemize}

  We prove that the required properties in Lemma \ref{lem:swirconst} are met.

  Properties \ref{lem:swicir:p1} and \ref{lem:swicir:p2} are straightforward
  from the construction.

  Property \ref{lem:swicir:p3}. For first $T - 1$ stages,
  \begin{align*}
    u'_t(v_{(1, t)}; v_t) = x'_t(v_{(1, t)}) \cdot v_t - p'_t(v_{(1, t)}) = 0
    \Infer \text{(\ref{mech:ir})};
  \end{align*}
  for the last stage,
  \begin{align*}
       u'_T(v_{(1, T)}; v_T)
    &= x'_t(v_{(1, t)}) \cdot v_t - p'_t(v_{(1, t)})  \\
    &= x'_t(v_{(1, t)}) \cdot v_t
          - \bigg(\sum_{\tau = 1}^T p_\tau(v_{(1, \tau)})
                  - \sum_{\tau = 1}^{T - 1} p'_\tau(v_{(1, \tau)})\bigg)  \\
    &= \sum_{\tau = 1}^T x'_\tau(v_{(1, \tau)}) \cdot v_\tau
       - \sum_{\tau = 1}^T p_\tau(v_{(1, \tau)})  \\
    &= \sum_{\tau = 1}^T x_\tau(v_{(1, \tau)}) \cdot v_\tau
       - \sum_{\tau = 1}^T p_\tau(v_{(1, \tau)})  \\
    &= \sum_{\tau = 1}^T u_\tau(v_{(1, \tau)}; v_\tau)  \\
    &\geq 0
       \Infer \text{(\ref{mech:ir})},
  \end{align*}
  where the last inequality is from the IR condition (\ref{mech:defir}) of $M$.

  Property \ref{lem:swicir:p4} is also straightforward. Firstly, rewrite the IC
  condition as follows by including all the utilities from past stages,
  \begin{align*}
    \sum_{\tau = 1}^t u_\tau(v_{(1, \tau)}; v_\tau) + U_t(v_{(1, t)})
      \geq \sum_{\tau = 1}^{t - 1} u_\tau(v_{(1, \tau)}; v_\tau)
           + u_t(v_{(1, t - 1)}, v'_t; v_t)
           + U^{b_{(t + 1, T)}}_t(v_{(1, t - 1)}, v'_t).
  \end{align*}
  Then by properties \ref{lem:swicir:p1} and \ref{lem:swicir:p2}, the
  left-hand-side for $M$ and $M'$ are the same, and similar with the
  right-hand-side.
\end{proof}

% \subsection{Proof of Lemma \ref{lem:swirconstnpt}}\label{app:lem:swirconstnpt}
%
%   To remove, probably.
% \input{prf-lem-swirconstnpt.tex}
%
% \subsection{Proof of Lemma \ref{lem:swicir}}
% \input{prf-lem-swicir.tex}
%
% \subsection{Omitted Details for the Proof of Lemma \ref{lem:swicir}}
% \label{app:lem:swicir}
% \input{prf-lem-swicir-details.tex}
%
% \subsection{Omitted Details for Remark \ref{rem:swir}}\label{app:rem:swir}
% \input{prf-rem-swir.tex}

\subsection{Proof of Lemma \ref{lem:sconic} \& \ref{lem:sconir}}
\label{app:lem:scon}

  Throughout the proof of Lemma \ref{lem:sconic} and \ref{lem:sconir}, we use
  the following notations,
  \begin{align*}
    \bal_\tau = \bal_\tau(v_{(1, \tau - 1)}),~
    \bal^{(t)}_\tau = \bal_\tau(v_{(1, t - 1)}, v'_t, v_{(t + 1, \tau)}),~
    s_\tau = s_\tau(\bal_\tau),~s^{(t)}_\tau = s_\tau(\bal^{(t)}_\tau),
  \end{align*}
  where the superscript ${}^{(t)}$ means the valuation of stage $t$ is replaced
  by some $v'_t$.

  We also enhance both lemmas by generalizing the conclusions to a border set
  of direct mechanisms (beyond the direct mechanism we map the bank account
  mechanism $B$ to).

  For any given bank account mechanism $B$, let $M$ be any direct mechanism
  that has the same allocation rule and generates the same revenue with $B$.
  Formally,
  \begin{align}\label{prf:scon:constr}
    \forall t \in[T], v_{(1, T)} \in \dV_{(1, T)},~&
    \left\{\begin{array}{l}
      x_t(v_{(1, t)}) = z_t(\bal_t, v_t)  \\
      \displaystyle
      \sum_{\tau = 1}^T p_\tau(v_{(1, \tau)}) = \sum_{\tau = 1}^T
                  \Big(q_\tau(\bal_\tau, v_\tau) + s_\tau(\bal_{\tau})\Big)
    \end{array}\right.  \\
    \Infer~&
    \sum_{\tau = 1}^T u_\tau(v_{(1, \tau)}; v_\tau)
    = \sum_{\tau = 1}^T \big(\hu_\tau(\bal_\tau, v_\tau; v_\tau)
                             - s_\tau(\bal_\tau)\big). \label{eq:utlequiv}
  \end{align}

  Then Lemma \ref{lem:sconic} and \ref{lem:sconir} work for $M$ accordingly.

% Proof of Lemma Sufficient Conditions IC

\begin{proof}[Proof of Lemma \ref{lem:sconic}]
  % (has the same allocation
  % and payment rules), i.e.,
  % \begin{align}\label{prf:scon:constr}
  %   \forall t \in[T], v_{(1, T)} \in \dV_{(1, T)},~&
  %   \left\{\begin{array}{l}
  %     x_t(v_{(1, t)}) = z_t(\bal_t, v_t)  \\
  %     \displaystyle
  %     \sum_{\tau = 1}^T p_\tau(v_{(1, \tau)}) = \sum_{\tau = 1}^T
  %                 \Big(q_\tau(\bal_\tau, v_\tau) + s_\tau(\bal_{\tau})\Big)
  %   \end{array}\right.  \\
  %   \Infer~&
  %   \sum_{\tau = 1}^T u_\tau(v_{(1, \tau)}; v_\tau)
  %   = \sum_{\tau = 1}^T \big(\hu_\tau(\bal_\tau, v_\tau; v_\tau)
  %                            - s_\tau(\bal_\tau)\big). \label{eq:utlequiv}
  % \end{align}
  % For such $M$, we say it is the direct mechanism induced by $B$.

  For any $v_{(1, t)}$, consider the following,
  \begin{align*}
    & \sum_{\tau = 1}^{t - 1} u_\tau(v_{(1, \tau)}; v_\tau)
      + u_t(v_{(1, t - 1)}, v'_t; v_t)
      + U_t(v_{(1, t - 1)}, v'_t)  \\
    =~& \sum_{\tau = 1}^{t - 1} u_\tau(v_{(1, \tau)}; v_\tau)
        + u_t(v_{(1, t - 1)}, v'_t; v_t)
        + \E_{v_{(t + 1, T)}}\bigg[\sum_{\tau = t + 1}^T
          u_\tau(v_{(1, t - 1)}, v'_t, v_{(t + 1, \tau)}; v_\tau)\bigg]  \\
    =~& \E_{v_{(t + 1, T)}}\bigg[
          \sum_{\tau = 1}^{t - 1} \Big(\hu_\tau(\bal_\tau, v_\tau; v_\tau)
          - s_\tau\Big)
        + \hu_t(\bal_t, v'_t; v_t) - s_t
        + \sum_{\tau = t + 1}^T \Big(\hu_\tau(\bal^{(t)}_\tau, v_\tau; v_\tau)
          - s^{(t)}_\tau\Big)\bigg],
  \end{align*}
  where the last equation comes from replacing $\sum_{\tau = 1}^T u_\tau$
  according to (\ref{eq:utlequiv}).

  Hence,
  \begin{align*}
      & u_t(v_{(1, t)}; v_t) + U_t(v_{(1, t)})
        - \Big(u_t(v_{(1, t - 1)}, v'_t; v_t)
               + U_t(v_{(1, t - 1)}, v'_t)\Big)  \\
    =~& \sum_{\tau = 1}^{t - 1} u_\tau(v_{(1, \tau)}; v_\tau)
      + u_t(v_{(1, t)}; v_t) + U_t(v_{(1, t)})
      - \bigg(\sum_{\tau = 1}^{t - 1} u_\tau(v_{(1, \tau)}; v_\tau)
             + u_t(v_{(1, t - 1)}, v'_t; v_t)
             + U_t(v_{(1, t - 1)}, v'_t)\bigg)  \\
    =~& \E_{v_{(t + 1, T)}}\bigg[
          \sum_{\tau = 1}^{t - 1} \Big(\hu_\tau(\bal_\tau, v_\tau; v_\tau)
          - s_\tau\Big)
        + \hu_t(\bal_t, v_t; v_t) - s_t
        + \sum_{\tau = t + 1}^T \Big(\hu_\tau(\bal_\tau, v_\tau; v_\tau)
          - s_\tau\Big)\bigg]  \\
      & - \E_{v_{(t + 1, T)}}\bigg[
            \sum_{\tau = 1}^{t - 1} \Big(\hu_\tau(\bal_\tau, v_\tau; v_\tau)
            - s_\tau\Big)
          + \hu_t(\bal_t, v'_t; v_t) - s_t
          + \sum_{\tau = t + 1}^T \Big(\hu_\tau(\bal^{(t)}_\tau, v_\tau; v_\tau)
            - s^{(t)}_\tau\Big)\bigg]  \\
    =~& \E_{v_{(t + 1, T)}}\bigg[
            \hu_t(\bal_t, v_t; v_t) - \hu_t(\bal_t, v'_t; v_t)
          + \sum_{\tau = t + 1}^T \Big(\hu_\tau(\bal_\tau, v_\tau; v_\tau)
                                    - \hu_\tau(\bal^{(t)}_\tau, v_\tau; v_\tau)
          - s_\tau + s^{(t)}_\tau\Big)\bigg]  \\
    =~& \sum_{\tau = t + 1}^{T} \E_{v_{(t + 1, \tau - 1)}}\bigg[
          \E_{v_{\tau}}\Big[\hu_\tau(\bal_\tau, v_\tau; v_\tau)
                            - \hu_\tau(\bal^{(t)}_\tau, v_\tau; v_\tau)\Big]
          - s_\tau + s^{(t)}_\tau\bigg]
        + \hu_t(\bal_t, v_t; v_t) - \hu_t(\bal_t, v'_t; v_t).
    %
    % \stackrel{\text{(\ref{cond:spd})}}{=}
    %     0 + \hu_t(\bal_t, v_t; v_t) - \hu_t(\bal_t, v'_t; v_t)
    % \stackrel{\text{(\ref{cond:sic})}}{\geq} 0 \Infer \text{(\ref{mech:ic})}.
  \end{align*}

  Note that by (\ref{cond:spd}),
  \begin{align*}
    \E_{v_{\tau}}\Big[\hu_\tau(\bal_\tau, v_\tau; v_\tau)
                      - \hu_\tau(\bal^{(t)}_\tau, v_\tau; v_\tau)\Big]
    - s_\tau + s^{(t)}_\tau = 0,
  \end{align*}
  and by (\ref{cond:sic}),
  \begin{align*}
    \hu_t(\bal_t, v_t; v_t) - \hu_t(\bal_t, v'_t; v_t) \geq 0.
  \end{align*}

  Therefore we have
  \begin{align*}
    u_t(v_{(1, t)}; v_t) + U_t(v_{(1, t)})
      - \Big(u_t(v_{(1, t - 1)}, v'_t; v_t) + U_t(v_{(1, t - 1)}, v'_t)\Big)
    \geq 0,
  \end{align*}
  which is exactly the IC constraint (\ref{mech:ic}) for $M$, namely,
  (\ref{cond:sic})(\ref{cond:spd}) $\Infer$ (\ref{mech:ic}).
\end{proof}

% Proof of Lemma Sufficient Conditions IR

\begin{proof}[Proof of Lemma \ref{lem:sconir}]
  % (has the same allocation
  % and payment rules), i.e.,
  % \begin{align}\label{prf:scon:constr}
  %   \forall t \in[T], v_{(1, T)} \in \dV_{(1, T)},~&
  %   \left\{\begin{array}{l}
  %     x_t(v_{(1, t)}) = z_t(\bal_t, v_t)  \\
  %     \displaystyle
  %     \sum_{\tau = 1}^T p_\tau(v_{(1, \tau)}) = \sum_{\tau = 1}^T
  %                 \Big(q_\tau(\bal_\tau, v_\tau) + s_\tau(\bal_{\tau})\Big)
  %   \end{array}\right.  \\
  %   \Infer~&
  %   \sum_{\tau = 1}^T u_\tau(v_{(1, \tau)}; v_\tau)
  %   = \sum_{\tau = 1}^T \big(\hu_\tau(\bal_\tau, v_\tau; v_\tau)
  %                            - s_\tau(\bal_\tau)\big). \label{eq:utlequiv}
  % \end{align}
  % For such $M$, we say it is the direct mechanism induced by $B$.

  From the balance update formula (\ref{eq:defbal}), we have
  \begin{align*}
      \sum_{\tau = 1}^T d_\tau(\bal_\tau, v_\tau)
    = \sum_{\tau = 1}^T \Big(\bal_{\tau + 1} - \bal_\tau
                             + s_\tau(\bal_\tau)\Big)
    = \bal_{T + 1} - \bal_1 + \sum_{\tau = 1}^T s_\tau(\bal_\tau).
  \end{align*}

  Then sum up (\ref{cond:ir}) from $1$ to $T$,
  \begin{align*}
      \sum_{\tau = 1}^T \hu_\tau(\bal_\tau, v_\tau; v_\tau)
    \geq \sum_{\tau = 1}^T d_\tau(\bal_\tau, v_\tau)
    = \bal_{T + 1} - \bal_1 + \sum_{\tau = 1}^T s_\tau(\bal_\tau)
    \geq \sum_{\tau = 1}^T s_\tau(\bal_\tau),
  \end{align*}
  where the last inequality is from the fact that $\bal_1 = 0$ (by definition),
  and $\bal_{T + 1} \geq 0$ ($d_t$ is non-negative, and $s_t$ is no more than
  $\bal_t$). Then by (\ref{eq:utlequiv}),
  \begin{align*}
    \forall v_{(1, T)} \in \dV_{(1, T)},~
    \sum_{\tau = 1}^T u_\tau(v_{(1, \tau)}; v_\tau) =
      \sum_{\tau = 1}^T \hu_\tau(\bal_\tau, v_\tau; v_\tau)
         - \sum_{\tau = 1}^T s_\tau(\bal_\tau) \geq 0
    \Infer \text{(\ref{mech:defir})}.
  \end{align*}

  Thus $M$ is IR. In other words, (\ref{cond:ir}) $\Infer$ (\ref{mech:defir}).

  % To sum up, we show that for any BAM $B$, the induced equivalent direct mechanism $M$ is IC and
  % IR, if the sufficient conditions are satisfied.
\end{proof}

% \subsection{Proof of Lemma \ref{lem:swirimp}}\label{app:lem:swirimp}
% \input{prf-lem-swirimp.tex}

\subsection{Proof of Lemma \ref{lem:mechtrans}}\label{app:lem:mechtrans}
% Proof of lemma mechtrans

\begin{proof}[Proof of Lemma \ref{lem:mechtrans}]
  Without loss of generality\footnote{For $M$ which is not that case, construct
  another $\overline{M}$ with the same allocation rules, but redefine payment
  rules as: for $t < T$, $\overline{p}_t(v_{(1, t)}) = \overline{x}_t(v_{(1,
  t)}) \cdot v_t$; for $t = T$, $\overline{p}_T(v_{(1, T)}) =
  \overline{x}_T(v_{(1, T)}) \cdot v_T - \Utl(M | v_{(1, T)})$. Note that
  $\overline{M}$ is IC and IR, and has the same allocation, overall utility,
  and overall revenue for any $v_{(1, T)}$.}, suppose that
  \begin{align}\label{eq:asspayb}
    \forall t < T, v_{(1, t)} \in \dV_{(1, t)},~u_t(v_{(1, t)}; v_t) = 0.
  \end{align}
  % which then implies that
  % \begin{align}\label{eq:asspaybalt}
  %   \forall v_{(1, t)} \in \dV_{(1, t)},~
  %       \Utl(M | v_{(1, T)}) = u_T(v_{(1, T)}; v_T).
  % \end{align}

  Consider the following iterative procedure that constructs a sequence of
  direct mechanisms, $M^{(1)}, \ldots, M^{(T)}$, where $M^{(t)} = \left\langle
  x^{(t)}_{(1, T)}, p^{(t)}_{(1, T)} \right\rangle$. We start with $M^{(1)} =
  M$, and end up with $M^{(T)} = M'$ as the desirable symmetric direct
  mechanism.

  Given that we have constructed $M^{(1)}, \ldots, M^{(t)}$, we define the
  equivalence relation $\sim_{(t)}$ between any two histories as follows,
  \begin{align}\label{eq:equivrelt}
    v_{(1, t')} \sim_{(t)} v'_{(1, t')} \iff
        \Utl(M^{(t)} | v_{(1, t')}) = \Utl(M^{(t)} | v'_{(1, t')}),
  \end{align}
  and denote the expected revenue of future stages as
  $\pi^{(t)}_{t'}(v_{(1, t')})$. Formally,
  \begin{align*}
    \pi^{(t)}_{t'}(v_{(1, t')}) = \E_{v_{(t' + 1, T)}} \left[
      \sum_{\tau = t' + 1}^T p^{(t)}_\tau(v_{(1, \tau)})\right].
  \end{align*}

  If $T = 1$, $M = M^{(1)} = M'$.

  For $T > 1$, we construct $M^{(t + 1)}$ from $M^{(t)}$ as follows.
  \begin{itemize}
    \item For any $\tau \leq t$, $v_{(1, \tau)} \in \dV_{(1, \tau)}$,
          $M^{(t + 1)}$ is the same with $M^{(t)}$,
          \begin{align}\label{eq:pfdmech1}
            x^{(t + 1)}_\tau(v_{(1, \tau)}) = x^{(t)}_\tau(v_{(1, \tau)}),~
            p^{(t + 1)}_\tau(v_{(1, \tau)}) = p^{(t)}_\tau(v_{(1, \tau)}).
          \end{align}
    \item For any $\tau \geq t + 1$, $v_{(1, \tau)} \in \dV_{(1, \tau)}$,
          $M^{(t + 1)}$ simulates what $M^{(t)}$ does at history $v^*_{(1, t)}$,
          \begin{align}\label{eq:pfdmech2}
            x^{(t + 1)}_\tau(v_{(1, \tau)})
              = x^{(t)}_\tau(v^*_{(1, t)}, v_{(t + 1, \tau)}),~
            p^{(t + 1)}_\tau(v_{(1, \tau)})
              = p^{(t)}_\tau(v^*_{(1, t)}, v_{(t + 1, \tau)}),
          \end{align}
          where $v^*_{(1, t)} = h^*_t(v_{(1, t)}) \in \dV_{(1, t)}$ is the
          selected representative history for the equivalence class $[v_{(1,
          t)}]_{(t)}$ defined by the equivalence relation
          (\ref{eq:equivrelt}), in order to satisfy
          \begin{itemize}
            \item the representative history $v^*_{(1, t)}$ is also in this
                  equivalence class $[v_{(1, t)}]_{(t)}$,
                  \begin{align}\label{eq:mechtutl}
                    \Utl(M^{(t)} | v^*_{(1, t)}) = \Utl(M^{(t)} | v_{(1, t)});
                  \end{align}
            \item the selection of the representative history is unique within
                  each equivalnce class,
                  \begin{align}\label{eq:mechtsym}
                    v_{(1, t)} \sim_{(t)} v'_{(1, t)} \Infer
                      h^*_t(v_{(1, t)}) = h^*_t(v'_{(1, t)});
                  \end{align}
            \item the seller's revenue at the representative history $v^*_{(1,
                  t)}$ is no less than the expected revenue in this equivalence
                  class,
                  \begin{align}\label{eq:mechtrev}
                    \pi^{(t)}_t(v^*_{(1, t)}) \geq
                      \E_{v'_{(1, t)}}\Big[\pi^{(t)}_t(v'_{(1, t)})
                      \big| v'_{(1, t)} \sim_{(t)} v_{(1, t)}\Big].
                  \end{align}
          \end{itemize}
  \end{itemize}

  We finally verify that the so constructed mechanism $M' = M^{(T)}$ satisfies
  the desirable properties by induction: (i) $M'$ is symmetric; (ii) $M'$
  guarantees the same buyer utility as $M$ does, and generates revenue no less
  than $M$ does.

  For the inductive cases, we need to prove the following two properties, where
  (\ref{itm:induc1}) is identical with property (i) when $t = T$, and property
  (ii) is implied if (\ref{itm:induc2}) holds for $t \leq T$.
  \begin{enumerate}
    \item \label{itm:induc1}
          $M^{(t)}$ is symmetric for any pair of equivalent histories of length
          at most $t - 1$, i.e.,
          \begin{align}\label{eq:inductsymm}
            &\forall \tau \leq t - 1, v_{(1, \tau)}, v'_{(1, \tau)} \in
              \dV_{(1, \tau)},~s.t.~v_{(1, \tau)} \sim_{(t)} v'_{(1, \tau)},
            \nonumber \\
            \Infer~& \forall \tau < t' \leq T,
                             v_{(\tau + 1, t')} \in \dV_{(\tau + 1, t')},~
              \left\{\begin{array}{l}
                x^{(t)}_{t'}(v_{(1, t')})
                  = x^{(t)}_{t'}(v'_{(1, \tau)}, v_{(\tau + 1, t')})  \\
                p^{(t)}_{t'}(v_{(1, t')})
                  = p^{(t)}_{t'}(v'_{(1, \tau)}, v_{(\tau + 1, t')})
              \end{array}\right.,
          \end{align}
    \item \label{itm:induc2}
          $\Utl(M^{(t)}) = \Utl(M^{(t - 1)})$, $\Rev(M^{(t)}) \geq
          \Rev(M^{(t - 1)})$.
  \end{enumerate}

  Base case: for $t = 1$, the properties hold trivially.

  Assume by induction that the above two properties (\ref{itm:induc1}) and
  (\ref{itm:induc2}) hold for $t = \frT < T$, we prove that for $t = \frT + 1$.

  {\bf \noindent Property (\ref{itm:induc2}) holds.} We first show that
  $M^{(\frT + 1)}$ guarantees the same buyer utility as $M^{(\frT)}$ does.

  For any history $v_{(1, \frT)} \in \dV_{(1, \frT)}$, by the definition of the
  conditional expected utility,
  \begin{align}\label{eq:prfmcutmp}
      \Utl(M^{(\frT + 1)} | v_{(1, \frT)})
    &= \E_{v_{(\frT + 1, T)}}\left[\sum_{\tau = 1}^T
                  u^{(\frT + 1)}_\tau(v_{(1, \tau)}; v_\tau)\right]
       \nonumber  \\
    &= \sum_{\tau = 1}^{\frT} u_\tau(v_{(1, \tau)}; v_\tau)
         + \E_{v_{(\frT + 1, T)}}\left[\sum_{\tau = \frT + 1}^T
                      u^{(\frT + 1)}_\tau(v_{(1, \tau)}; v_\tau)\right]
       \nonumber  \\
    &= \sum_{\tau = 1}^{\frT} u^{(\frT + 1)}_\tau(v_{(1, \tau)}; v_\tau)
         + U^{(\frT + 1)}_\frT(v_{(1, \frT)}).
  \end{align}

  Recall that by assumption (\ref{eq:asspayb}), for $t < T$,
  \begin{align*}
    x_t(v_{(1, t)}) \cdot v_t - p_t(v_{(1, t)}) = u_t(v_{(1, t)}; v_t) = 0.
  \end{align*}

  According to construction (\ref{eq:pfdmech1}) and (\ref{eq:pfdmech2}), it is
  not hard to see that for $\tau \leq \frT < T$, there exists some $v'_{(1,
  \tau - 1)} \in \dV_{(1, \tau - 1)}$ such that,
  \begin{align*}
    u^{(\frT + 1)}_\tau(v_{(1, \tau)}; v_\tau)
      &= x^{(\frT + 1)}_\tau(v_{(1, \tau)}) \cdot v_\tau
          - p^{(\frT + 1)}_\tau(v_{(1, \tau)})  \\
      &= x_\tau(v'_{(1, \tau - 1)}, v_\tau) \cdot v_\tau
          - p_\tau(v'_{(1, \tau - 1)}, v_\tau) = 0.
  \end{align*}

  Combining (\ref{eq:prfmcutmp}), we have
  \begin{align}\label{eq:prfmcu}
    \Utl(M^{(\frT + 1)} | v_{(1, \frT)}) = U^{(\frT + 1)}_\frT(v_{(1, \frT)}).
  \end{align}

  By construction (\ref{eq:pfdmech2}), for $\tau \geq \frT + 1$,
  $v^*_{(1, \frT)} = h^*_\frT(v_{(1, \frT)})$ being the representative history,
  \begin{align*}
    u^{(\frT + 1)}_\tau(v_{(1, \tau)}; v_\tau)
      &= x^{(\frT + 1)}_\tau(v_{(1, \tau)}) \cdot v_\tau
          - p^{(\frT + 1)}_\tau(v_{(1, \tau)})  \\
      &= x^{(\frT)}_\tau(v^*_{(1, \frT)}, v_{(\frT + 1, \tau)}) \cdot v_\tau
          - p^{(\frT)}_\tau(v^*_{(1, \frT)}, v_{(\frT + 1, \tau)})
       =  u^{(\frT)}_\tau(v^*_{(1, \frT)}, v_{(\frT + 1, \tau)}; v_\tau).
  \end{align*}

  Hence
  \begin{align*}
      U^{(\frT + 1)}_\frT(v_{(1, \frT)})
    &= \E_{v_{(\frT + 1, T)}}\left[\sum_{\tau = \frT + 1}^T
                 u^{(\frT + 1)}_\tau(v_{(1, \tau)}; v_\tau)\right]  \\
    &= \E_{v_{(\frT + 1, T)}}\left[\sum_{\tau = \frT + 1}^T
          u^{(\frT)}_\tau(v^*_{(1, \frT)}, v_{(\frT + 1, \tau)}; v_\tau)\right]
     = U^{(\frT)}_\frT(v^*_{(1, \frT)}).
  \end{align*}

  Similar with (\ref{eq:prfmcu}), we have
  \begin{align*}
    \Utl(M^{(\frT)} | v^*_{(1, \frT)}) = U^{(\frT)}_\frT(v^*_{(1, \frT)}).
  \end{align*}

  Combine (\ref{eq:prfmcu}) and the two equations above,
  \begin{align*}
      \Utl(M^{(\frT + 1)} | v_{(1, \frT)}) = U^{(\frT + 1)}_\frT(v_{(1, \frT)})
      = U^{(\frT)}_\frT(v^*_{(1, \frT)}) = \Utl(M^{(\frT)} | v^*_{(1, \frT)}).
  \end{align*}

  Furthermore, by construction (\ref{eq:mechtutl}), we conclude that
  \begin{align}\label{eq:prfmechutl}
    \forall v_{(1, \frT)} \in \dV_{(1, \frT)},~
      \Utl(M^{(\frT + 1)} | v_{(1, \frT)}) = \Utl(M^{(\frT)} | v^*_{(1, \frT)})
          = \Utl(M^{(\frT)} | v_{(1, \frT)}).
  \end{align}

  Then (\ref{eq:prfmechutl}) implies $\Utl(M^{(\frT + 1)}) = \Utl(M^{(\frT)})$,
  i.e.,
  \begin{align*}
      \Utl(M^{(\frT + 1)})
    = \E_{v_{(1, \frT)}}\left[\Utl\left(M^{(\frT + 1)}
                                        \big| v_{(1, \frT)}\right)\right]
    = \E_{v_{(1, \frT)}}\left[\Utl\left(M^{(\frT)}
                                        \big| v_{(1, \frT)}\right)\right]
    = \Utl(M^{(\frT)}).
  \end{align*}

  We then show that $M^{(\frT + 1)}$ generates revenue no less than $M^{(\frT)}$
  does.

  Recall that by construction (\ref{eq:pfdmech1}), for $\tau \leq \frT$,
  \begin{align*}
    p^{(\frT + 1)}_\tau(v_{(1, \tau)}) = p^{(\frT)}_\tau(v_{(1, \tau)}),
  \end{align*}
  and by construction (\ref{eq:pfdmech2}), for $\tau \geq \frT + 1$,
  \begin{align*}
    p^{(\frT + 1)}_\tau(v_{(1, \tau)})
      = p^{(\frT)}_\tau(v^*_{(1, \frT)}, v_{(\frT + 1, \tau)}).
  \end{align*}

  Then according to the definition of $\pi$, we have
  \begin{align*}
      \pi^{(\frT + 1)}_\frT(v_{(1, \frT)})
    &= \E_{v_{(\frT + 1, T)}}\left[
         \sum_{\tau = \frT + 1}^T p^{(\frT + 1)}_\tau(v_{(1, \tau)})\right]
         \nonumber  \\
    &= \E_{v_{(\frT + 1, T)}}\left[\sum_{\tau = \frT + 1}^T
         p^{(\frT)}_\tau(v^*_{(1, \frT)}, v_{(\frT + 1, \tau)})\right]
     = \pi^{(\frT)}_\frT(v^*_{(1, \frT)}).
  \end{align*}

  Moreover, by definition,
  \begin{align}\label{eq:prfmechrev}
       \Rev(M^{(\frT + 1)})
    &= \E_{v_{(1, T)}}\left[\sum_{\tau = 1}^T
              p^{(\frT + 1)}_\tau(v_{(1, \tau)})\right] \nonumber  \\
    &= \E_{v_{(1, \frT)}}\left[\sum_{\tau = 1}^\frT
              p^{(\frT + 1)}_\tau(v_{(1, \tau)})
          + \E_{v_{(\frT + 1, T)}}\left[\sum_{\tau = \frT + 1}^T
                  p^{(\frT + 1)}_\tau(v_{(1, \tau)})\right]\right] \nonumber  \\
    &= \E_{v_{(1, \frT)}}\left[\sum_{\tau = 1}^\frT
              p^{(\frT + 1)}_\tau(v_{(1, \tau)})
          + \pi^{(\frT + 1)}_\frT(v_{(1, \frT)})\right] \nonumber  \\
    &= \E_{v_{(1, \frT)}}\left[\sum_{\tau = 1}^\frT
              p^{(\frT)}_\tau(v_{(1, \tau)})
          + \pi^{(\frT)}_\frT\left(h^*_\frT(v_{(1, \frT)})\right)\right]
  \end{align}

  For any fixed $v_{(1, \frT)} \in \dV_{(1, \frT)}$, by the property
  (\ref{eq:mechtrev}) of the representative history $v^*_{(1, \frT)} =
  h^*_\frT\left(v_{(1, \frT)}\right)$, we have
  \begin{align}\label{eq:prfmechrevpig}
       \E_{v_{(1, \frT)}}\left[\pi^{(\frT)}_\frT\left(
                               h^*_\frT(v_{(1, \frT)})\right)\right]
    &\geq \E_{v_{(1, \frT)}}\left[
            \E_{v'_{(1, \frT)}}\left[\pi^{(\frT)}_\frT(v'_{(1, \frT)})
                                     \Big| v'_{(1, \frT)} \sim_{(\frT)}
                                           v_{(1, \frT)}\right]\right]
       \nonumber  \\
    &= \E_{v'_{(1, \frT)}}\left[
            \E_{v_{(1, \frT)}}\left[\pi^{(\frT)}_\frT(v'_{(1, \frT)})
                                     \Big| v'_{(1, \frT)} \sim_{(\frT)}
                                           v_{(1, \frT)}\right]\right]
       \nonumber  \\
    &= \E_{v'_{(1, \frT)}}\left[\pi^{(\frT)}_\frT(v'_{(1, \frT)})\right]
     = \E_{v_{(1, \frT)}}\left[\pi^{(\frT)}_\frT(v_{(1, \frT)})\right].
  \end{align}

  By similar argument with (\ref{eq:prfmechrev}),
  \begin{align*}
    \Rev(M^{(\frT)}) = \E_{v_{(1, \frT)}}\left[\sum_{\tau = 1}^\frT
      p^{(\frT)}_\tau(v_{(1, \tau)}) + \pi^{(\frT)}_\frT(v_{(1, \frT)})\right].
  \end{align*}

  Then combining (\ref{eq:prfmechrev}) and (\ref{eq:prfmechrevpig}),
  \begin{align*}
      \Rev(M^{(\frT + 1)}) - \Rev(M^{(\frT)})
    = \E_{v_{(1, \frT)}}\left[\pi^{(\frT)}_\frT\left(
                            h^*_\frT(v_{(1, \frT)})\right)\right]
      - \E_{v_{(1, \frT)}}\left[\pi^{(\frT)}_\frT(v_{(1, \frT)})\right]
    \geq 0.
  \end{align*}

  In other words, $\Rev(M^{(\frT + 1)}) \geq \Rev(M^{(\frT)})$.

  {\bf \noindent Proterty (\ref{itm:induc1}) holds.}

  Since the goal (\ref{eq:inductsymm}) has mathematically identical formulation
  on $x$ and $p$, we prove (\ref{eq:inductsymm}) for $x$ only.

  Consider the following cases:
  \begin{itemize}
    \item For $t' \leq \frT$, it follows from the construction
          (\ref{eq:pfdmech1}) and the inductive assumption (\ref{eq:inductsymm})
          as follows,
          \begin{align*}
            & \forall \tau < t' \leq \frT, v_{(1, t')} \in \dV_{(1, t')},
                      v'_{(1, \tau)} \sim_{(\frT)} v_{(1, \tau)},  \\
            & x^{(\frT + 1)}_{t'}(v_{(1, t')}) = x^{(\frT)}_{t'}(v_{(1, t')})
                = x^{(\frT)}_{t'}(v'_{(1, \tau)}, v_{(\tau + 1, t')})
                = x^{(\frT + 1)}_{t'}(v'_{(1, \tau)}, v_{(\tau + 1, t')}).
          \end{align*}
          Note that the second ``$=$'' comes from the inductive assumption
          (\ref{eq:inductsymm}) on $t = \frT$, where $\tau \leq t - 1 = \frT -
          1$ is required.
    \item For $t' > \frT$ and $\tau < \frT$, we show that $v_{(1, \tau)}
          \sim_{(\frT + 1)} v'_{(1, \tau)}$ implies
          \begin{align}\label{eq:prfmechcaseredg}
            v_{(1, \frT)} \sim_{(\frT + 1)}
                  \left(v'_{(1, \tau)}, v_{(\tau + 1, \frT)}\right),
          \end{align}
          and it reduces to the following $\tau = \frT$ case then.

          Recall by (\ref{eq:prfmechutl}), $\Utl(M^{(\frT + 1)} | v_{(1, \frT)})
          = \Utl(M^{(\frT)} | v_{(1, \frT)})$. Hence by the definition of
          conditional expected utility,
          \begin{align*}
            \Utl(M^{(\frT + 1)} | v_{(1, \tau)})
              &= \E_{v_{(\tau + 1, \frT)}}\left[
                   \Utl(M^{(\frT + 1)} | v_{(1, \frT)})\right]  \\
              &= \E_{v_{(\tau + 1, \frT)}}\left[
                   \Utl(M^{(\frT)} | v_{(1, \frT)})\right]
               = \Utl(M^{(\frT)} | v_{(1, \tau)}).
          \end{align*}
          Then $v_{(1, \tau)} \sim_{(\frT + 1)} v'_{(1, \tau)}$ implies
          $v_{(1, \tau)} \sim_{(\frT)} v'_{(1, \tau)}$, since
          \begin{align*}
            \Utl(M^{(\frT)} | v_{(1, \tau)})
              &= \Utl(M^{(\frT + 1)} | v_{(1, \tau)})  \\
              &= \Utl(M^{(\frT + 1)} | v'_{(1, \tau)})
               = \Utl(M^{(\frT)} | v'_{(1, \tau)}).
          \end{align*}
          By inductive assumption (\ref{eq:inductsymm}) (on $t = \frT$), the
          submechanisms of $M^{(\frT)}$ at history $v_{(1, \tau)}$ and $v'_{(1,
          \tau)}$ are identical\footnote{Note that the demonstration based on
          the inductive assumption on $t = \frT$ only works for $\tau \leq t -
          1 = \frT - 1$, which is certainly true when the integer $\tau <
          \frT$.}. Therefore\footnote{By assumption (\ref{eq:asspayb}) and
          construction (\ref{eq:pfdmech1}) and (\ref{eq:pfdmech2}),
          $\Utl(M^{(\frT)} | v_{(1, T)}) = u^{(\frT)}_T(v_{(1, T)}; v_T)$,
          which equals to $u^{(\frT)}_T(v'_{(1, \tau)}, v_{(\tau + 1, T)}; v_T)$
          because these two submechanisms are identical.}, for any
          $v_{(\tau + 1, T)} \in \dV_{(\tau + 1, T)}$,
          \begin{align*}
            \Utl(M^{(\frT)} | v_{(1, T)})
              = \Utl(M^{(\frT)} | v'_{(1, \tau)}, v_{(\tau + 1, T)}).
          \end{align*}
          It then suggests that $v_{(1, \frT)}$ and $(v'_{(1, \tau)},
          v_{(\tau + 1, \frT)})$ are equivalent with respect to the equivalence
          relation $\sim_{(\frT)}$, i.e.,
          \begin{align*}
              \Utl(M^{(\frT)} | v_{(1, \frT)})
            &= \E_{v_{(\frT + 1, T)}}\left[\Utl(M^{(\frT)} | v_{(1, T)})\right]
              \\
            &= \E_{v_{(\frT + 1, T)}}\left[\Utl(M^{(\frT)} |
                                      v'_{(1, \tau)}, v_{(\tau + 1, T)})\right]
             = \Utl(M^{(\frT)} | v'_{(1, \tau)}, v_{(\tau + 1, T)}).
          \end{align*}
          Again, by (\ref{eq:prfmechutl}),
          \begin{align*}
            \Utl(M^{(\frT + 1)} | v_{(1, \frT)})
              &= \Utl(M^{(\frT)} | v_{(1, \frT)})  \\
              &= \Utl(M^{(\frT)} | v'_{(1, \tau)}, v_{(\tau + 1, T)})
               = \Utl(M^{(\frT + 1)} | v'_{(1, \tau)}, v_{(\tau + 1, T)}).
          \end{align*}
          Hence we have proved (\ref{eq:prfmechcaseredg}), and it reduces to
          the following $\tau = \frT$ case.
    \item For $t' > \frT$ and $\tau = \frT$, $\forall v_{(1, \frT)}
          \sim_{(\frT + 1)} v'_{(1, \frT)}$, by (\ref{eq:prfmechutl}), we have
          $v_{(1, \frT)} \sim_{(\frT)} v'_{(1, \frT)}$.

          Then by the uniqueness (\ref{eq:mechtsym}) of the representative
          history, we have
          \begin{align*}
            h^*_\frT(v_{(1, \frT)}) = h^*_\frT(v'_{(1, \frT)})
                                    = v^*_{(1, \frT)}.
          \end{align*}
          According to construction (\ref{eq:pfdmech2}), for any $v_{(\frT + 1,
          t')} \in \dV_{(\frT + 1, t')}$,
          \begin{align*}
              x^{(\frT + 1)}_{t'}(v'_{(1, \frT)}, v_{(\frT + 1, t')})
            = x^{(\frT)}_{t'}(v^*_{(1, \frT)}, v_{(\frT + 1, t')})
            = x^{(\frT + 1)}_{t'}(v_{(1, \frT)}, v_{(\frT + 1, t')}).
          \end{align*}
  \end{itemize}

  Therefore, we have finished the proof of the inductive properties
  (\ref{itm:induc1}) and (\ref{itm:induc2}) for $t = \frT + 1$.

  Hence $M^{(\frT + 1)}$ is symmetric for any pair of equivalent histories of
  length at most $\frT$. Therefore, by induction, $M' = M^{(T)}$ is symmetric
  and,
  \begin{align*}
    & \Utl(M') = \Utl(M^{(T)}) = \cdots = \Utl(M^{(1)}) = \Utl(M)  \\
    & \Rev(M') \geq \Rev(M^{(T)}) \geq \cdots \geq \Rev(M^{(1)}) \geq \Rev(M).
  \end{align*}
  
  In particular, by (\ref{eq:pfdmech1}) and (\ref{eq:pfdmech2}), if $M$ is
  deterministic, $M'$ is also deterministic.
\end{proof}

\subsection{Proof of Lemma \ref{lem:symmstrong}}\label{app:lem:symmstrong}
% Proof of lemma symm

\begin{proof}[Proof of Lemma \ref{lem:symmstrong}]
  % Note that $M$ is symmetric, therefore, by (\ref{eq:bbal}), the construction
  % of $z_t$ is consistent on equivalent histories that lead to the same balance
  % $\bal_t$. In other words, $z_t$ is a valid function.
  %
  To prove that $B^{g, y}$ is a BAM, we must verify the following,
  \begin{itemize}
    \item $B^{g, y}$ is a valid bank account mechanism:
          \begin{itemize}
            \item $B^{g, y} = \langle z_{(1, T)}, q_{(1, T)}, \bal_{(1, T)},
                  d_{(1, T)}, s_{(1, T)} \rangle$ are defined on their domains.
            \item $q_t$ and $d_t$ are always non-negative.
            \item $s_t(\bal_t)$ is independent with $v_t$, and no more than
                  $\bal_t$.
            \item The balance is updated according to formula
                  (\ref{eq:defbal}).
          \end{itemize}
    \item $B^{g, y}$ is a BAM: IC and IR constraints
          (\ref{cond:sic})-(\ref{cond:ir}) are satisfied.
  \end{itemize}

  We now verify these statements one by one.
  % By rewriting equation (\ref{mech:ic}), we have
  % \begin{align*}
  %   x_t(v_{(1, t)})v_{t} - \big(p_t(v_{(1, t)}) - U_t(v_{(1, t)})\big)
  %   \geq x_t(v_{(1, t - 1)}, v'_t)v_{t} - \big(p_t(v_{(1, t - 1)}, v'_t)
  %        - U_t(v_{(1, t - 1)}, v'_t)\big),
  % \end{align*}
  % which can be considered as a IC condition for a one-shot static direct
  % mechanism with allocation rule $x_t(v_{(1, t - 1)}, \bcdot)$ and payment rule
  % $p_t(v_{(1, t - 1)}, \bcdot) - U_t(v_{(1, t - 1)}, \bcdot)$. Then by Envelope
  % theorem,
  % \begin{align}\label{eq:envic}
  %   \frac{\partial \big(u_t(v_{(1, t)}; v_t) + U_t(v_{(1, t)})\big)}
  %        {\partial v_t}
  %   = \frac{\partial \Big(x_t(v_{(1, t)})v_{t} -
  %                         \big(p_t(v_{(1, t)}) - U_t(v_{(1, t)})\big)\Big)}
  %          {\partial v_t}
  %   = x_t(v_{(1, t)}).
  % \end{align}
  \begin{itemize}
    \item Since $M$ is symmetric, then by definition,
          \begin{align*}
            \forall t \in [T], v_{(1, t)}, v'_{(1, t)} \in \dV_{(1, t)},~
            \Utl(M | v_{(1, t - 1)}) = \Utl(M | v'_{(1, t - 1)})
                \Infer x_t(v_{(1, t)}) = x_t(v'_{(1, t - 1)}, v_t).
          \end{align*}

          Combining with the construction of $\bal_t$ and $z_t$ (see
          (\ref{eq:bbal}) and (\ref{eq:ball})), we have,
          \begin{align*}
                  &~\bal_t(v_{(1, t - 1)}) = \bal'_t(v_{(1, t - 1)})  \\
            \Infer&~
              \Utl(M | v_{(1, t - 1)}) = \bal_t(v_{(1, t - 1)}) + \mu_{t - 1}
                                       = \bal_t(v'_{(1, t - 1)}) + \mu_{t - 1}
                                       = \Utl(M | v'_{(1, t - 1)})  \\
            \Infer&~
              z_t\left(\bal_t(v_{(1, t - 1)}), v_t\right)
                  = x_t(v_{(1, t)})
                  = x_t(v'_{(1, t - 1)}, v_t)
                  = z_t\left(\bal_t(v'_{(1, t - 1)}), v_t\right),
          \end{align*}
          which then suggests that $z_t(\bal_t, v_t)$ is a valid function,
          since for any input $(\bal_t, v_t)$, $z_t$ is never assigned two
          different values.

          Then it is also true for $q_t$ and $d_t$, since they are constructed
          based on $z_t$ (see (\ref{eq:bpay}) and (\ref{eq:bdps})).

          Also note that by construction (\ref{eq:bbal}), $\bal_t$ is always
          non-negative. Hence $z_t$, $q_t$, $d_t$, and $s_t$ are defined on
          their domains (the first parameter is the balance --- a non-negative
          real, and the second parameter is the value $v_t$).
    \item By the construction of $z_t$, (\ref{eq:ball}), $z_t(\bal_t, v_t) =
          y_t(v_{(1, t)})$. Then by the statement (\ref{eq:lemsymmstat}) in
          Lemma \ref{lem:symmstrong}, $z_t(\bal_t, v_t) = y_t(v_{(1, t)}) =
          x_t(v_{(1, t)})$. Since the symmetric direct mechanism $M$ satisfies
          IC condition (\ref{mech:ic}), $x_t$ must be (weakly) increasing in
          $v_t$.

          By the construction of $q_t$, (\ref{eq:bpay}), we have
          \begin{align*}
            q_t(\bal_t, v_t) = z_t(\bal_t, v_t) \cdot v_t
                                - \int_{\Zero}^{v_t} z_t(\bal_t, v) \ud v
                             \geq 0.
          \end{align*}

          By the construction of $d_t$, (\ref{eq:bdps}), we have
          \begin{align*}
            d_t(\bal_t, v_t) = \hu_t(\bal_t, v_t; v_t)
                             = \int_{\Zero}^{v_t} z_t(\bal_t, v) \ud v
                             \geq 0.
          \end{align*}

          Hence both of $q_t$ and $d_t$ are always non-negative.
    \item We claim that according to the construction of $B^{g, y}$, $M$ is IC
          implies $\partial s_t(\bal_t) / \partial v_t = \Zero$.

          Formally, from the construction of $s_t$ (\ref{eq:bspd}),
          \begin{align*}
            s_t(\bal_t) = \bal_t + d_t(\bal_t, v_t) - \bal_{t + 1}.
          \end{align*}
          By plugging the construction of $d_t$ and $\bal_{t + 1}$ into the
          formula above, we have
          \begin{align*}
            s_t(\bal_t) &= \bal_t + \hu_t(\bal_t, v_t; v_t)
                             - g_t(v_{(1, t)}) + \mu_t  \\
                        &= \bal_t + \hu_t(\bal_t, v_t; v_t)
                             - \left(\Utl(M | v_{(1, t)}) - \mu_t\right),
          \end{align*}
          where the second ``$=$'' is from the statement (\ref{eq:lemsymmstat}).

          Then by taking partial derivative with respect to $v_t$ on both
          sides, we have
          \begin{align}\label{eq:prfsymmspdpd}
              \frac{\partial s_t(\bal_t)}{\partial v_t}
            = \frac{\partial \bal_t}{\partial v_t}
              + \frac{\partial \hu_t(\bal_t, v_t; v_t)}{\partial v_t}
              - \frac{\partial \left(\Utl(M | v_{(1, t)}) - \mu_t\right)}
                     {\partial v_t}.
          \end{align}

          By the construction of $\bal_t(v_{(1, t - 1)})$, (\ref{eq:bbal}),
          $\bal_t$ is independent with $v_t$, hence $\partial \bal_t / \partial
          v_t = \Zero$.

          Since $\hu_t(\bal_t, v_t; v_t) = \int_{\Zero}^{v_t} z_t(\bal_t, v)
          \ud v$, $\partial \hu_t(\bal_t, v_t; v_t) / \partial v_t =
          z_t(\bal_t, v_t)$.

          Note that for $\tau \leq t - 1$, $u_\tau(\bal_\tau, v_\tau; v_\tau)$
          is independent with $v_t$, and $\mu_t = \inf_{v'_{(1, t)}}
          g_t(v'_{(1, t)})$ is also independent with $v_t$. Hence
          \begin{align*}
              \frac{\partial \left(\Utl(M | v_{(1, t)}) - \mu_t\right)}
                   {\partial v_t}
            &=\frac{\partial \left(\sum_{\tau = 1}^t
                                      u_\tau(\bal_\tau, v_\tau; v_\tau)
                                    + U_t(v_{(1, t)}) - \mu_t\right)}
                   {\partial v_t}  \\
            &=\frac{\partial \left(u_t(\bal_t, v_t; v_t)
                                    + U_t(v_{(1, t)})\right)}{\partial v_t},
          \end{align*}
          which then by (\ref{eq:envic}) equals to $x_t(v_{(1, t)})$.

          Now (\ref{eq:prfsymmspdpd}) reduces to the follows,
          \begin{align*}
              \frac{\partial s_t(\bal_t)}{\partial v_t}
            &=\frac{\partial \bal_t}{\partial v_t}
              + \frac{\partial \hu_t(\bal_t, v_t; v_t)}{\partial v_t}
              - \frac{\partial \left(\Utl(M | v_{(1, t)}) - \mu_t\right)}
                   {\partial v_t}  \\
            &=\Zero + z_t(\bal_t, v_t) - x_t(v_{(1, t)}) = \Zero.
          \end{align*}
          Recall that by construction, $z_t(\bal_t, v_t) = x_t(v_{(1, t)})$.

          % \begin{align*}
          %     \frac{\partial s_t(\bal_t)}{\partial v_t}
          %   \stackrel{\text{(\ref{eq:bbal})(\ref{eq:bspd})}}{=}~&
          %     \frac{\partial \bal_t}{\partial v_t}
          %       + \frac{\partial \hu_t(\bal_t, v_t; v_t)}{\partial v_t}
          %       - \frac{\partial \big(\Utl(M | v_{(1, t)}) - \mu_t\big)}
          %              {\partial v_t}  \\
          %   \stackrel{\text{(\ref{eq:bpay})(\ref{eq:bdps})}}{=}~&
          %     \Zero + z_t(\bal_t, v_t)
          %       - \frac{\partial \big(u_t(v_{(1, t)}; v_t)
          %                             + U_t(v_{(1, t)})\big)}{\partial v_t}
          %       \\
          %   \stackrel{\text{(\ref{eq:envic})}}{=}~~~~&
          %         z_t(\bal_t, v_t) - x_t(v_{(1, t)})
          %   \stackrel{\text{(\ref{eq:ball})}}{=} \Zero.
          % \end{align*}
          Hence $s_t$ constructed as (\ref{eq:bspd}) is a function of $\bal_t$.
          By letting $v_t = \Zero$, we show that $s_t(\bal_t)$ is no more than
          $\bal_t$,
          \begin{align*}
            s_t(\bal_t) = \bal_t + d_t(\bal_t, \Zero)
                              - \bal_{t + 1}(\bal_t, \Zero)
                        \leq \bal_t + d_t(\bal_t, \Zero) = \bal_t,
          \end{align*}
          where $d_t(\bal_t, \Zero) = \hu_t(\bal_t, \Zero; \Zero) = 0$.
    \item The balance update formula (\ref{eq:defbal}) is directly implied by
          the construction of $s_t$, (\ref{eq:bspd}).
    \item As we mentioned previously, $z_t(\bal_t, v_t) = x_t(v_{(1, t)})$ is
          weakly increasing in $v_t$. Then by Envelope theorem, the construction
          of $q_t$, (\ref{eq:bpay}), implies the in-stage IC condition
          (\ref{cond:sic}).
    \item According to the construction of $B^{g, y}$, the fact that
          $s_t(\bal_t)$ is independent with $v_t$ implies another IC condition
          (\ref{cond:spd}). Formally, by adopting the construction of $s_t$,
          $\bal_t$, and $d_t$, we have
          \begin{align*}
               s_t(\bal_t)
            &= \bal_t + d_t(\bal_t, v_t) - \bal_{t + 1}  \\
            &= \Utl(M | v_{(1, t - 1)}) - \mu_{t - 1} + \hu_t(\bal_t, v_t; v_t)
               - \left(\Utl(M | v_{(1, t)}) - \mu_t\right)  \\
            &= \Utl(M | v_{(1, t - 1)}) - \Utl(M | v_{(1, t)})
               + \hu_t(\bal_t, v_t; v_t) - \mu_{t - 1} + \mu_t.
          \end{align*}

          Taking expectation over $v_t$ on both sides, we have
          \begin{align*}
               s_t(\bal_t)
            &= \E_{v_t}\left[s_t(\bal_t)\right]  \\
            &= \E_{v_t} \left[\Utl(M | v_{(1, t - 1)}) - \Utl(M | v_{(1, t)})
                        + \hu_t(\bal_t, v_t; v_t) - \mu_{t - 1} + \mu_t\right]
               \\
            &= \E_{v_t} \big[\hu_t(\bal_t, v_t; v_t)\big] - \mu_{t - 1} + \mu_t,
          \end{align*}
          which implies the IC condition (\ref{cond:spd}), because $\mu_{t - 1}$
          and $\mu_t$ are two constants with respect to $\bal_t$ and
          $v_{(1, t)}$.

          Note that the first ``$=$'' comes from the fact that $s_t(\bal_t)$ is
          independent with $v_t$. For the second ``$=$'', we use the fact that
          \begin{align*}
               \E_{v_t}\left[\Utl(M | v_{(1, t)})\right]
            &= \E_{v_t}\left[\E_{v_{(t + 1, T)}}
                              \left[\Utl(M | v_{(1, T)})\right]\right]  \\
            &= \E_{v_{(t, T)}}\left[\Utl(M | v_{(1, T)})\right]  \\
            &= \Utl(M | v_{(1, t - 1)}).
          \end{align*}
    \item Finally, the construction of $d_t$, (\ref{eq:bdps}), directly implies
          the IR condition (\ref{cond:ir}).
  \end{itemize}

  Note that $M$ and $B^{g, y}$ are equivalent in allocation rule,
  $x_t(v_{(1, t)}) = z_t(\bal_t, v_t)$. Hence $\Eff(B^{g, y}) = \Eff(M)$.
  Then we show that they are also equivalent on buyer utility to conclude that
  they are also equivalent on revenue.

  Recall that by definition,
  \begin{align*}
      \Utl(B^{g, y} | v_{(1, T)})
    &=\sum_{\tau = 1}^T z_\tau(\bal_\tau, v_\tau) \cdot v_\tau
        - \sum_{\tau = 1}^T q_t(\bal_\tau, v_\tau)
        - \sum_{\tau = 1}^T s_\tau(\bal_\tau)  \\
    &=\sum_{\tau = 1}^T \Big(\big(z_\tau(\bal_\tau, v_\tau) \cdot v_\tau
                                  - q_t(\bal_\tau, v_\tau)\big)
                             - s_\tau(\bal_\tau)\Big)  \\
    &=\sum_{\tau = 1}^T \big(\hu_\tau(\bal_\tau, v_\tau; v_\tau)
                             - s_\tau(\bal_\tau)\big).
  \end{align*}

  Furthermore, by the construction of $d_t$ and $s_t$ (see (\ref{eq:bdps}) and
  (\ref{eq:bspd})), we have
  \begin{align*}
      \Utl(B^{g, y} | v_{(1, T)})
    &=\sum_{\tau = 1}^T \big(\hu_\tau(\bal_\tau, v_\tau; v_\tau)
                             - s_\tau(\bal_\tau)\big)  \\
    &=\sum_{\tau = 1}^T \Big(d_\tau(\bal_\tau, v_\tau)
                             - \big(\bal_\tau + d_\tau(\bal_\tau, v_\tau)
                                    - \bal_{\tau + 1}\big)\Big)  \\
    &=\sum_{\tau = 1}^T \big(\bal_{\tau + 1} - \bal_\tau\big)  \\
    &=\bal_{T + 1} - \bal_1.
  \end{align*}

  Note that by definition, $\bal_1 = 0$, and
  \begin{align*}
    \bal_{T + 1} = \Utl(M | v_{(1, T)}) - \mu_T
                 = \Utl(M | v_{(1, T)})
                   - \min\left\{0, \inf_{v'_{(1, T)}}
                                      \Utl(M | v'_{(1, T)})\right\}
                 = \Utl(M | v_{(1, T)}),
  \end{align*}
  where $\mu_T = 0$, because $\Utl(M | v'_{(1, T)}) \geq 0$ by the IR condition
  (\ref{mech:defir}).

  Therefore
  \begin{align*}
    \Utl(B^{g, y} | v_{(1, T)}) = \bal_{T + 1} - \bal_1 = \Utl(M | v_{(1, T)}),
  \end{align*}
  which then implies
  \begin{align*}
    & \Utl(B^{g, y}) = \Utl(M),  \\
    & \Rev(B^{g, y}) = \Rev(M).
  \end{align*}
\end{proof}
  %
  % Combining with the follows, we claim that they are also equivalent in payment
  % rule, therefore $\Utl(B^{g, y}) = \Utl(M)$ and $\Rev(B^{g, y}) = \Rev(M)$.
  % \begin{align*}
  %   \forall v_{(1, T)} \in \dV_{(1, T)},~
  %   \Utl(B^{g, y} | v_{(1, T)})
  %     &~= \sum_{\tau = 1}^T\big(\hu_\tau(\bal_\tau, v_\tau; v_\tau)
  %                             - s_\tau(\bal_\tau)\big)
  %       = \sum_{\tau = 1}^T (\bal_{\tau + 1} - \bal_{\tau})  \\
  %     &~= \bal_{T + 1} - \bal_1
  %       = \Utl(M | v_{(1, T)}) - \mu_T = \Utl(M | v_{(1, T)}),
  % \end{align*}
  % where by defintion, $\mu_T = \min\big\{0, \inf_{v_{(1, T)}} \Utl(M |
  % v_{(1, T)})\big\} = 0$, since $\Utl(M | v_{(1, T)}) \geq 0$, by
  % (\ref{mech:defir}).

\subsection{Proof of Theorem \ref{thm:corebam}}\label{app:thm:corebam}
% Proof of Theorem corebam

Before conducting the proof, we first give some intuitions behind the theorem
as follows,
% The intuitions behind are as follows,
\begin{itemize}
  \item $g_t$ is convex in $v_t$ and weakly increasing in $v_t$: then
        $g_t(v_{(1, t - 1)}, v_t)$ is the utility function of some IC and IR
        stage mechanism for stage $v_t$, by the Envelope theorem.
  \item $g_t$ is consistent: then $g_t(v_{(1, t)})$ can be considered as the
        conditional utility $\Utl(M | v_{(1, t)})$ plus a constant for some
        direct mechanism $M$, i.e.,
        \begin{align*}
          g_t(v_{(1, t)}) = \Utl(M | v_{(1, t)}) + \mathrm{constant}.
        \end{align*}
  \item If $g_t$ is further symmetric: then the direct mechanism $M$ above is
        a symmetric direct mechanism.
  \item Finally, if $y_t$ is the sub-gradient of $g_t$ with respect to $v_t$
        and its range is $\dX_t$: then $y_t$ is exactly the allocation function
        of the $M$.
\end{itemize}

In other words, $B^{g, y}$ is a core BAM, if and only if there is a symmetric
direct mechanism $M$ such that $g$ equals to its conditional utility plus some
constants (one for each stage), and $y$ is its allocation function.

The proof of Theorem \ref{thm:corebam} is similar to the proof of Lemma
\ref{lem:symmstrong}. % We refer readers to Appendix for the proof.
The following corollary summarizes the important conclusions we have until the
end of Section \ref{sec:bam}.
\begin{corollary}\label{cor:wlog}
  Let $M$ be a revenue-optimal direct mechanism, subject to IC, IR, and with
  fixed buyer utility. One can construct a symmetric (Lemma
  \ref{lem:mechtrans}), stage-wise IC and stage-wise IR (Lemma \ref{lem:swic}
  to \ref{lem:swirconst}) mechanism $M'$ that yields the same revenue and
  guarantees the same buyer utility. Furthermore, $M'$ can be induced by a Core
  BAM (Theorem \ref{thm:rep}, Lemma \ref{lem:sconic} and \ref{lem:sconir}).
\end{corollary}
  % (Recall that $M$ induced by $B$ is any direct mechanism that implements
  % $B$.)
  % Moreover, if the utility of the lowest type is $0$, i.e.,
  % $\min_{v_{(1, T)}}\Utl(M | v_{(1, T)}) = 0$, it is WLOG to further assume
  % that $M$ is stage-wise NPT.
% As a result, it is without loss of generality to study BAMs and core BAMs
% in the rest of the paper.

\begin{proof}[Proof of Theorem \ref{thm:corebam}]
  ``$\Infer$'': $y$ is the sub-gradient of $g$ and $g$ is consistent,
  symmetric, and convex and increasing in $v_t$, $\Infer$ $B^{g, y}$ is a Core
  BAM.

  By Definition \ref{def:corebam}, we need to show that $B^{g, y}$ is a valid
  BAM. The following proof is almost the same with the proof of Lemma
  \ref{lem:symmstrong}, so we only outline the differences here.
  \begin{itemize}
    \item Since $g$ is symmetric,
          \begin{align*}
            g_{t - 1}(v_{(1, t - 1)}) = g_{t - 1}(v'_{(1, t - 1)})
              \Infer g_t(v_{(1, t)}) = g_t(v'_{(1, t - 1)}, v_t),
          \end{align*}
          and $y_t$ is the sub-gradient of $g_t$ with respect to $v_t$, so by
          the construction (\ref{eq:bbal}) of $\bal_t$
          \begin{align*}
            \bal_t(v_{(1, t - 1)}) = \bal_t(v'_{(1, t - 1)})
              & \Infer g_{t - 1}(v_{(1, t - 1)}) = g_{t - 1}(v'_{(1, t - 1)})
                \\
              & \Infer y_t(v_{(1, t)}) = \frac{\partial g_t(v_{(1, t)})}
                                       {\partial v_t}
                                     = \frac{\partial g_t(v'_{(1, t - 1)}, v_t)}
                                       {\partial v_t}
                                     = y_t(v'_{(1, t - 1)}, v_t),
          \end{align*}
          which then implies $z_t(\bal_t, v_t) = y_t(v_{(1, t)})$ is a valid
          function on input $(\bal_t, v_t)$. Also note that the range of $y_t$
          (hence the range of $z_t(\bal_t, v_t)$) is $\dX_t$ by definition.
    \item $z_t$, $q_t$, $d_t$, and $s_t$ are defined on their domains; $q_t$
          and $d_t$ are always non-negative: the same with the proof of Lemma
          \ref{lem:symmstrong}.
    \item By the construction of $s_t$, $\bal_t$, $d_t$, and $z_t$ (see
          (\ref{eq:bspd}), (\ref{eq:bbal}), (\ref{eq:bdps}) and
          (\ref{eq:ball})),
          \begin{align*}
            s_t(\bal_t) &= \bal_t + d_t(\bal_t, v_t) - \bal_{t + 1}  \\
                        &= g_{t - 1}(v_{(1, t - 1)}) - \mu_{t - 1}
                           + \int_{\Zero}^{v_t} z_t(\bal_t, v) \ud v
                           - g_t(v_{(1, t)}) + \mu_t  \\
                        &= g_{t - 1}(v_{(1, t - 1)})
                           + \int_{\Zero}^{v_t} y_t(v_{(1, t - 1)}, v) \ud v
                           - g_t(v_{(1, t)}) - \mu_{t - 1} + \mu_t.
          \end{align*}
          Taking partial derivative with respect to $v_t$, we have
          \begin{align*}
            \frac{\partial s_t(\bal_t)}{v_t}
                = \Zero + y_t(v_{(1, t)})
                  - \frac{\partial g(v_{(1, t)})}{\partial v_t} - \Zero + \Zero.
          \end{align*}
          Since $y_t(v_{(1, t)}) = \partial g_t(v_{(1, t)}) / \partial v_t$, we
          conclude that $\partial s_t(\bal_t) / \partial v_t = \Zero$. Hence
          $s_t$ constructed as (\ref{eq:bspd}) is a function of $\bal_t$.
    \item $s_t(\bal_t) \leq \bal_t$; the balance update formula holds: the same
          with the proof of Lemma \ref{lem:symmstrong}.
    \item Since $z_t(\bal_t, v_t) = y_t(v_{(1, t)})$ is the sub-gradient of
          $g_t(v_{(1, t)})$ with respect to $v_t$, which is also convex and
          weakly increasing in $v_t$, then by the Envelope theorem, the
          construction (\ref{eq:bpay}) of $q_t$ implies the in-stage IC
          condition (\ref{cond:sic}).
    \item $g_t$ is consistent implies that $s_t(\bal_t) - \E_{v_t}[\hu_t(\bal_t,
          v_t; v_t)]$ is a constant, i.e.,
          \begin{align*}
              s_t(\bal_t) - \E_{v_t}\left[\hu_t(\bal_t, v_t; v_t)\right]
            &=\E_{v_t}\left[s_t(\bal_t) - d_t(\bal_t, v_t)\right]  \\
            &=\E_{v_t}\left[\big(\bal_t + d_t(\bal_t, v_t) - \bal_{t + 1}\big)
                            - d_t(\bal_t, v_t)\right]  \\
            &=\E_{v_t}\left[\bal_t - \bal_{t + 1}\right]  \\
            &=\E_{v_t}\left[g_{t - 1}(v_{(1, t - 1)}) - \mu_{t - 1}
                            - g_t(v_{(1, t)}) + \mu_t\right]  \\
            &=g_{t - 1}(v_{(1, t - 1)}) - \E_{v_t}\left[g_t(v_{(1, t)})\right]
              - \mu_{t - 1} + \mu_t  \\
            &=\textrm{constant},
          \end{align*}
          where the last ``$=$'' is from the consistency of $g_t$ and the fact
          that $\mu_{t - 1}$ and $\mu_t$ are constants.

          It then implies the IC condition (\ref{cond:spd}).
          % \begin{align*}
          %   & s_t(\bal_t) - \E_{v_t}\big[\hu_t(\bal_t, v_t; v_t)\big]
          %     = s_t(\bal_t) - \E_{v_t}\big[d_t(\bal_t, v_t)\big]  \\
          %   \stackrel{\text{(\ref{eq:defbal})}}{=}&
          %       g(v_{(1, t - 1)}) - \mu_{t - 1}
          %       - \E_{v_t}\big[g(v_{(1, t)}) - \mu_t\big]
          %    = \mu_t - \mu_{t - 1} + c_t \Infer \text{(\ref{cond:spd})}.
          % \end{align*}
    \item The IR condition (\ref{cond:ir}) holds: the same with the proof of
          Lemma \ref{lem:symmstrong}.
  \end{itemize}

  ``$\Reduce$'': $B^{g, y}$ is a core BAM $\Infer$ $y$ is the sub-gradient
  of $g$ and $g$ is consistent, symmetric, and convex and increasing in $v_t$.

  By definition, the Core BAM $B^{g, y}$ is a valid BAM. Most of the arguments
  for the sufficient direction also apply to the necessary direction.
  \begin{itemize}
    \item By the construction (\ref{eq:bbal})-(\ref{eq:bspd}),
          \begin{align*}
            \frac{\partial s_t(\bal_t)}{v_t}
                = y_t(v_{(1, t)}) - \frac{\partial g(v_{(1, t)})}{\partial v_t},
          \end{align*}
          which must be $\Zero$, since $s_t(\bal_t)$ is invariant with respect
          to $v_t$. Hence $y_t$ is the sub-gradient of $g_t$ with respect to
          $v_t$, i.e.,
          \begin{align*}
            y_t(v_{(1, t)}) = \frac{\partial g(v_{(1, t)})}{\partial v_t}.
          \end{align*}
    \item Since $s_t(\bal_t) = \bal_t + d_t(\bal_t, v_t) - \bal_{t + 1}$ is a
          valid function on input $(\bal_t)$, hence $\forall v_{(1, t - 1)},
          v'_{(1, t - 1)} \in \dV_{(1, t - 1)}, v_t \in \dV_t,$
          \begin{align*}
            \bal_t(v_{(1, t - 1)}) = \bal_t(v'_{(1, t - 1)})
              \Infer s_t\left(\bal_t(v_{(1, t - 1)})\right)
                        = s_t\left(\bal_t(v'_{(1, t - 1)})\right).
          \end{align*}

          By substituting $s_t(\bal_t)$, we have
          \begin{align*}
            \bal_t(v_{(1, t - 1)}) + d_t(\bal_t, v_t) - \bal_{t + 1}(v_{(1, t)})
              = \bal_t(v'_{(1, t - 1)}) + d_t(\bal_t, v_t)
                  - \bal_{t + 1}(v'_{(1, t - 1)}, v_t),
          \end{align*}
          which then implies
          \begin{align*}
            \bal_{t + 1}(v_{(1, t)}) = \bal_{t + 1}(v'_{(1, t - 1)}, v_t),
          \end{align*}
          since $\bal_t(v_{(1, t - 1)}) = \bal_t(v'_{(1, t - 1)})$.

          % $z_t(\bal_t, v_t) = y_t(v_{(1, t)})$ is a valid function on input
          % $(\bal_t, v_t)$, hence $\forall v_{(1, t - 1)}, v'_{(1, t - 1)} \in
          % \dV_{(1, t - 1)}, v_t \in \dV_t,$
          % \begin{align*}
          %   \bal_t(v_{(1, t - 1)}) = \bal_t(v'_{(1, t - 1)})
          %     \Infer y_t(v_{(1, t)}) = z_t(\bal_t, v_t)
          %                            = y_t(v'_{(1, t - 1)}, v_t).
          % \end{align*}

          Note that by the construction of $\bal_t$,
          \begin{align*}
            \bal_t(v_{(1, t - 1)}) = \bal_t(v'_{(1, t - 1)})
              \iff g_{t - 1}(v_{(1, t - 1)}) = g_{t - 1}(v'_{(1, t - 1)}),
          \end{align*}
          and
          \begin{align*}
            \bal_{t + 1}(v_{(1, t)}) = \bal_{t + 1}(v'_{(1, t - 1)}, v_t)
              \iff g_t(v_{(1, t)}) = g_t(v'_{(1, t - 1)}, v_t).
          \end{align*}
          %
          % and by $y_t(v_{(1, t)}) = \partial g_t(v_{(1, t)}) / \partial v_t$,
          % \begin{align*}
          %   \forall v_t \in \dV_t, y_t(v_{(1, t)}) = y_t(v'_{(1, t - 1)}, v_t)
          %     \iff g_t(v_{(1, t)}) = g_t(v'_{(1, t - 1)}, v_t)
          %                            + C(v_{(1, t - 1)}, v'_{(1, t - 1)}),
          % \end{align*}
          % where $C(v_{(1, t - 1)}, v'_{(1, t - 1)})$ is a constant with respect
          % to $v_t$.
          %
          % We then show that $C(v_{(1, t - 1)}, v'_{(1, t - 1)}) \equiv 0$.
          % Since $s_t(\bal_t) = \bal_t + d_t(\bal_t, v_t) - \bal_{t + 1}$ is a
          % valid function on input $(\bal_t)$, hence
          % \begin{align*}
          %   \bal_t(v_{(1, t - 1)}) = \bal_t(v'_{(1, t - 1)})
          %     \Infer s_t\left(\bal_t(v_{(1, t - 1)})\right)
          %               = s_t\left(\bal_t(v'_{(1, t - 1)})\right).
          % \end{align*}

          Hence $g_t$ must be symmetric, i.e., $\forall v_{(1, t - 1)},
          v'_{(1, t - 1)} \in \dV_{(1, t - 1)}, v_t \in \dV_t,$
          \begin{align*}
            g_{t - 1}(v_{(1, t - 1)}) = g_{t - 1}(v'_{(1, t - 1)})
              &\Infer \bal_t(v_{(1, t - 1)}) = \bal_t(v'_{(1, t - 1)})  \\
              &\Infer \bal_{t + 1}(v_{(1, t)})
                        = \bal_{t + 1}(v'_{(1, t - 1)}, v_t)  \\
              &\Infer g_t(v_{(1, t)}) = g_t(v'_{(1, t - 1)}, v_t).
          \end{align*}

    \item Since the range of $z_t(\bal_t, v_t)$ must be $\dX_t$, so is the
          range of $y_t(v_{(1, t)})$.
    \item By the in-stage IC condition (\ref{cond:sic}) and the Envelope
          theorem, $y_t(v_{(1, t)}$ must be weakly increasing in $v_t$, hence
          $g_t$, being the integration of $y_t$, must be convex and increasing
          in $v_t$.
    \item Finally, the another IC condition (\ref{cond:spd}) implies that
          $s_t(\bal_t) - \E_{v_t}\big[\hu_t(\bal_t, v_t; v_t)\big]$ must be a
          constant (with respect to $\bal_t$). Hence,
          \begin{align*}
              g_{t - 1}(v_{(1, t - 1)}) - \E_{v_t}\left[g_t(v_{(1, t)})\right]
            &=\bal_t + \mu_{t - 1} - \E_{v_t}\left[\bal_{t + 1} + \mu_t\right]
              \\
            &=s_t(\bal_t) - \E_{v_t}\big[d_t(\bal_t, v_t; v_t)\big]
              + \mu_{t - 1} - \mu_t  \\
            &=s_t(\bal_t) - \E_{v_t}\big[\hu_t(\bal_t, v_t; v_t)\big]
              + \mu_{t - 1} - \mu_t  \\
            &= \text{constant}.
          \end{align*}
          Therefore $g_t$ is consistent.
  \end{itemize}
\end{proof}

\subsection{Proof of Theorem \ref{thm:bl}}\label{app:thm:bl}
% Proof of Theorem bl

\begin{proof}[Proof of Theorem \ref{thm:bl}]
  Assume by contradiction that there exists a core BAM, $B$, such that,
  \begin{align}\label{eq:assump}
      \Rev(B) > \Rev(\SMA) + \E_{v_{(1, T)}}\bigg[
                  \sum_{\tau = 1}^T s_\tau(\bal_\tau)\bigg].
  \end{align}
  Then we construct a mechanism $M$ that operates independently on each stage
  but generates strictly higher revenue than the optimal separate mechanism in
  at least one stage.

  To ensure that the constructed mechanism is history-independent and IC-IR
  within each single stage. We create a fake bank account to mimic the behavior
  of the real bank account in $B$. One can think of the balance $\bal^\dag_t$
  of the fake bank account as simply a number, satisfying that the distribution
  of $\bal^\dag_t$ is the same as the distribution of the real balance (as a
  random-valued function of previous types).

  Consider the following construction of $M$ based on $B$.
  \begin{align*}
    x_t(v_t) = z_t(\bal^\dag_t, v_t),~p_t(v_t) = q_t(\bal^\dag_t, v_t),
  \end{align*}
  where in each stage $t$, $\bal^\dag_t$ is a random variable defined by the
  following equations for $\tau$ from $1$ to $t - 1$.
  \begin{align*}
    \bal^\dag_{\tau + 1} = \bal^\dag_\tau - s_\tau(\bal^\dag_\tau)
                           + d_\tau(\bal^\dag_\tau, v^\dag_\tau),~
    v^\dag_\tau \sim \F_\tau.
  \end{align*}
  Thus $\bal^\dag_t$ is independent of the buyer's real types, and has the
  same distribution as $\bal_t$. Note that $v_t$ is independent of the history,
  hence independent of $\bal_t$ as well. Then
  \begin{align}
    & \E_{v_{(1, t - 1)}} [q_t(\bal_t, v_t)]
        = \E_{\bal^\dag_t} [q_t(\bal^\dag_t, v_t)]
        = \E_{\bal^\dag_t} [p_t(v_t)],
      \nonumber  \\
    \Infer~&
      \Rev(M) = \sum_{\tau = 1}^T \E_{\bal^\dag_\tau, v_\tau} \big[
                p_\tau(v_\tau)\big]
              = \Rev(B) - \E_{v_{(1, T)}} \bigg[\sum_{\tau = 1}^T
                  s_\tau(\bal_\tau)\bigg]
              > \Rev(\SMA)
      \nonumber  \\
    \Infer~& \exists \tau,~s.t.~\E_{\bal^\dag_\tau, v_\tau}
        \big[p_\tau(v_\tau)\big] > \rho_\tau,
      \label{eq:contr}
  \end{align}
  namely, at least in one stage $\tau$, $M$ generates strictly higher revenue
  than the optimal revenue for stage $\tau$, $\rho_\tau$.

  By construction, $M$ is history-independent, and by
  (\ref{cond:sic})(\ref{cond:ir}) and $q_t \geq 0$ (by the construction of
  core BAM), $M$ is also IC and IR within each single stage. Thus
  (\ref{eq:contr}) contradicts the optimality of $\SMA$.
\end{proof}

\subsection{Proof of Lemma \ref{lem:3appspddpsupb}}\label{app:lem:3appspddpsupb}
% Proof of Lemma 3appspddpsupb

\begin{proof}[Proof of Lemma \ref{lem:3appspddpsupb}]
  From previous definition and analysis, we have two upper bounds on
  $s_t(\bal_t)$.
  \begin{enumerate}
    \item By Definition \ref{def:bam}, $s_t(\bal_t) \leq \bal_t$.
    \item By IC constraint (\ref{cond:spd}) of BAM, we have
          \begin{align*}
            s_t(\bal_t) = \E_{v_t}[\hu_t(\bal_t, v_t; v_t)
                                   - \hu(0, v_t; v_t)] + s_t(0)
                        \leq \E_{v_t}[\hu_t(\bal_t, v_t; v_t)]
                        %\leq \int_{\dV_t} \One \cdot v \ud \F(v) 
                        = \Val_t.
          \end{align*}
  \end{enumerate}

  For any $\bal_t$, the minimum of the above two bounds can be achieved if
  \begin{align*}% \label{eq:aimppp}
    & \E_{v_t}\left[\hu(\bal_t, v_t; v_t)\right]
          % = \min\left\{\bal_t, \int_{\dV_t} \One \cdot v \ud \F(v)\right\}
          = \min\left\{\bal_t, \Val_t\right\}  \\
    & \hu(0, v_t; v_t) = s_t(0) = 0.
  \end{align*}

  Furthermore, by IR constraint (\ref{cond:ir}), we have,
  \begin{align*}
    d_t(\bal_t, v_t) \leq \hu_t(\bal_t, v_t; v_t) \leq \Val_t.
  \end{align*}

  % On the other hand, by (\ref{eq:defbal}), we have another upper bound on the
  % sum of spends, i.e.,
  % \begin{align*}
  %   \sum_{\tau = 1}^T s_\tau(\bal_\tau)
  %     = \sum_{\tau = 1}^T d_\tau(\bal_\tau, v_\tau) - \bal_{T + 1}
  %     \leq \sum_{\tau = 1}^T d_\tau(\bal_\tau, v_\tau).
  % \end{align*}
  %
  % Furthermore, by (\ref{cond:ir}), we have,
  % \begin{align*}
  %   \sum_{\tau = 1}^T s_\tau(\bal_\tau)
  %     \leq \sum_{\tau = 1}^T d_\tau(\bal_\tau, v_\tau)
  %     \leq \sum_{\tau = 1}^T \hu_\tau(\bal_\tau, v_\tau; v_\tau)
  %     \leq \sum_{\tau = 1}^T \One \cdot v_\tau,
  % \end{align*}
  % where the sum of $d_\tau(\bal_\tau, v_\tau)$ can achieve the upper bound, if
  % \begin{align*}% \label{eq:aimgff}
  %   d_\tau(\bal_\tau, v_\tau) = \One \cdot v_\tau,~
  %   z_\tau(\bal_\tau, v_\tau) = \One,~q_\tau(\bal_\tau, v_\tau) = 0.
  % \end{align*}
\end{proof}

\subsection{Proof of Lemma \ref{lem:sbd}}\label{app:lem:sbd}
% Proof of Lemma sbd

We prove the following stronger version of Lemma \ref{lem:sbd}. Lemma
\ref{lem:sbd} is then directly implied by Lemma \ref{lem:3appspddpsupb} and
Lemma \ref{lem:sbdstrong}.

\begin{lemma}\label{lem:sbdstrong}
  For deposit sequence $\frd_{(1, t)}$ and spend sequence $\frs_{(1, t)}$
  with given bounds on each of the deposits and spends,
  \begin{align*}
    \forall \tau \in [t],~
    0 \leq \frd_\tau \leq \One \cdot v_\tau,~\frs_\tau \leq \Val_\tau,
  \end{align*}
  and the corresponding balance defined with respect to the balance update
  formula,
  \begin{align*}
    \bal_1 = 0,~\bal_{\tau + 1} = \bal_\tau + \frd_\tau - \frs_\tau,~
    \frs_\tau \leq \bal_\tau, \forall \tau \in [t].
  \end{align*}

  We have the following upper bound on the cumulated spends,
  \begin{align}\label{eq:sbd}
    \forall v_{(1, t)} \in \dV_{(1, t)},~
    \sum_{\tau = 1}^t \frs_\tau \leq
      \sum_{\tau = 1}^t s^*_\tau(\bal^*_\tau).
  \end{align}
  % where $d^*_{(1, t)}, s^*_{(1, t)}$ are the greedy deposit-spend policy:
  % deposit and spend as much as allowed (meeting the upper bounds on deposits
  % and spends),
  % \begin{align*}
  %   \forall \tau \in [t],~
  %   d^*_\tau(\bal^*_\tau, v_\tau) = \One \cdot v_\tau,~
  %   s^*_\tau(\bal^*_\tau) = \min\{\bal^*_\tau, \Val_\tau\},
  % \end{align*}
  % and $\bal^*_t$ is the balance defined by $d^*_t$ and $s^*_t$ accordingly,
  % \begin{align*}
  %   \bal^*_1 = 0,~\bal^*_{\tau + 1} =
  %       \bal^*_\tau + d^*_\tau(\bal^*_\tau, v_\tau) - s^*_\tau(\bal^*_\tau),
  %   \forall \tau \in [t].
  % \end{align*}
\end{lemma}

\begin{proof}[Proof of Lemma \ref{lem:sbdstrong}]
  Let $t' \leq t$ be the largest element such that $s^*_{t'}(\bal^*_{t'}) =
  \bal^*_{t'}$. Then according to the balance update formula,
  \begin{align*}
      \sum_{\tau = 1}^t \frs_\tau
    &= \frs_{t'} + \sum_{\tau = t' + 1}^t \frs_\tau
        + \sum_{\tau = 1}^{t' - 1}(\bal_\tau + \frd_\tau - \bal_{\tau + 1})  \\
    &= \frs_{t'} + \sum_{\tau = t' + 1}^t \frs_\tau
        + \bal_1 - \bal_{t'} + \sum_{\tau = 1}^{t' - 1} \frd_\tau.
  \end{align*}

  Remember that by definition, $\bal_1 = \bal^*_1 = 0$, and the upper bounds on
  $\frd_\tau$ and $\frs_\tau$,
  \begin{align*}
    \forall \tau \in [t],~
        \frd_\tau \leq \One \cdot v_\tau,
        \frs_\tau \leq \bal_\tau, \frs_\tau \leq \Val_\tau.
  \end{align*}

  Hence
  \begin{align}\label{eq:sbdupb}
       \sum_{\tau = 1}^t \frs_\tau
    &= \frs_{t'} + \sum_{\tau = t' + 1}^t \frs_\tau
          + \bal_1 - \bal_{t'} + \sum_{\tau = 1}^{t' - 1} \frd_\tau
       \nonumber  \\
    &\leq \bal_{t'} + \sum_{\tau = t' + 1}^t \Val_\tau
          + \bal_1 - \bal_{t'} + \sum_{\tau = 1}^{t' - 1} \One \cdot v_\tau
       \nonumber  \\
    & = \sum_{\tau = t' + 1}^t \Val_\tau
          + \sum_{\tau = 1}^{t' - 1} \One \cdot v_\tau.
  \end{align}

  Recall the construction of $d^*$ and $s^*$,
  \begin{align*}
    \forall \tau \in [t],~
        d^*_\tau(\bal^*_\tau, v_\tau) = \One \cdot v_\tau,
        s^*_\tau(\bal^*_\tau) = \min\{\bal^*_\tau, \Val_\tau\}.
  \end{align*}

  Combining the definition of $t'$, we have
  \begin{align*}
    \forall \tau > t',~s^*_\tau(\bal^*_\tau) = \Val_\tau.
  \end{align*}

  Therefore, apply the balance update formula again,
  \begin{align*}
       \sum_{\tau = 1}^t s^*_\tau(\bal^*_\tau)
    &= s^*_{t'}(\bal^*_{t'}) + \sum_{\tau = t' + }^t s^*_\tau(\bal^*_\tau)
        + \sum_{\tau = 1}^{t' - 1}\big(\bal^*_\tau
                                       + d^*_\tau(\bal^*_\tau, v_\tau)
                                       - \bal^*_{\tau + 1}\big)  \\
    &= \bal^*_{t'} + \sum_{\tau = t' + 1}^t \Val_\tau
        + \bal^*_1 - \bal^*_{t'} + \sum_{\tau = 1}^{t' - 1} \One \cdot v_\tau  \\
    &= \sum_{\tau = t' + 1}^t \Val_\tau
          + \sum_{\tau = 1}^{t' - 1} \One \cdot v_\tau,
  \end{align*}
  which, by (\ref{eq:sbdupb}), is the upper bound on $\sum_{\tau = 1}^t
  \frs_\tau$.

  Thus we conclude that
  \begin{align*}
    \sum_{\tau = 1}^t s^*_\tau(\bal^*_\tau) \geq \sum_{\tau = 1}^t \frs_\tau.
  \end{align*}
\end{proof}

\subsection{Proof of Lemma \ref{lem:3appbam}}\label{app:lem:3appbam}
% Proof of Lemma 3appbam

\begin{proof}[Proof of Lemma \ref{lem:3appbam}]

  We prove that the constructed $B$ (by Mechanism \ref{mech:3app}) is a BAM.

  With parameter $\theta = 3s_t(\bal_t) = 3\min\{\bal_t, \Val_t / 3\} \leq
  \Val_t$, $\langle x_t^\PM(\theta, \bcdot), p_t^\PM(\theta, \bcdot) \rangle$
  is the posted-price-for-the-grand-bundle mechanism for stage $t$ with
  non-negative posted-price by definition. Hence $z_t, q_t, d_t, s_t$, are
  defined on the domains required by the definition of bank account mechanism.

  Within each single stage and for any balance $\bal_t$, stage mechanism
  $\langle z_t(\bal_t, \bcdot), q_t(\bal_t, \bcdot) \rangle$ is truthful,
  because it is a randomization of three truthful stage mechanisms: the optimal
  history-independent mechanism, the give-for-free mechanism, and the
  posted-price-for-the-grand-bundle mechanism with a parameter independent with
  the reported type on current stage, $v_t$. Hence $B$ satisfies the IC
  constraint (\ref{cond:sic}).

  Another IC constraint (\ref{cond:spd}) is verified as follows,
  \begin{align*}
      & \E_{v_t}\big[\hu_t(\bal_t, v_t; v_t) - \hu_t(\bal'_t, v_t; v_t)\big]  \\
    =~& \frac13 \E_{v_t}\Big[
            \big(z_t(\bal_t, v_t) \cdot v_t - q_t(\bal_t, v_t)\big)
            - \big(z_t(\bal'_t, v_t) \cdot v_t - q_t(\bal'_t, v_t)\big)\Big]  \\
    =~& \frac13 \E_{v_t}\bigg[
                  \Big(x^{\SM}_t(v_t) + \One
                       + x^{\PM}_t\big(3s_t(\bal_t), v_t\big)\Big) \cdot v_t
                  - \Big(p^{\SM}_t(v_t)
                         + p^{\PM}_t\big(3s_t(\bal_t), v_t\big)\Big)\bigg]  \\
      & - \frac13 \E_{v_t}\bigg[
                  \Big(x^{\SM}_t(v_t) + \One
                       + x^{\PM}_t\big(3s_t(\bal'_t), v_t\big)\Big) \cdot v_t
                  - \Big(p^{\SM}_t(v_t)
                         + p^{\PM}_t\big(3s_t(\bal'_t), v_t\big)\Big)\bigg]  \\
    =~& \frac13 \Big(\E_{v_t}\big[u^\PM_t(3s_t(\bal_t), v_t; v_t)\big]
          - \E_{v_t}\big[u^\PM_t(3s_t(\bal'_t), v_t; v_t)\big]\Big).
  \end{align*}

  By the definition of the posted-price-for-the-grand-bundle mechanism with
  parameter $\theta$, the buyer utility $u^\PM_t(\theta, v_t; v_t) = \theta$
  (see (\ref{eq:pppu})), hence
  \begin{align*}
    & \E_{v_t}\big[\hu_t(\bal_t, v_t; v_t) - \hu_t(\bal'_t, v_t; v_t)\big]
      = \frac13 \Big(\E_{v_t}\big[u^\PM_t(3s_t(\bal_t), v_t; v_t)\big]
          - \E_{v_t}\big[u^\PM_t(3s_t(\bal'_t), v_t; v_t)\big]\Big)  \\
    =~& \frac13 \big(3s_t(\bal_t) - 3s_t(\bal'_t)\big)
      = s_t(\bal_t) - s_t(\bal'_t) \Infer \text{(\ref{cond:spd})}.
  \end{align*}

  By Lemma \ref{lem:sconic}, conditions (\ref{cond:sic}) and (\ref{cond:spd})
  are sufficient conditions of stage-wise IC, so we have proved the IC property
  of $B$.

  Finally, by the construction of $z_t$, $q_t$, and $d_t$, (see (\ref{eq:3all}),
  (\ref{eq:3pay}), and (\ref{eq:3dps}) respectively), we can prove the IR
  constraint (\ref{cond:ir}) as follows. Hence $B$ is a BAM.
  \begin{align*}
    \hu_t(\bal_t, v_t; v_t)
      = \frac13 \Big(u^\SM_t(v_t) + \One \cdot v_t
                     + u^\PM_t\big(3s_t(\bal_t), v_t; v_t\big)\Big)
      \geq \frac13 v_t \cdot \One = d_t(\bal_t, v_t)
    \Infer \text{(\ref{cond:ir})}.
  \end{align*}
\end{proof}

\subsection{Proof of Lemma \ref{lem:pairb}}\label{app:lem:pairb}
% Proof of Lemma pairb

\begin{proof}[Proof of Lemma \ref{lem:pairb}]
  Since $s^\sgm_t(\bal^\sgm_t) = 0$ if $\sgm_t = 0$, we cannot prove the lemma
  by arguing $s^\sgm_t(\bal^\sgm_t) \geq s^*_t(\bal^*_t) / 4$. Even if when
  $\sgm_t = 1$, it is still nearly impossible to argue for similar lower bounds
  of $s^\sgm_t(\bal^\sgm_t)$, because $\bal^\sgm_t$ --- the upper bound of
  $s^\sgm_t(\bal^\sgm_t)$ --- depends on the first $t - 1$ bits of $\sgm$,
  which have $2^{t - 1}$ different cases.

  The trick here is to consider a pair of $\sgm$'s, say $\sgm$ and $\sgp(t)$,
  and to argue that if $\sgm_t = 1$, then $s^\sgm_t(\bal^\sgm_t) +
  s^{\sgp(t)}_t(\bal^{\sgp(t)}_t) \geq s^*_t(\bal^*_t)$.

  In particular, for any $\sgm \in \{0, 1\}^T$ and $t \in [T]$, let $\sgp(t) =
  \osg_1 \cdots \osg_{t - 1} \sgm_t \cdots \sgm_T$, which is different from
  $\sgm$ in the first $t - 1$ bits, and the same with $\sgm$ in the rest of the
  bits.

  In order to use Lemma \ref{lem:sbdstrong}, consider the following deposit
  sequence, $\frd_{(1, t - 1)}$, and spend sequence, $\frs_{(1, t - 1)}$,
  \begin{align}\label{eq:dappsds}
    \forall \tau < t,~
    \frd_\tau = d^\sgm_\tau(\bal^\sgm_\tau, v_\tau)
                + d^{\sgp(t)}_\tau(\bal^{\sgp(t)}, v_\tau),
    \frs_\tau = s^\sgm_\tau(\bal^\sgm_\tau) + s^{\sgp(t)}_\tau(\bal^{\sgp(t)}).
  \end{align}

  Recall the definitions of $d^\sgm_\tau$ and $s^\sgm_\tau$, i.e.,
  \begin{align*}
    \text{if}~\sgm_\tau = 0:~&
      d^\sgm_\tau(\bal^\sgm_\tau, v_\tau) = \One \cdot v_\tau,~
      s^\sgm_\tau(\bal^\sgm_\tau)         = 0;  \\
    \text{if}~\sgm_\tau = 1:~&
      d^\sgm_\tau(\bal^\sgm_\tau, v_\tau) = 0,~
      s^\sgm_\tau(\bal^\sgm_\tau)         = \min\{\bal^\sgm_\tau, \Val_\tau\}.
  \end{align*}

  Hence
  \begin{align}\label{eq:dappsdsv}
    \forall \tau < t,~
    \frd_\tau = \One \cdot v_\tau,
    \frs_\tau = \left\{\begin{array}{ll}
      \min\{\bal^\sgm_\tau, \Val_\tau\}, & \text{if}~\sgm_\tau = 1  \\
      \min\{\bal^{\sgp(t)}_\tau, \Val_\tau\}, & \text{if}~\sgm_\tau = 0
    \end{array}\right..
  \end{align}

  Consider the following balance $\bal_{(1, t)}$ defined by $\frd_{(1, t - 1)}$
  and $\frs_{(1, t - 1)}$,
  \begin{align*}
    \bal_1 = 0,~\bal_{\tau + 1} = \bal_\tau + \frd_\tau - \frs_\tau, \forall
      \tau \in [t - 1].
  \end{align*}

  By (\ref{eq:dappsds}), it is not hard to demonstrate the fact below through
  induction from $\tau = 1$ to $\tau = t$,
  \begin{align}\label{eq:dappeqbal}
    \forall \tau \in [t],~\bal_\tau = \bal^\sgm_\tau + \bal^{\sgp(t)}_\tau.
  \end{align}

  Also note that from (\ref{eq:dappsdsv}), we have
  \begin{align*}
    \forall \tau < t,~\frd_\tau = \One \cdot v_\tau, \frs_\tau \leq \Val_\tau.
  \end{align*}

  Hence the following inequality is from Lemma \ref{lem:sbdstrong},
  \begin{align}\label{eq:virsbd}
    \sum_{\tau = 1}^{t - 1} \frs_\tau
      \leq \sum_{\tau = 1}^{t - 1} s^*_\tau(\bal^*_\tau).
  \end{align}

  Apply the balance update formula (\ref{eq:defbal}) recursively,
  \begin{align*}
    \bal_t = \bal_1 + \sum_{\tau = 1}^{t - 1} \frd_\tau
                    - \sum_{\tau = 1}^{t - 1} \frs_\tau.
  \end{align*}

  Combining $\bal_1 = 0$ (by definition), (\ref{eq:dappsdsv}),
  (\ref{eq:dappeqbal}), and (\ref{eq:virsbd}), we have
  \begin{align*}
      \bal^\sgm_t + \bal^{\sgp(t)}_t
    = \bal_t
    \geq 0 + \sum_{\tau = 1}^{t - 1} \One \cdot v_\tau
           - \sum_{\tau = 1}^{t - 1} s^*_\tau(\bal^*_\tau),
  \end{align*}
  where by the definition of $d^*_{(1, T)}$ and $s^*_{(1, T)}$,
  \begin{align*}
      0 + \sum_{\tau = 1}^{t - 1} \One \cdot v_\tau
        - \sum_{\tau = 1}^{t - 1} s^*_\tau(\bal^*_\tau)
    = \bal^*_1 + \sum_{\tau = 1}^{t - 1} d^*_\tau(\bal^*_\tau, v_\tau)
               - \sum_{\tau = 1}^{t - 1} s^*_\tau(\bal^*_\tau)
    = \bal^*_t,
  \end{align*}
  hence
  \begin{align}\label{eq:dappbaltb}
    \forall v_{(1, t)} \in \dV_{(1, t)},~
      \bal^\sgm_t + \bal^{\sgp(t)}_t \geq \bal^*_t.
  \end{align}

  Note that $\sgm$ agrees with $\sgp(t)$ on the $t$-th bit, namely, if $\sgm_t
  = 1$, then by construction of $B^\sgm$, $\forall v_{(1, t)} \in \dV_{(1, t)}$,
  \begin{align}\label{eq:dapplemgoal}
      s^\sgm_t(\bal^\sgm_t) + s^{\sgp(t)}_t(\bal^{\sgp(t)}_t)
    & = \min\{\bal^\sgm_t, \Val_t\} + \min\{\bal^{\sgp(t)}_t, \Val_t\}
      \nonumber  \\
    & \geq \min\{\bal^\sgm_t + \bal^{\sgp(t)}_t, \Val_t\} \nonumber  \\
    & \geq \min\{\bal^*_t, \Val_t\} = s^*_t(\bal^*_t),
  \end{align}
  where the last inequality is from (\ref{eq:dappbaltb}).

  Now we are ready to finish the proof. $\forall t \in [T], v_{(1, t)} \in
  \dV_{(1, t)}$,
  \begin{align*}
      \E_\sgm \Big[s^\sgm_t(\bal^\sgm_t)\Big]
    & = \E_\sgm \Big[s^\sgm_t(\bal^\sgm_t) \Big| \sgm_t = 1\Big]
            \cdot \Pr[\sgm_t = 1]
        + \E_\sgm \Big[s^\sgm_t(\bal^\sgm_t) \Big| \sgm_t = 0\Big]
            \cdot \Pr[\sgm_t = 0]  \\
    & = \frac12 \E_\sgm \Big[s^\sgm_t(\bal^\sgm_t) \Big| \sgm_t = 1\Big] + 0  \\
    & = \frac14 \E_\sgm \Big[s^\sgm_t(\bal^\sgm_t)
                    + s^{\sgp(t)}_t(\bal^{\sgp(t)}_t) \Big| \sgm_t = 1\Big]  \\
    & \geq \frac14 s^*_t(\bal^*_t),
  \end{align*}
  where the last inequality is from (\ref{eq:dapplemgoal}).
\end{proof}

\subsection{Proof of Theorem \ref{thm:dp}}\label{app:thm:dp}
% Proof of Theorem dp

\begin{proof}[Proof of Theorem \ref{thm:dp}]
  By Theorem \ref{thm:rep}, it is WLOG to assume that the optimal mechanism is
  a core BAM $B^{g, y}$. Moreover, by Theorem \ref{thm:corebam} and Lemma
  \ref{lem:symmstrong}, it is WLOG to assume that $g(h) = \Utl(B | h)$. Since
  $y$ is the sub-gradient of $g$ and $g$ is symmetric, we denote
  $z_t(\bal_t, v_t)$ by $y_t\big(g(v_{(1, t - 1)}), v_t\big)$ throughout the
  proof.

  Let $\phi_t(\xi)$ be defined as follows,
  \begin{align}\label{eq:dpstat}
    \phi_t(\xi) = \max_{B : g(v_{(1, t)}) = \xi} \E_{v(t + 1, T)}\bigg[
      \sum_{\tau = t + 1}^T y_t\big(g(v_{(1, \tau - 1)}), v_\tau\big)
                              \cdot v_\tau  - \bal_{T + 1}\bigg].
  \end{align}
  Consider the following program.
  \begin{align}
    \mathrm{maximize}   \quad & \E_{v_t}\Big[
      y_t(\xi, v_t) \cdot v_t + \phi_t\big(g(\xi, v_t)\big)\Big] \label{dp:obj}
        \\
    \mathrm{subject~to} \quad & \hu_t(\xi, v_t; v_t) \geq
      \hu_t(\xi, v'_t; v_t),~\forall v_t \in \dV_t \label{dp:sic}  \\
    & g(\xi, v_t) = \xi + \hu_t(\xi, v_t; v_t)
        - \E_{v'_t}\big[\hu_t(\xi, v'_t; v'_t)\big] \label{dp:spd}  \\
    & g(\xi, v_t) \geq 0,~\forall v_t \in \dV_t \label{dp:ir}
  \end{align}
  where
  \begin{lemma}\label{lem:dpgen}
    $\max_{\xi \geq 0} \phi_0(\xi)$ is the maximum expected revenue, and
    $\phi_{t - 1}(\xi)$ is computed based on $\phi_t$ according to the above
    program (\ref{dp:obj}), where the entire algorithm is essentially a dynamic
    program starting with $\phi_T(\xi) = -\bal_{T + 1} = -\xi$.\footnote{Recall
    that $\mu_T = 0$, hence $\bal_{T + 1} = \Utl(B | v_{(1, T)})$.}
  \end{lemma}

  \begin{proof}[Proof of Lemma \ref{lem:dpgen}]
    According to the construction of Lemma \ref{lem:symmstrong},
    \begin{align*}
      \max_{\xi \geq 0} \phi_0(\xi) &= \max_{B : \Utl(B) \geq 0}
        \E_{v_{(1, T)}}\bigg[\sum_{\tau = 1}^T
          y_t\big(\Utl(B | v_{(1, \tau - 1)}), v_\tau\big) \cdot v_\tau
            - \bal_{T + 1}\bigg]  \\
      &= \max_B \E_{v_{(1, T)}}\bigg[\sum_{\tau = 1}^T \Big(
          \hu_\tau(\bal_\tau, v_\tau; v_\tau) + q_\tau(\bal_\tau, v_\tau)\Big)
            - \bal_{T + 1}\bigg]  \\
      &= \max_B \E_{v_{(1, T)}}\bigg[\sum_{\tau = 1}^T d_\tau(\bal_\tau, v_\tau)
            - \bal_{T + 1} + \sum_{\tau = 1}^T q_\tau(\bal_\tau, v_\tau)\bigg]
            \\
      &= \max_B \E_{v_{(1, T)}}\bigg[\sum_{\tau = 1}^T
           \Big(s_\tau(\bal_\tau) + q_\tau(\bal_\tau, v_\tau)\Big) - \bal_1
           \bigg]  \\
      &= \max_B \Rev(B).
    \end{align*}
    Hence  $\max_{\xi \geq 0} \phi_0(\xi)$ is the maximum expected revenue.

    In the meanwhile,
    \begin{itemize}
      \item (\ref{dp:sic}) $\iff$ $g$ is convex in $v_t$;
      \item (\ref{dp:spd}) $\iff$ $g$ is consistent;
      \item the definition of $\hu_t$ and (\ref{dp:spd}) $\iff$ $y_t$ is the
            sub-gradient of $g(\xi, v_t)$;
      \item finally, by the assumption that $g(h) = \Utl(B | h)$, $g$ must be
            symmetric, and non-negative, while (\ref{dp:ir}) guarantees that
            $\phi_t(g(\xi, v_t))$ is valid, since $\phi_t$ is not defined on
            negative inputs.
    \end{itemize}
    In summary, by Theorem \ref{thm:corebam}, $B^{g, y}$ is a core BAM (with
    assuming $g(h) = \Utl(B | h)$), if and only if $\forall t \in [T]$, $(g, y)$
    is feasible in program (\ref{dp:obj}).
  \end{proof}

  Particularly, for one item per stage case, $\forall t \in [T],~\dV_t=\Real_+$,
  the buyer has discrete stage type distribution for each stage, i.e.,
  \begin{align*}
    \forall t \in [T], \exists v^1_t < \cdots < v^{k_t}_t \in \dV_t,~s.t.~
    \sum_{i = 1}^{k_t}\Pr_{v \sim \F_t}[v = v^i_t] = 1.
  \end{align*}
  Then program (\ref{dp:obj}) reduces to the following forms, i.e.,
  \begin{align*}
    \mathrm{maximize}   \quad & \E_{v_t}\Big[
      y_t(\xi, v_t) \cdot v_t + \phi_t\big(g(\xi, v_t)\big)\Big]  \\
    \mathrm{subject~to} \quad & y_t(\xi, v_t) = \int_0^{v_t} w_t(\xi, v) \ud v
      \\
    & \int_{\dV_t} w_t(\xi, v) \ud v = 1  \\
    & g(\xi, v_t) = \xi + \hu_t(\xi, v_t; v_t)
        - \E_{v'_t}\big[\hu_t(\xi, v'_t; v'_t)\big]  \\
    & g(\xi, v_t) \geq 0, w_t(\xi, v_t) \geq 0,~\forall v_t \in \dV_t.
  \end{align*}
  \begin{lemma}\label{lem:icallo}
    For one item per stage case, the allocation rule must be a randomization of
    posted-price mechanisms.

    Particularly, for discrete valuation distribution cases, it is WLOG to
    restrict to the convex combinations of posted-price mechanisms with the
    prices only from the set of discrete types, $v^1_t, \ldots, v^{k_t}_t$.
  \end{lemma}

  While for discrete valuation distributions, WLOG, suppose that $v_t^1 = 0$,
  then the program further reduces to the following form, where $\Pr^i_t$
  denotes the probability of type $v^i_t$ on stage $t$, i.e., $\Pr^i_t =
  \Pr_{v \sim \F_t}[v = v^i_t]$.
  \begin{align*}
    \mathrm{maximize} \quad & \sum_{i = 1}^{k_t}
        {\Pr}^i_t \cdot \Big(y^i_t \cdot v^i_t + \phi_t\big(g^i_t\big)\Big)  \\
    \mathrm{subject~to} \quad & y^i_t = \sum_{j = 1}^{i} w^j_t   \\
    & \sum_{i = 1}^{k_t} w^i_t = 1  \\
    & \hu^i_t = \sum_{j = 1}^i w^j_t (v^i_t - v^j_t)  \\
    & g^i_t = \xi + \hu^i_t - \sum_{j = 1}^{k_t} {\Pr}^i_t \cdot \hu^j_t  \\
    & g^i_t \geq 0,w^i_t \geq 0,~\forall i = 1, \ldots, k_t.
  \end{align*}
  Denote the value of the program by $\val(\xi, \phi_t, w_t)$, and the optimal
  value by $\Prog(\xi, \phi_t) = \max_{w_t} \val(\xi, \phi_t, w_t)$, where
  $\xi$ and $\phi_t$ are parameters ($\phi_t$ is a function) of the program,
  and $w_t = (w^1_t, \ldots, w^{k_t}_t)$ are the variables (others are all
  determined by $w_t$). Note that $w_t$ takes different values for different
  parameters.

  This program is a linear program, if $\phi_t$ is a concave and piece-wise
  linear function. Moreover, the LP is of polynomial size, if $\phi_t$ has at
  most polynomially many pieces.

  Then the following lemma completes our proof that for any $\epsilon > 0$, the
  LP-based algorithm achieves an $\epsilon$-approximation by solving at most
  $O(TN / \epsilon)$ many LPs, and each LP of size at most $O(k_t N/\epsilon)$.
  \begin{lemma}\label{lem:hprop}
    $\phi_t(\xi)$ is concave, and can be $\kappa$-approximated by two concave
    piece-wise linear functions, $\up_t(\xi)$ and $\op_t(\xi)$. Formally,
    \begin{align*}
      \forall \xi, \up_t(\xi) \leq \phi_t(\xi) \leq \op_t(\xi),~
        \op_t(\xi) - \up_t(\xi) \leq \kappa \max_\zeta \phi_t(\zeta),
    \end{align*}
    where the number of pieces of $\up_t$ and $\op_t$ is bounded by
    $O(N + 1 / \kappa)$.
  \end{lemma}
  \begin{proof}[Proof of Lemma \ref{lem:hprop}]
    We first show that $\phi_t(\xi)$ is concave. Let $B$ and $B'$ be the
    sub-BAMs from stage $t + 1$ to $T$ that reach $\phi_t(\xi)$ and
    $\phi_t(\xi')$ respectively. Note that $\E_{v_{(t + 1, T)}}[\bal_{T + 1}] =
    \E_{v_{(t + 1, T)}}[\Utl(B | v_{(1, T)})] = \Utl(B | v_{(1, t)}) = \xi$.
    Then for any $\eta \in [0, 1]$, $\xi'' = \eta \xi + (1 - \eta)\xi'$,
    \begin{align*}
      \phi_t(\xi'') &\geq \E_{v(t + 1, T)}\bigg[\sum_{\tau = t + 1}^T
                            y''_t \cdot v_\tau - \bal''_{T + 1}\bigg]  \\
                    &= \E_{v(t + 1, T)}\bigg[\sum_{\tau = t + 1}^T
                            \big(\eta y_t + (1 - \eta) y'_t\big) \cdot
                                v_\tau\bigg] - \xi''  \\
                    &= \eta \phi_t(\xi) + (1 - \eta) \phi_t(\xi'),
    \end{align*}
    where $y''_t$ is a convex combination of $y_t$ and $y'_t$. Note that in
    general, the convex combination of two BAMs ($B$ and $B'$) is not a BAM,
    because it might not be symmetric. However, by Theorem \ref{thm:rep}, there
    is a BAM that admits the same expected utility, $\xi''$, and better
    expected efficiency. In other words, $y''_t$ might indicate a mechanism
    that is not a BAM, but there is a BAM that generates higher revenue than
    this mechanism. Hence $\phi_t(\xi'')$ is even larger.

    Then by the following lemma, we could efficiently compute $\up_{t - 1}$ and
    $\op_{t - 1}$ from given $\up_t$ and $\op_t$. Hence we get an LP-based
    algorithm to achieve an arbitrarily close approximation to the optimal
    mechanism.
    \begin{lemma}\label{lem:concaveapprox}
      For any concave function $\varphi$ defined on interval $[a, b]$, if
      \begin{itemize}
        \item $\varphi(a)$ and $\varphi(b)$ are given;
        \item there exist $|\beta_a| \leq +\infty$ and $|\beta_b| \leq +\infty$
              such that
              \begin{align*}
                \varphi(\xi) \leq \beta_a(\xi - a) + \varphi(a),~
                \varphi(\xi) \leq \beta_b(\xi - b) + \varphi(b);
              \end{align*}
      \end{itemize}
      then for any $\kappa > 0$, a pair of lower and upper bounds, $\uvp$ and
      $\ovp$ can be computed via $O(n)$ queries to the evaluation oracle of
      $\varphi$, such that
      \begin{itemize}
        \item let $\beta = (\varphi(b) - \varphi(a)) / (b - a)$, then
              \begin{align*}
                n \leq \frac4\kappa + \log\frac{(\beta_a - \beta_b)^2}
                                      {(\beta_a - \beta)(\beta_b - \beta)};
              \end{align*}
        \item both of $\uvp$ and $\ovp$ are concave and piece-wise linear and
              have at most $O(n)$ pieces;
        \item the gap between $\uvp$ and $\ovp$ is $\delta = \big(\max_\zeta
              \varphi(\zeta) - \min\{\varphi(a), \varphi(b)\}\big)\kappa$.
      \end{itemize}
    \end{lemma}
    Suppose we have lower and upper bounds for $\phi_t$, i.e., $\up_t \leq
    \phi_t \leq \op_t$. Then the functions $\up^*_{t - 1}(\xi) = \Prog(\xi,
    \up_t)$ and $\op^*_{t - 1}(\xi) = \Prog(\xi, \op_t)$ are lower and upper
    bounds for $\phi_{t - 1}$, respectively, i.e.,
    \begin{align*}
      \Prog(\xi, \up_t) &= \val(\xi, \up_t, \underline{w}^*_t)
        \leq \val(\xi, \phi_t, \underline{w}^*_t)
        \leq \val(\xi, \phi_t, w^*_t) \leq \Prog(\xi, \phi_t)
        = \phi_{t - 1}(\xi),  \\
      \Prog(\xi, \op_t) &= \val(\xi, \op_t, \overline{w}^*_t)
        \geq \val(\xi, \op_t, w^*_t) \geq \val(\xi, \phi_t, w^*_t)
        \geq \Prog(\xi, \phi_t)
        = \phi_{t - 1}(\xi).
    \end{align*}
    Let $\up_{t - 1}$ be the lower bound of $\up^*_{t - 1}$ and $\op_{t - 1}$
    be the upper bound of $\op^*_{t - 1}$. So $\up_{t - 1} \leq \phi_{t - 1}
    \leq \op_{t - 1}$, and the gap is bounded as follows,
    \begin{align*}
      \op_{t - 1} - \up_{t - 1} &= (\op_{t - 1} - \op^*_{t - 1}) +
          (\op^*_{t - 1} - \up^*_{t - 1}) + (\up^*_{t - 1} - \up_{t - 1})  \\
        &\leq \kappa + 2(T - t)\kappa + \kappa = 2(T - t + 1)\kappa,
    \end{align*}
    where
    \begin{align*}
      \op^*_{t - 1} - \up^*_{t - 1}
      & = \val(\xi, \op_{t - 1}, \overline{w}^*_{t - 1})
          - \val(\xi, \up_{t - 1}, \underline{w}^*_{t - 1})  \\
      & \leq \val(\xi, \op_{t - 1}, \overline{w}^*_{t - 1})
          - \val(\xi, \up_{t - 1}, \overline{w}^*_{t - 1})  \\
      & = \sum_{i = 1}^{k_{t - 1}} {\Pr}^i_{t - 1} \cdot \Big(
            \op_{t - 1}(g^i_{t - 1}(\overline{w}^*_{t - 1}))
              - \up_{t - 1}(g^i_{t - 1}(\overline{w}^*_{t - 1}))\Big)  \\
      & \leq \max_{\zeta \geq 0}
                \big(\op_{t - 1}(\zeta) - \up_{t - 1}(\zeta)\big)
        \leq 2(T - t)\kappa.
    \end{align*}
    Hence the gap between $\up_0$ and $\op_0$ is bounded by $T\kappa$.

    Let $N$ be the input size of the problem. Then we prove that
    \begin{itemize}
      \item Lemma \ref{lem:concaveapprox} can be used to approximate
            $\up^*_{t - 1}$ and $\op^*_{t - 1}$ with $\up_t$ and $\op_t$ being
            given;
      \item the multiplicative $\epsilon$-approximation can be achieved by
            letting $\kappa = 4\epsilon / T$;
      \item the entire algorithm needs to solve $O\big((N + T/\epsilon)T\big)$
            many LPs, each of size at most $O\big((T / \epsilon + N)k_t\big)$.
    \end{itemize}

    We verify that $\phi_t(\xi)$ satisfies the requirements for Lemma
    \ref{lem:concaveapprox}. Consider the boundary estimation of $\phi_t(\xi)$.
    Note that by assumption, $v^1_t = 0$, hence $\hu^1_t = 0$. Then
    \begin{align*}
      g^1_t \geq 0 \Infer \sum_{j = 1}^{k_t} {\Pr}^i_t \cdot \hu^j_t \leq \xi.
    \end{align*}
    Meanwhile, for $i < k_t$,
    \begin{align*}
      & y^i_t = \sum_{j = 1}^i w^j_t
        \leq \max_{i' \leq i} \frac1{v^{i + 1}_t - v^{i'}_t} \cdot
              \sum_{j = 1}^i (v^{i + 1}_t - v^j_t)w^j_t
        \leq \frac1{v^{i + 1}_t - v^i_t} \cdot \hu^{i + 1}_t \\
      \Infer & \sum_{i = 1}^{k_t}{\Pr}^i_t \cdot y^i_t \cdot v^i_t
      \leq {\Pr}^{k_t}_t \cdot y^{k_t}_t \cdot v^{k_t}_t
            + \sum_{i = 1}^{k_t - 1} {\Pr}^i_t \cdot \hu^{i + 1}_t \cdot
              \frac{v^{i + 1}_t}{v^{i + 1}_t - v^i_t}  \\
      & \quad \leq {\Pr}^{k_t}_t \cdot v^{k_t}_t + \max_{i < k_t}
            \frac{{\Pr}^i_t v^{i + 1}_t}{{\Pr}^{i + 1}_t(v^{i + 1}_t - v^i_t)}
            \cdot \sum_{i = 1}^{k_t} {\Pr}^i_t \cdot \hu^i_t
      \leq {\Pr}^{k_t}_t \cdot v^{k_t}_t + c_t \cdot \xi,
    \end{align*}
    where $c_t = \max_{i < k_t} ({\Pr}^i_t v^i_t) /
    ((v^i_t - v^{i - 1}_t){\Pr}^{i + 1}_t) < +\infty$. Combining with the
    concavity of $\phi_{t + 1}$, we conclude that
    \begin{align}\label{eq:philupb}
      \phi_t(\xi) \leq {\Pr}^{k_t}_t \cdot v^{k_t}_t
                        + c_t \cdot \xi + \phi_{t + 1}(\xi) \leq \cdots
      \leq \xi\sum_{\tau = t}^T c_\tau
           + \sum_{\tau = t}^T {\Pr}^{k_\tau}_\tau \cdot v^{k_\tau}_\tau.
    \end{align}
    Meanwhile, when $\xi = 0$, $\hu^i_t \leq 0 \Infer w^1_t = \cdots =
    w^{k_t - 1}_t = 0,~w^{k_t}_t = 1$. Hence
    \begin{align}\label{eq:philval}
      \phi_t(0) = {\Pr}^{k_t}_t \cdot v^{k_t}_t + \phi_{t + 1}(0) = \cdots
      = \sum_{\tau = t}^T {\Pr}^{k_\tau}_\tau \cdot v^{k_\tau}_\tau,
    \end{align}
    which meets the upper bound (\ref{eq:philupb}) of $\phi_t$ at $\xi = 0$.

    By letting $y^1_t = \cdots = y^{k_t}_t = 1$, we know that when $\xi \geq
    \sum_{\tau = t}^T \Val_\tau$, $g^i_t \geq \sum_{\tau = t + 1}^T \Val_\tau$.
    Then by backward induction from $T$, for large enough $\xi \geq
    \sum_{\tau = t}^T \Val_\tau$, $\phi_t(\xi) = \sum_{\tau = t}^T \Val_\tau -
    \xi$. Since $\phi_t$ is concave, it is further upper bounded as follows,
    \begin{align}\label{eq:phirupb}
      \phi_t(\xi) \leq \sum_{\tau = t}^T \Val_\tau - \xi.
    \end{align}

    In summary, $\phi_t$ is a concave function on interval $\big[0,
    \sum_{\tau = t}^T \Val_\tau\big]$, and the function values on two ends are
    $\sum_{\tau = t}^T {\Pr}^{k_\tau}_\tau \cdot v^{k_\tau}_\tau$ and $0$,
    respectively. Moreover, it also has two upper bounds,
    (\ref{eq:philupb})(\ref{eq:phirupb}), that go through the two ends
    respectively. Hence, $\beta_a = \sum_{\tau = t}^T$, and $\beta_b = -1$. So
    \begin{align*}
      \log\frac{(\beta_a - \beta_b)^2}{(\beta_a - \beta)(\beta_b - \beta)}
        = O(N).
    \end{align*}

    Then by Lemma \ref{lem:concaveapprox}, $\phi_t(\xi)$ can be
    $2(T - t)\kappa$-approximated by $\up_t(\xi)$ and $\op_t(\xi)$ with at most
    $O(\frac1{\kappa} + N)$ many queries to $\up_{t + 1}(\xi)$ and
    $\op_{t + 1}(\xi)$.
  \end{proof}
  Therefore, there is an LP-based FPTAS via dynamic program to achieve a
  multiplicative $\epsilon$-approximation.
\end{proof}

\subsection{Proof of Lemma \ref{lem:icallo}}\label{app:lem:icallo}
\begin{proof}[Proof of Lemma \ref{lem:icallo}]
  Note that for one item per stage case, the mechanism is truthful within any
  stage $t$, if and only if $\hu_t(\xi, v_t; v_t)$ is increasing and convex in
  $v_t$, and $z_t(\xi, v_t)$ is the sub-gradient of $\hu_t$ with respective to
  $v_t$. It is equivalent to $z_t(\xi, v_t)$ being non-negative and increasing
  in $v_t$, and $\hu_t$ being the integration of $z_t(\xi, v_t)$ over $v_t$.

  $z_t(\xi, v_t)$ is non-negative and increasing in $v_t$, if and only if
  $z_t(\xi, v_t)$ is a convex combination of step functions, i.e.,
  \begin{align*}
    z_t(\xi, v_t) = \int_0^{+\infty} w_t(\xi, v) \cdot \I[v_t \geq v] \ud v
                  = \int_0^{v_t} w_t(\xi, v) \ud v,
  \end{align*}
  where $\int_0^{+\infty}w_t(\xi, v) \ud v = 1$ and $w_t(\xi, v) \geq 0$.

  Particularly, for discrete types, to compute the optimal mechanism, we could
  assume that $w_t(\xi, v)$ is non-zero only at points $v^1_t, \ldots,
  v^{k_t}_t$.

  To prove this, for any $v^*$ such that $v^i_t < v^* < v^{i + 1}_t$, we show
  that if $w_t(\xi, v^*) > 0$, we can weakly improve the objective value without
  either breaking any constraint, or changing the values of $g(\xi, v^i_t)$.

  Let $w'_t(\xi, v) = w_t(\xi, v)$ except for three points, i.e.,
  \begin{align*}
    w'_t(\xi, v^i_t)       &= w_t(\xi, v^i_t) + \frac{v^{i + 1}_t - v^*}
                              {v^{i + 1}_t - v^i_t}w_t(\xi, v^*),  \\
    w'_t(\xi, v^{i + 1}_t) &= w_t(\xi, v^i_t) + \frac{v^* - v^i_t}
                              {v^{i + 1}_t - v^i_t}w_t(\xi, v^*),  \\
    w'_t(\xi, v^*)         &= 0.
  \end{align*}
  Then it is straightforward to verify the following facts.
  \begin{align*}
    & w'_t(\xi, v) \geq 0,  \\
    & \int_0^{+\infty} w'_t(\xi, v) \ud v
        = \int_0^{+\infty} w_t(\xi, v) \ud v = 1,  \\
    & y'_t(\xi, v^j_t) = \int_0^{v^j_t} w'_t(\xi, v) \ud v
                       = \int_0^{v^j_t} w_t(\xi, v) \ud v
                       = y_t(\xi, v^j_t), \forall j \neq i  \\
    & y'_t(\xi, v^i_t) = \int_0^{v^i_t} w'_t(\xi, v) \ud v
                       = \int_0^{v^i_t} w_t(\xi, v) \ud v
                         + \frac{v^{i + 1}_t - v^*}{v^{i + 1}_t - v^i_t}
                           w_t(\xi, v^*)
                       > y_t(\xi, v^i_t),  \\
    & \hu'_t(\xi, v^j_t; v^j_t)
                       = \int_0^{v^j_t} y'_t(\xi, v) \ud v
                       = \int_0^{v^j_t} y_t(\xi, v) \ud v
                       = \hu_t(\xi, v^j_t; v^j_t), \forall j \leq i \\
    & \hu'_t(\xi, v^j_t; v^j_t)
                       = \int_0^{v^j_t} y'_t(\xi, v) \ud v
                       = \int_0^{v^j_t} y_t(\xi, v) \ud v
                         + \int_{v^i_t}^{v^{i + 1}_t}
                            \Big(y'_t(\xi, v) - y_t(\xi, v)\Big) \ud v  \\
    & \qquad \qquad    = \hu_t(\xi, v^j_t; v^j_t), \forall j \geq i + 1  \\
    \Infer &
      g'(\xi, v^j_t) = g(\xi, v^j_t), \forall j  \\
  \end{align*}
  where
  \begin{align*}
     & \int_{v^i_t}^{v^{i + 1}_t} \Big(y'_t(\xi, v) - y_t(\xi, v)\Big) \ud v
       = \int_{v^i_t}^{v^{i + 1}_t} \int_0^v
          \Big(w'_t(\xi, \nu) - w_t(\xi, \nu)\Big) \ud \nu \ud v  \\
    =& \int_{v^i_t}^{v^{i + 1}_t} \int_{v^i_t}^v
          \Big(w'_t(\xi, \nu) - w_t(\xi, \nu)\Big) \ud \nu \ud v
       = \int_{v^i_t}^{v^{i + 1}_t} \int_{\nu}^{v^{i + 1}_t}
          \Big(w'_t(\xi, \nu) - w_t(\xi, \nu)\Big) \ud v \ud \nu   \\
    =& \big(v^{i + 1}_t - v^i_t\big) \cdot
         \big(w'_t(\xi, v^i_t) - w_t(\xi, v^i_t)\big) +
         \big(v^{i + 1}_t - v^*\big) \cdot
         \big(w'_t(\xi, v^*) - w_t(\xi, v^*)\big)  \\
     &+\big(v^{i + 1}_t - v^{i + 1}_t\big) \cdot
         \big(w'_t(\xi, v^{i + 1}_t) - w_t(\xi, v^{i + 1}_t)\big) \\
    =& \big(v^{i + 1}_t - v^i_t\big) \cdot
          \frac{v^{i + 1}_t - v^*}{v^{i + 1}_t - v^i_t}w_t(\xi, v^*)
        + \big(v^{i + 1}_t - v^*\big) \cdot
          \big(-w_t(\xi, v^*)\big)
        + 0 = 0.
  \end{align*}
  As desired, no constraint gets violated, but objective value gets improved.
\end{proof}

\subsection{Proof of Lemma \ref{lem:concaveapprox}}
\label{app:lem:concaveapprox}
\begin{proof}[Proof of Lemma \ref{lem:concaveapprox}]
  Since $\varphi$ is concave, hence $\beta_a \geq \beta_b$. If they coincide,
  i.e., $\beta_a = \beta_b$, then we are done, and $\varphi$ must be the line
  segment that connects $(a, \varphi(a))$ and $(b, \varphi(b))$.

  Otherwise, $\beta_a - \beta_b > 0$, and since $|\beta_a|, |\beta_b| <
  +\infty$, $\beta_a - \beta_b < +\infty$.

  By adding a linear offset to $\varphi$, we assume WLOG that $\varphi(a) =
  \varphi(b) = 0$. Hence $\beta_a > 0 > \beta_b$. Note that $\max_\xi
  \varphi(\xi)\min\{\varphi(a), \varphi(b)\}$ is not increased after adding the
  linear offset.

  Then for any $\delta > 0$, we make $n$ queries to the evaluation oracle of
  $\varphi$ at $l_1, \ldots, l_m, l_{m + 1}, \ldots, l_n$, where $l_m$ is the
  last point such that $\varphi(l_{m - 1}) \leq \varphi(l_m) >
  \varphi(l_{m + 1})$. Since $\varphi$ is concave, such $m$ is unique.

  Then points $l_1, \ldots, l_m$ are adaptively chosen according to the
  following equations, until $\varphi$ starts to decrease.
  \begin{align*}
    & l_0 = a, l_i = \min\{l_{i - 1} + \delta / \beta_{i - 1}, b\},  \\
    & \beta_0 = \beta_a,~
      \beta_i = \frac{\varphi(l_i) - \varphi(l_{i - 1})}{l_i - l_{i - 1}}
                - \frac{\varphi(l_i)}{l_i - b}.
  \end{align*}
  Since $\varphi$ is concave, $\beta_i \geq 0$, and $\beta_i \leq \beta_a -
  \beta_b < +\infty$. Hence $l_i > l_{i - 1}$, as long as $l_{i - 1} < b$. Thus
  $m < +\infty$.

  For the last half of points, $l_{m + 1}, \ldots, l_n$, are chosen similarly
  but in the reverse order, starting from $l_{n + 1} = b$ until the next point
  is less than $l_m$.

  In the remainder of the proof, we show that $m \leq 2\max_\xi\varphi(\xi) /
  \delta + \log(\beta_a - \beta_b)/(-\beta_b)$, which then implies
  that $n \leq 4\max_\xi\varphi(\xi) / \delta +
  \log(\beta_a - \beta_b)^2/(-\beta_a\beta_b)$.

  Note that to achieve a given multiplicative approximation ratio $\kappa$,
  we can let $\delta = \kappa \cdot \varphi(\frac{a + b}2)$, because
  $2\varphi(\frac{a + b}2) \geq \max_{\xi} \varphi(\xi) \geq
  \varphi(\frac{a + b}2)$.

  Then by connecting points $(a, \varphi(a)), (l_1, \varphi(l_1)), \ldots,
  (l_n, \varphi(l_n)), (b, \varphi(b))$, we get a lower bound $\uvp$ for
  $\varphi$, which is concave and piece-wise linear, and has at most $n + 1$
  pieces. The upper bound is simply $\ovp(\xi) = \uvp(\xi) + \delta$.

  From now on, we only consider $0 \leq i \leq m$. For any $l_i \leq l \leq
  l_{i + 1}$, $\varphi(l)$ is $\delta$-approximated by the line segment that
  connects $(l_i, \varphi(l_i))$ and $(l_{i + 1}, \varphi(l_{i + 1}))$, i.e.,
  by the concavity of $\varphi$,
  \begin{align*}
    \varphi(l) \geq \frac{\varphi(l_i)(b - l)}{b - l_i},\quad
    \varphi(l) \leq \frac{\varphi(l_{i - 1})(l - l_i)}{l_{i - 1} - l_i}
                     + \frac{\varphi(l_i)(l_{i - 1} - l)}{l_{i - 1} - l_i}.
  \end{align*}
  The gap between the upper and lower bounds of $\varphi(l)$ is,
  \begin{align*}
    & \frac{\varphi(l_{i - 1})(l - l_i)}{l_{i - 1} - l_i}
    + \frac{\varphi(l_i)(l_{i - 1} - l)}{l_{i - 1} - l_i}
    - \frac{\varphi(l_i)(b - l)}{b - l_i}
    \\
    &= \bigg(\frac{\varphi(l_i) - \varphi(l_{i - 1})}{l_i - l_{i - 1}}
            - \frac{\varphi(l_i)}{l_i - b}\bigg)l +
       \bigg(\frac{\varphi(l_i)l_{i - 1} - \varphi(l_{i - 1})l_i}
                  {l_{i - 1} - l_i}
             - \frac{\varphi(l_i)b}{b - l_i}\bigg) \\
    &= \beta_i l + \bigg(\frac{\varphi(l_i)l_{i - 1} - \varphi(l_{i - 1})l_i}
                              {l_{i - 1} - l_i}
                         - \frac{\varphi(l_i)b}{b - l_i}\bigg)
    \\
    &\leq \beta_i (l_i + \delta / \beta_i) + \bigg(\frac{\varphi(l_i)l_{i - 1}
          - \varphi(l_{i - 1})l_i}{l_{i - 1} - l_i}
          - \frac{\varphi(l_i)b}{b - l_i}\bigg)  \\
    &= \delta + \varphi(l_i) - \varphi(l_i) = \delta.
  \end{align*}

  Then, we show that $m$ is small, so the $\delta$-approximation is achieved
  by a few of queries. Let $\beta^*_i$ be defined as follows,
  \begin{align*}
    \beta^*_i = \beta_i + \frac{\varphi(l_i)}{l_i - b} - \beta_b
      \geq \beta_i + \frac{\beta_b(l_i - b)}{l_i - b} - \beta_b = \beta_i.
  \end{align*}
  Note that $\varphi(l_m) > 0 = \varphi(b)$, hence $l_m < b$, then for
  $i \leq m - 1$,
  \begin{align*}
    & \beta_i = \frac{\varphi(l_i) - \varphi(l_{i - 1})}{l_i - l_{i - 1}}
                - \frac{\varphi(l_i)}{l_i - b}
              \geq \frac{\varphi(l_i) - \varphi(l_{i - 1})}{l_i - l_{i - 1}}
              = \beta^*_i - \beta^*_{i + 1},  \\
    & \beta^*_i(l_{i + 1} - l_i) \geq \beta_i(l_{i + 1} - l_i)
      = \beta_i \cdot \frac{\delta}{\beta_i} = \delta,  \\
    \Infer~&
      \delta - (\beta^*_i - \beta^*_{i + 1} + \beta_b)(l_{i + 1} - l_i) \leq
      (\beta^*_{i + 1} - \beta_b)(l_{i + 1} - l_i)
      = \varphi(l_{i + 1}) - \varphi(l_i),  \\
    \Infer~& \delta \leq
      (\beta^*_i - \beta^*_{i + 1} + \beta_b)(l_{i + 1} - l_i)
      + \varphi(l_{i + 1}) - \varphi(l_i)  \\
    &~~= \frac{(\beta^*_i - \beta^*_{i + 1} + \beta_b)\delta}{\beta_i}
            + \varphi(l_{i + 1}) - \varphi(l_i),  \\
    \Infer~& m\delta \leq
      \sum_{i = 0}^{m - 1} \bigg(
            \frac{(\beta^*_i - \beta^*_{i + 1} + \beta_b)\delta}{\beta_i}
            + \varphi(l_{i + 1}) - \varphi(l_i)\bigg)  \\
    &~~=    \varphi(l_m) + \delta \sum_{i = 0}^{m - 1}
            \frac{\beta^*_i - \beta^*_{i + 1} + \beta_b}{\beta_i}  \\
    &~~\leq \max_{\xi} \varphi(\xi) + \frac{(m + k)\delta}2  \\
    \Infer~& m \leq \frac{2\max_{\xi} \varphi(\xi)}\delta + k,
  \end{align*}
  where there are at most $k$ different $i$'s such that
  $(\beta^*_i - \beta^*_{i + 1} + \beta_b) / \beta_i \geq 1/2$. Because
  \begin{align*}
    \frac{\beta^*_i - \beta^*_{i + 1} + \beta_b}{\beta_i} \geq \frac12
    \Infer \beta^*_i \geq (\beta_i - \beta^*_i) + 2\beta^*_{i + 1} - 2\beta_b
           \geq 2\beta^*_{i + 1}.
  \end{align*}
  Since for any $i \leq m - 1$, $\beta^*_i \leq \beta_a - \beta_b$,
  $\beta^*_{i + 1} \geq - \beta_b$, and $\beta^*_i \geq \beta^*_{i + 1}$, $k$
  must be no more than $\log (\beta_a - \beta_b) / (-\beta_b)$.

  Combining with the last half of points, we have
  \begin{align*}
    n \leq \frac{4}{\delta}\max_{\xi}\varphi(\xi)
      + \log \frac{(\beta_a - \beta_b)^2}{-\beta_a\beta_b}.
  \end{align*}
\end{proof}

% Acknowledgments
% \begin{acks}
%   Acknowledgments.
% \end{acks}

\end{document}